\documentclass[11pt]{article}
\pdfoutput=1
\usepackage[margin=1in]{geometry} 

\usepackage{etoolbox}

\usepackage[round]{natbib}
\usepackage{array,amsmath,amsthm,amssymb,amsfonts,titling,bbm, titlesec, bm, physics, enumitem, accents, setspace, fancyhdr, natbib, geometry, pdflscape}
\usepackage{graphicx}
\usepackage{float}
\usepackage[font=footnotesize,labelfont=bf]{caption}
\usepackage{array}
\usepackage[title]{appendix}
\usepackage{afterpage}
\usepackage{listings}
\usepackage{prodint}
\usepackage{algorithm}
\usepackage{comment}
 \usepackage{booktabs}
\usepackage{multirow,booktabs}
\usepackage{subcaption}
\newlength{\continueindent}
\usepackage{etoolbox}
\usepackage{setspace}
\usepackage[dvipsnames]{xcolor}
\definecolor{Bleu}{RGB}{0,0,204}
\usepackage[hypertexnames=false]{hyperref}
\hypersetup{
colorlinks,
  citecolor=Bleu,
  linkcolor=Bleu,
  urlcolor=Bleu}
\newcommand{\INPUT}{\item[\textbf{Input:}]}

\newcolumntype{P}[1]{>{\centering\arraybackslash}p{#1}}

\usepackage{amsmath}
\usepackage{amssymb}
\usepackage{amsthm,thmtools}
\usepackage{enumitem}
\usepackage{amsfonts}
\RequirePackage{algorithm}
\RequirePackage{algorithmic}
\usepackage{bm,upgreek}
\usepackage{mathrsfs}  
\usepackage{calc}
\usepackage{subcaption}
\usepackage{authblk}

\theoremstyle{plain}
\theoremstyle{definition} 
\newtheorem{example}{Example} 
\renewcommand{\qedsymbol}{$\square$}
\newtheorem{theorem}{Theorem}
\newtheorem{corollary}{Corollary}

\newtheorem{lemma}{Lemma}

\newtheorem*{remark*}{Remark}

\providecommand{\keywords}[1]
{
  \small	
  \textbf{\textit{Keywords:}} #1
}
\usepackage[utf8]{inputenc}
\usepackage{textcomp}

\DeclareMathOperator*{\argmin}{argmin}

\DeclareMathOperator*{\esssup}{ess\,sup}

\usepackage[capitalize,noabbrev]{cleveref}
\DeclareMathOperator*{\logit}{logit}
\DeclareMathOperator*{\expit}{expit}
\usepackage{authblk}

\usepackage[most]{tcolorbox}

\DeclareFontFamily{U}{jkpmia}{}
\DeclareFontShape{U}{jkpmia}{m}{it}{<->s*jkpmia}{}
\DeclareFontShape{U}{jkpmia}{bx}{it}{<->s*jkpbmia}{}
\DeclareMathAlphabet{\mathfrak}{U}{jkpmia}{m}{it}
\SetMathAlphabet{\mathfrak}{bold}{U}{jkpmia}{bx}{it}

\usepackage[hang,flushmargin]{footmisc}

\usepackage{tikz}
\usetikzlibrary{positioning, arrows.meta, shapes.geometric}
\usetikzlibrary{calc}

\title{Automatic Debiased Machine Learning for Smooth Functionals\\ of Nonparametric M-Estimands}

\author[1,2]{Lars van der Laan\thanks{\footnotesize Corresponding author: lvdlaan@uw.edu}}
\author[2]{Aur\'elien Bibaut}
\author[2,3]{Nathan Kallus}
\author[1]{Alex Luedtke}

\affil[1]{Department of Statistics, University of Washington}
\affil[2]{Netflix Research}
\affil[3]{Cornell Tech, Cornell University}

 \usepackage{setspace}
\newcommand{\slightspacing}{\setstretch{1.175}}

\begin{document}

\maketitle

\slightspacing

\begin{abstract}
\singlespacing
We develop a unified framework for \emph{automatic debiased machine learning} (autoDML) for inference on a broad class of statistical parameters. The framework applies to any smooth functional of a nonparametric \emph{M-estimand}, defined as the minimizer of a population risk over an infinite-dimensional linear space. Examples include counterfactual regression, quantile, and survival functions, as well as conditional average treatment effects. Rather than requiring manual derivation of influence functions, our approach automates the construction of debiased estimators using three ingredients: the gradient and Hessian of the loss function and a linear approximation of the target functional. Estimation reduces to solving two risk minimization problems, one for the M-estimand and one for a Riesz representer. The framework accommodates Neyman-orthogonal loss functions that depend on nuisance parameters and extends to vector-valued M-estimands through joint risk minimization. We characterize the efficient influence function and construct efficient autoDML estimators via one-step correction, targeted minimum loss estimation, and sieve-based plug-in methods. Under quadratic risk, these estimators satisfy double robustness for linear functionals. We further show that they are robust to mild misspecification of the M-estimand model, incurring only second-order bias. We illustrate the method by estimating long-term survival probabilities under a semiparametric two-parameter beta-geometric failure model.
\end{abstract}


\section{Introduction}

\slightspacing

Inference on real-valued summaries of probability distributions is central to many scientific applications, including treatment effect estimation, survival analysis, and policy learning. These summaries are typically low-dimensional functionals of high- or infinite-dimensional nuisance functions \citep{vanderLaanRose2011, DoubleML}. A key example is the average treatment effect (ATE) under full confounding adjustment, which is the difference of partial averages of the outcome regression function \citep{bang2005doubly}. More broadly, such nuisance functions are often characterized as minimizers of a population risk, such as mean squared error in regression. It is appealing to use machine learning methods to estimate this function, given their flexibility and practical success. However, simply reporting the functional applied to the machine learning estimate may result in bias and slow convergence, and leaves the question of inference unresolved.

Several methods have been developed to ``de-bias" the na\"ive plug-in estimator to achieve \(\sqrt{n}\)-consistency and asymptotic normality under suitable conditions. These methods include one-step estimation \citep{pfanzagl1985contributions, bickel1993efficient}, estimating equations and double machine learning \citep{robins1995analysis, robinsCausal, vanderlaanunified, DoubleML}, and targeted minimum loss-based estimation \citep{laan_rubin_2006, vanderLaanRose2011}. For certain sieve estimators, plug-in methods may suffice, though inference still requires related asymptotic analysis \citep{shen1997methods, spnpsieve, sieveOneStepPlugin, SieveQiu, discussionMLEMark, undersmoothedHAL}. These frameworks typically involve two stages: (i) estimating nuisance functions using flexible machine learning tools, such as regularized empirical risk minimization, and (ii) a debiasing step to enable valid inference.  The latter often requires manual derivation of the efficient influence function for the target parameter, along with additional nuisance estimation tailored to that parameter \citep{bickel1993efficient}. This derivation is often complex and parameter-specific, requiring specialized expertise and case-by-case analysis.

A growing body of work has sought to automate debiased estimation and inference in semiparametric statistics and causal inference. These methods aim to simplify the construction of debiased machine learning estimators for novel parameters, thereby avoiding complex derivations based on differential geometry and functional analysis. One line of work uses numerical differentiation to approximate the efficient influence function for pathwise differentiable functionals \citep{frangakis2015deductive, luedtke2015discussion, carone2019toward, ichimura2022influence, jordan2022empirical}. However, these methods require careful tuning and can be computationally demanding, since they involve perturbing the data distribution and numerically differentiating the parameter for each observation, which may limit their practical applicability. More recently, \citet{luedtke2024simplifying} proposed a framework based on automatic differentiation of pathwise differentiable functionals \citep{rall1981automatic}. This framework applies to parameters that can be expressed as compositions of primitive functions with known (precomputed) Hilbert-valued pathwise derivatives, such as regression functions and densities \citep{luedtke2024one}.

A widely used approach for regression functionals is \textit{automatic debiased machine learning} (autoDML) \citep{chernozhukov2022automatic}. The key insight is that, for a smooth functional of a regression function, constructing the efficient influence function and debiasing an estimator requires only one additional nuisance parameter: the Riesz representer of the linearized functional \citep{williams2025riesz, hines2025learning}. Because this representer is characterized as the solution to a risk minimization problem, it can be estimated directly using flexible machine learning methods that accommodate custom loss functions, including random forests \citep{breiman2001random, chernozhukov2021automatic}, gradient-boosted trees \citep{BoostingFruend, lee2025rieszboost}, and neural networks \citep{abdi1999neural}. Beyond simplifying debiasing, autoDML can also improve practical performance by targeting the Riesz representer directly, sometimes outperforming more indirect estimation strategies \citep{chernozhukov2018auto, chernozhukov2021automatic, van2024stabilized}. To date, however, autoDML has been developed primarily for regression settings, including generalized regression models \citep{chernozhukov2022automatic, ichimura2022influence, chernozhukov2024automaticRiesz}, regression under covariate shift \citep{chernozhukov2023automatic}, sequential regression \citep{chernozhukov2022automatic}, and linear inverse problems \citep{bennett2023source, van2025automatic}.

Despite these advances, existing autoDML methods remain largely confined to regression settings, in which the target parameter is a functional of a real-valued regression function defined as the minimizer of a known loss and involving no additional nuisance components. However, many parameters in semiparametric statistics arise as smooth functionals of more general \textit{M-estimands}, that is, population risk minimizers over infinite-dimensional linear spaces. Examples include counterfactual regression functions \citep{rubin2007doubly}, density ratios \citep{hines2025learning}, quantile functions \citep{wang2018quantile}, hazard and survival functions \citep{westling2024inference}, and pseudo-outcome regressions, such as the conditional average treatment effect \citep{foster2023orthogonal, yang2023forster}. Moreover, existing autoDML methods rely on problem-specific Riesz representers to enable automatic debiasing. It is therefore natural to ask which structural features of a parameter allow automatic debiased estimation by such methods.


\subsection{Contributions of this work}

We develop a unified autoDML framework for inference on smooth functionals of M-estimands defined through general risk functions. The framework accommodates (Neyman-orthogonal) loss functions that depend on unknown nuisance components requiring data-adaptive estimation, as well as vector-valued M-estimands obtained by jointly minimizing a single coupled objective over product spaces. Drawing on semiparametric efficiency theory, we characterize the efficient influence function and show that the structure of these problems admits automatic debiasing via Riesz representers. In particular, we show that the Riesz representation theorem applies to the linearized functional under the inner product induced by the Hessian of the population risk. This yields a unified perspective on automatic debiasing, showing that it is determined jointly by the local behavior of the target functional and the curvature of the population risk for the M-estimand.

\vspace{0.3em}
Our main contributions are as follows:
\begin{enumerate}
    \item[(i)] We develop a general autoDML framework for inference on smooth functionals of infinite-dimensional M-estimands. We propose estimators based on one-step estimation, targeted minimum loss-based estimation, and sieve-based plug-in estimation. For quadratic risk functions, these estimators satisfy a form of double robustness for linear functionals.

    \item[(ii)] We formalize smooth functionals of M-estimands as \emph{nonparametric projection parameters}. For this class of parameters, we derive a functional von Mises expansion and characterize the efficient influence function and the corresponding efficiency bound.

    \item[(iii)] We establish the asymptotic efficiency of autoDML estimators for this class of parameters.

    \item[(iv)] We study the effect of model misspecification and show that the approximation error induced by misspecification of the M-estimand model is second order.
\end{enumerate}

Together, our contributions address three limitations of existing autoDML methods: (1) they do not allow nuisance components in loss functions; (2) they do not allow functionals of multiple M-estimands that jointly minimize a coupled objective; and (3) they lack a formal efficiency theory based on projection parameters. The most closely related work is Section~3 of \citet{chernozhukov2024automaticRiesz}, which extends autoDML estimators for regression functionals to functionals of real-valued M-estimands that separately minimize known loss functions, none of which depend on nuisance parameters. By formalizing the notion of a Hessian Riesz representer, we show that this setting arises naturally as a special case of our framework. Related work is discussed in more detail in Section~\ref{sec:relatedwork}.

This paper is organized as follows. Section~\ref{sec:problemsetup} introduces the generalized autoDML framework, presents illustrative examples, and reviews related literature. Section~\ref{sec::eif} derives the functional von Mises expansion and the efficient influence function for functionals of M-estimands. Section~\ref{sec::autodml} presents the proposed autoDML estimators, with the corresponding asymptotic theory developed in Section~\ref{sec::theory}. Section~\ref{sec::exp} illustrates the flexibility and practical performance of the framework through an application to long-term survival in the beta-geometric model.

\section{Overview of the generalized autoDML framework}


\label{sec:problemsetup}

 \label{sec::overview}
\subsection{Proposed Approach}

Suppose \( Z_1, Z_2, \ldots, Z_n \) are i.i.d.\ draws from \( Z \sim P_0 \), where \( P_0 \) is a distribution on \( \mathcal{Z} \) belonging to a nonparametric, convex model \( \mathcal{P} \). Our goal is to automate the construction of debiased estimators of a feature of an M-estimand \( \theta_{P_0} \), given by
\[
E_{P_0}[m(Z,\theta_{P_0})],
\]
where \( m : \mathcal{Z} \times \mathcal{H} \to \mathbb{R} \) is a known smooth functional. For each \(P \in \mathcal{P}\), the M-estimand \(\theta_P\) takes values in a normed linear space \((\mathcal{H}, \|\cdot\|_{\mathcal{H}})\) and is defined as the unique solution to the infinite-dimensional constrained M-estimation problem:
\begin{equation}
    \theta_P = \argmin_{\theta \in \mathcal{H}} L_P(\theta, \eta_P), \label{eqn::popminimizer}
\end{equation}
where \(L_P(\theta, \eta) := E_P[\ell_{\eta}(Z, \theta)]\) is a risk functional, \(\ell_{\eta} : \mathcal{Z} \times \mathcal{H} \to \mathbb{R}\) is a loss function, and \(\eta_P\) is a nuisance function valued in a normed linear space \((\mathcal{N}, \|\cdot\|_{\mathcal{N}})\). The estimand corresponds to the parameter \(\Psi(P_0)\), where \(\Psi : P \mapsto \psi_P(\theta_P)\) is defined over \(\mathcal{P}\), and \(\psi_P(\theta) := E_P[m(Z, \theta)]\) is a functional on \(\mathcal{H}\). Our framework supports joint minimization over multiple function classes and nonlinear functionals of several M-estimands. In particular, \( \theta_P \) and \( \eta_P \) may be vector-valued, for instance, when \( \mathcal{H} \) and \( \mathcal{N} \) are tensor products of \( L^2 \)-spaces. For notational simplicity, we write \( S_0 \) for any summary \( S_{P_0} \) of the true distribution \( P_0 \), such as \( \eta_0 \equiv \eta_{P_0} \) and $\psi_{0} \equiv \psi_{P_0}$. We denote \(P f := \int f(z)\,dP(z)\) and \(P_n\) as the empirical measure of \(\{Z_i\}_{i=1}^n\).

 We assume the risk \(L_0(\theta, \eta)\) is twice Fréchet differentiable in \(\theta\) and once differentiable in \(\eta\), with derivatives
\[
\begin{aligned}
\partial_{\theta} L_0(\theta_0, \eta_0)(h) 
&:= \left.\frac{d}{dt} L_0(\theta_0 + t h, \eta_0) \right|_{t=0}, \\
\partial_{\theta}^2 L_0(\theta_0, \eta_0)(h_1, h_2) 
&:= \left.\frac{d}{dt} \frac{d}{ds} L_0(\theta_0 + t h_1 + s h_2, \eta_0) \right|_{t=0,\, s=0}.
\end{aligned}
\]
Such smoothness requirements also appear in parametric M-estimation theory to permit inference. The precise differentiability assumptions are detailed in Section~\ref{sec::eif}. For this overview, we assume that \(\ell_{\eta_0}(\cdot,\theta)\) is twice G\^ateaux differentiable in \(\theta\), with \(\dot{\ell}_{\eta_0}(\theta_0) : \mathcal H \to L^2(P_0)\) and \(\ddot{\ell}_{\eta_0}(\theta_0) : \mathcal H \times \mathcal H \to L^1(P_0)\), and that, for all \(h,h_1,h_2 \in \mathcal H\),
\[
\partial_\theta L_0(\theta_0,\eta_0)(h) = P_0 \dot{\ell}_{\eta_0}(\theta_0)(h),
\qquad
\partial_\theta^2 L_0(\theta_0,\eta_0)(h_1,h_2)
=
P_0 \ddot{\ell}_{\eta_0}(\theta_0)(h_1,h_2).
\]

We assume that the loss \(\ell_{\eta_0}\) is Neyman-orthogonal in \(\eta_0\), so that small errors in estimating \(\eta_0\) do not affect the risk \(L_0(\cdot, \eta_0)\) to first order \citep{foster2023orthogonal}. Formally, we require that the cross-derivative of the risk satisfies \(\partial_{\eta} \partial_{\theta} L_0(\theta_0, \eta_0)(\cdot, \cdot) = 0\). In parametric models, this condition ensures that the M-estimator of \(\theta_0\) behaves asymptotically as if \(\eta_0\) were known, yielding root-\(n\) consistency and asymptotic normality. Neyman orthogonality imposes no substantive restriction, as any loss can typically be orthogonalized without altering the underlying risk. Orthogonal losses arise naturally in a wide range of settings, including causal inference and missing data problems \citep{rubin2007doubly, foster2023orthogonal}.



Let \(\eta_n\) and \(\theta_n\) be estimators of the nuisance function \(\eta_0\) and the M-estimand \(\theta_0\). For example, \(\theta_n\) could be obtained by empirical risk minimization based on the orthogonal risk \(\theta \mapsto \tfrac{1}{n}\sum_{i=1}^n \ell_{\eta_n}(Z_i, \theta)\). The plug-in estimator \(\psi_n(\theta_n) := \tfrac{1}{n}\sum_{i=1}^n m(Z_i, \theta_n)\) may then be used to estimate \(\Psi(P_0)\). However, when \(\theta_n\) is obtained using flexible statistical learning tools, this estimator typically lacks both \(n^{1/2}\) convergence and asymptotic normality due to bias from first-order dependence on the M-estimation error \(\theta_n - \theta_0\) \citep{vanderLaanRose2011, chernozhukov2018double}, even when \(\eta_0\) is known. To address this sensitivity, debiasing methods are typically required to remove the first-order bias of the plug-in estimator \(\psi_n(\theta_n)\).

For the special case in which the target parameter is a linear functional \(\psi_0\) of the outcome regression \(\theta_0 : x \mapsto E_0[Y \mid X = x]\), with \(Z = (X,Y)\), \citet{chernozhukov2021automatic} developed an automatic approach to debiasing the plug-in estimator \(\psi_0(\theta_n)\). They showed, by the Riesz representation theorem, that the plug-in error \(\psi_0(\theta_n) - \psi_0(\theta_0)\) can be written as \(\langle \alpha_0, \theta_n - \theta_0 \rangle_{L^2(P_{0,X})} = \int \alpha_0(x)\{\theta_n(x) - \theta_0(x)\}\,dP_{0,X}(x) = \int \alpha_0(x)\{\theta_n(x) - y\}\,dP_0(z)\) for an appropriate representer \(\alpha_0\). This yields the automatic DML estimator
\[
\psi_0(\theta_n) - \frac{1}{n}\sum_{i=1}^n \alpha_n(X_i)\{\theta_n(X_i) - Y_i\},
\]
introduced in \citet{chernozhukov2018double, chernozhukov2021automatic}. The second term corrects the plug-in error through a one-step bias correction \citep{bickel1993efficient}, thereby enabling fast rates and valid inference for \(\Psi(P_0)\), even when \(\alpha_0\) and \(\theta_0\) are estimated at slow rates. This estimator generalizes the well-known augmented inverse probability weighted estimator \citep{robinsCausal, robins1995analysis, bruns2023augmented}. A key contribution of \citet{chernozhukov2021automatic} is the observation that the Riesz representer \(\alpha_0\) can itself be characterized as the solution to a risk minimization problem:
\[
\arg\min_{\alpha \in \mathcal{H}} E_0\{\alpha(X)^2 - 2\,m(Z,\alpha)\} =
\arg\min_{\alpha \in \mathcal{H}} E_0\big[\{\alpha(X)\}^2 - 2\alpha(X)\alpha_0(X)\big]
=
\arg\min_{\alpha \in \mathcal{H}} E_0\big[\{\alpha(X) - \alpha_0(X)\}^2\big].
\]
This follows from the defining property of the Riesz representer, namely \(E_0\{m(Z,\alpha)\} = E_0\{\alpha(X)\alpha_0(X)\}\)  \citep{williams2025riesz}. Based on this characterization, they develop autoDML estimators by estimating \(\theta_0\) using the squared-error loss \(\ell_{\eta} : (z,\theta) \mapsto \{y - \theta(x)\}^2\), and estimating \(\alpha_0\) via \textit{Riesz regression}, that is, by minimizing the Riesz loss \((z,\alpha) \mapsto \alpha(x)^2 - 2\,m(z,\alpha)\).

We generalize the autoDML procedure to smooth functionals of M-estimands. Our key insight is that, for a generic M-estimand, the plug-in error \(\psi_0(\theta_n) - \psi_0(\theta_0)\) is naturally represented using the Hessian inner product \(\partial_{\theta}^2 L_0(\theta_0, \eta_0)(\cdot, \cdot)\) induced by the risk functional. Applying the Riesz representation theorem, we show that the first-order error satisfies \(\dot{\psi}_0(\theta_0)(\theta_n - \theta_0) = \partial_{\theta}^2 L_0(\theta_0, \eta_0)(\theta_n - \theta_0, \alpha_0)\), where \(\dot{\psi}_0(\theta_0): \mathcal{H} \rightarrow \mathbb{R}\) is the Fréchet derivative of \(\psi_0\) and \(\alpha_0 \in \mathcal{H}\) is its Riesz representer under the Hessian inner product. To derive a bias correction, we combine the first-order optimality condition \(\partial_{\theta} L_0(\theta_0, \eta_0)(\alpha_0) = 0\) with the Neyman-orthogonality property \(\partial_{\eta} \partial_{\theta} L_0(\theta_0, \eta_0)(\eta_n - \eta_0, \alpha_0) = 0\). Together, these yield the Hessian-gradient relation \(\partial_{\theta}^2 L_0(\theta_0, \eta_0)(\theta_n - \theta_0, \alpha_0) \approx \partial_{\theta} L_0(\theta_n, \eta_n)(\alpha_0)\). Hence, to leading order, the estimation error satisfies \(\psi_0(\theta_n) - \psi_0(\theta_0) \approx P_0\, \dot{\ell}_{\eta_n}(\theta_n)(\alpha_0)\), which motivates debiasing via the directional derivative of the loss.

\begin{figure}[htb!]
\centering
\resizebox{.8\textwidth}{!}{
\begin{tikzpicture}[
  x=1cm,y=1cm,
  font=\sffamily,
  >=Stealth,
  arrow/.style={semithick,->,line cap=butt,line join=miter},
  block/.style={rectangle,rounded corners=2pt,inner sep=6pt,minimum height=9mm,
                text width=42mm,align=center},
  blockL/.style={block,draw=orange!70!black,fill=orange!20},
  blockM/.style={block,draw=orange!70!black,fill=orange!15},
  blockY/.style={block,draw=yellow!60!black,fill=yellow!20},
  blockOut/.style={block,draw=green!55!black,fill=green!20,text width=50mm},
  spine/.style={semithick,draw=black!70}
]

\node[blockL]   (mest)      at (0,0)           {Estimate M-estimand};
\node[blockY]   (hess)      at (0,-2.1)        {Loss Hessian};

\node[blockM]   (plugin)    at (6,0)           {Estimate functional};
\node[blockM]   (riesz)     at (6,-2.1)        {Estimate Riesz representer};
\node[blockY]   (grad)      at (6,-4.2)        {Loss Gradient};

\node[blockOut] (debiased)  at (12.6,-2.1)     {Debias Estimator\\\footnotesize(One-step / TMLE)};

\coordinate (split) at ($(mest.east)!0.60!(plugin.west)$);

\draw[arrow] (mest.east) -- (split) -- (plugin.west);

\draw[spine] (split) -- (split |- grad.west);

\draw[arrow] (split |- riesz.west) -- (riesz.west);
\draw[arrow] (split |- grad.west)  -- (grad.west);

\draw[arrow] (hess.east) -- (riesz.west);

\draw[arrow] (riesz.south) -- (grad.north);

\draw[arrow] (plugin.east) -- ++(0.9,0) |- (debiased.west);
\draw[arrow] (grad.east)   -- ++(0.9,0) |- (debiased.west);

\end{tikzpicture}
}
\caption{\textbf{Given a specification of the orthogonal loss function and the target functional of the M-estimand}, the diagram illustrates the autoDML workflow.}
\end{figure}

We propose three \emph{generalized autoDML} estimators based on one-step estimation, targeted minimum loss-based estimation, and sieve-based plug-in estimation, each incorporating an automatic bias correction. Given an estimator \(\alpha_n\) of the Hessian Riesz representer \(\alpha_0\), these estimators take, either explicitly or implicitly, the form
\[
\widehat{\psi}_n := \psi_n(\theta_n) - \frac{1}{n}\sum_{i=1}^n \dot{\ell}_{\eta_n}(\theta_n)(\alpha_n)(Z_i),
\]
where the second term corrects the leading bias \(P_0 \dot{\ell}_{\eta_n}(\theta_n)(\alpha_0)\). To estimate \(\alpha_0\) automatically, we propose a \textit{generalized Riesz regression} approach based on the optimization problem
\[
\alpha_0 = \arg\min_{\alpha \in \mathcal{H}} \partial_{\theta}^2L_0(\theta_0, \eta_0)(\alpha, \alpha) - 2\,\dot{\psi}_0(\theta_0)(\alpha),
\]
which often corresponds to a quadratic loss of the form $\ddot{\ell}_{\eta_0}(\theta_0)(\alpha, \alpha)(z) - 2 \dot{m}_{\theta_0}(z,\alpha)$, where \(\dot{\psi}_0(\theta_0)(\alpha) = E_0[\dot{m}_{\theta_0}(Z,\alpha)]\) is the G\^ateaux derivative of \(\psi_0\) at \(\theta_0\) in the direction \(\alpha\). Given the nuisance estimator \(\eta_n\), the autoDML estimator \(\widehat{\psi}_n\) is automatic in the sense that its construction is agnostic to the specific functional and orthogonal loss; the same estimation pipeline applies across a broad class of smooth functionals. The procedure reduces debiased inference to two learning tasks: one for the M-estimand \(\theta_0\) and another for the Hessian representer \(\alpha_0\). The M-estimand can be learned via risk minimization using the loss \(\ell_{\eta_0}\), while the learning task for \(\alpha_0\) requires only the computation of the Hessian \(\ddot{\ell}_{\eta_n}(\theta_n)\) and a linear approximation of the target functional. Notably, the gradient and Hessian of the loss are often computed as part of standard optimization procedures for the M-estimand (e.g., gradient descent, boosting, or M-estimation) and are therefore readily available. Moreover, these derivatives can often be obtained analytically through straightforward calculus (see Lemma~\ref{lemma:smoothlosssuff} in Section~\ref{sec::eif}) or computed algorithmically via automatic differentiation \citep{rall1981automatic}.



Assuming \(n^{-1/4}\)-rate conditions on the nuisance estimators \(\eta_n\), \(\theta_n\), and \(\alpha_n\), we establish that the autoDML estimators admit the asymptotically linear expansion
\[
\widehat{\psi}_n - \Psi(P_0) = P_n \chi_0 + o_p(n^{-1/2}),
\quad \text{where} \quad
\chi_0(\cdot) := -\dot{\ell}_{\eta_0}(\theta_0)(\alpha_0)(\cdot) + m(\cdot, \theta_0) - \Psi(P_0),
\]
with \(\chi_0\) denoting the efficient influence function of \(\Psi\). For linear functionals and quadratic risk functions, we show that this expansion holds under a doubly robust rate condition:
\[
\partial^2_{\theta} L_0(\theta_0, \eta_0)(\theta_n - \theta_0, \alpha_n - \alpha_0) = o_p(n^{-1/2}),
\]
thereby extending the double robustness property for linear functionals of regression functions \citep{bang2005doubly, DoubleML}. Moreover, the autoDML estimator \(\widehat{\psi}_n\) is regular and nonparametrically efficient for \(\Psi(P_0)\), and \(n^{1/2}(\widehat{\psi}_n - \Psi(P_0))\) converges in distribution to a mean-zero normal random variable with variance \(\sigma_0^2 := \operatorname{Var}_0(\chi_0(Z))\). Wald-type confidence intervals and hypothesis tests can be constructed automatically using the influence-function-based variance estimator $\frac{1}{n} \sum_{i=1}^n \left\{ \chi_n(Z_i) \right\}^2,$
where \(\chi_n(\cdot) := -\dot{\ell}_{\eta_n}(\theta_n)(\alpha_n)(\cdot) + m(\cdot, \theta_n) - \psi_n(\theta_n)\).

\subsection{Connections to classical theory}

\subsubsection{Connection to one-step estimation in parametric models}

Our generalized autoDML estimator extends the classical one-step estimator from parametric models to infinite-dimensional settings. Consider the classical M-estimation problem of inferring a finite-dimensional M-estimand \(\theta_0 \in \mathbb{R}^d\), defined by \(\theta_0 := \argmin_{\theta \in \mathbb{R}^d} P_0 \ell(\cdot,\theta)\) for a known loss function \(\ell\). Given a preliminary estimator \(\theta_n\) of \(\theta_0\), the classical one-step estimator \citep{vandervaart2000asymptotic} is
\[
\widehat{\theta}_n
:=
\theta_n
-
\left(\frac{1}{n}\sum_{i=1}^n \nabla^2 \ell(Z_i,\theta_n)\right)^{-1}
\frac{1}{n}\sum_{i=1}^n \nabla \ell(Z_i,\theta_n),
\]
where \(\nabla \ell(\cdot,\theta_n) \in \mathbb{R}^d\) and \(\nabla^2 \ell(\cdot,\theta_n) \in \mathbb{R}^{d \times d}\) denote the gradient and Hessian of \(\theta \mapsto \ell(\cdot,\theta)\) at \(\theta_n\), respectively.

Now fix a linear functional \(\psi(\theta) := b^\top \theta\) for some \(b \in \mathbb{R}^d\). The induced one-step estimator of \(\psi(\theta_0)=b^\top\theta_0\) is
\[
\widehat{\psi}_n
=
b^\top \theta_n
-
\frac{1}{n}\sum_{i=1}^n
\left\langle
\nabla \ell(Z_i,\theta_n),
\left(\frac{1}{n}\sum_{i=1}^n \nabla^2 \ell(Z_i,\theta_n)\right)^{-1} b
\right\rangle_{\mathbb{R}^d}.
\]
To connect this with our framework, note that the first and second derivatives of the loss may be written as \(\dot{\ell}_{\theta_n}(v)(\cdot)=\langle \nabla \ell(\cdot,\theta_n), v \rangle_{\mathbb{R}^d}\) and \(\ddot{\ell}_{\theta_n}(v_1,v_2)(\cdot)=v_1^\top \nabla^2 \ell(\cdot,\theta_n) v_2\). Moreover, the vector on the right-hand side of the inner product satisfies
\[
\alpha_n
:=
\left(\frac{1}{n}\sum_{i=1}^n \nabla^2 \ell(Z_i,\theta_n)\right)^{-1} b,
\qquad
b^\top v
=
\alpha_n^\top
\left(\frac{1}{n}\sum_{i=1}^n \nabla^2 \ell(Z_i,\theta_n)\right)v
\quad \text{for all } v \in \mathbb{R}^d.
\]
Thus, \(\alpha_n\) is the empirical Hessian Riesz representer of the linear functional \(v \mapsto b^\top v\). The one-step bias correction can therefore be written as \(-P_n \dot{\ell}(\theta_n)(\alpha_n)\), so \(\widehat{\psi}_n\) is a special case of the autoDML estimator with \(\mathcal H = \mathbb{R}^d\).

\subsubsection{Connection to the canonical gradient in nonparametric models}

We illustrate how the Hessian Riesz representer recovers the classical canonical gradient for a pathwise differentiable functional of a density. Let \(\theta_0 := dP_0/d\mu\) denote the density of \(P_0\) with respect to a dominating measure \(\mu\) on \(\mathcal{Z}\), and let \(\mathcal{H} = L^2(\mu)\). Consider the least-squares risk and corresponding loss
\[
L_0(\theta) = \frac{1}{2}\int \theta(z)^2 \, d\mu(z) - E_0[\theta(Z)], \quad \ell(z,\theta) = \frac{1}{2}\int \theta(u)^2 \, d\mu(u) - \theta(z).
\]
We use the least-squares risk because it permits optimization over the unconstrained linear space \(L^2(\mu)\), which aligns with the linear-space setup of our framework.

In this example, the Hessian induces the usual \(L^2(\mu)\) inner product, so the Hessian Riesz representer \(\alpha_0\) of \(\dot{\psi}(\theta_0)\) is simply its \(L^2(\mu)\) Riesz representer:
\[
\dot{\psi}(\theta_0)(h) = \int h(z)\alpha_0(z)\,d\mu(z)
\qquad \text{for all } h \in \mathcal{H}.
\]
For any direction \(h \in L^\infty(\mu)\) satisfying \(\int h(z)\,d\mu(z)=0\), the submodel \((\theta_0+th: t \in \mathbb{R})\) is locally a valid density submodel through $\theta_0$ with score function \(s_h=h/\theta_0\) at $t = 0$. Along such submodels,
\[
\dot{\psi}(\theta_0)(h)
=
E_0\bigl[s_h(Z)\{\alpha_0-P_0\alpha_0\}(Z)\bigr]
\qquad \text{for all } h \in \mathcal{H} \text{ such that } \int h(z)\,d\mu(z)=0.
\]
Thus, \(\alpha_0 - P_0\alpha_0\) is the \(L^2(P_0)\) Riesz representer, or canonical gradient, of the pathwise derivative map \(s_h \mapsto \dot{\psi}(\theta_0)(h)\) on the tangent space \(L_0^2(P_0)\) of mean-zero, finite-variance score functions \citep{bickel1993efficient}. Under mild conditions, the classical tangent-space formulation \citep{bickel1993efficient} implies that \(\Psi\) is pathwise differentiable, with efficient influence function
\[
\chi_0(z)
=
\alpha_0(z) - E_0[\alpha_0(Z)] + m(z,\theta_0) - \Psi(P_0).
\]
Our M-estimand framework recovers the same efficient influence function, since
\[
-\dot{\ell}(\theta_0)(\alpha_0)(z)
=
\alpha_0(z) - \int \theta_0(u)\alpha_0(u)\,d\mu(u)
=
\alpha_0(z) - E_0[\alpha_0(Z)].
\]

Although our autoDML framework assumes that \(\mathcal{H}\) is linear, the same geometric idea extends to convex parameter spaces by defining the representer relative to the linear closure of the tangent cone \(\{h-\theta_0 : h \in \mathcal{H}\}\), provided the M-estimand satisfies the first-order optimality condition \(\partial_{\theta}L_0(\theta_0,\eta_0)(h)=0\) for all \(h \in \mathcal{H}\). For example, one could instead work with the log-likelihood loss over a convex class of densities. In that setting, \(\alpha_0\) coincides with the canonical gradient.

\subsection{Examples}
\label{sec::examples}

This section presents several examples of parameters in semiparametric statistics and causal inference that fall within our framework.

Our first example illustrates that continuous linear functionals of regressions, weight functions such as density ratios, and Riesz representers all fall within our framework \citep{chernozhukov2018double}. In particular, the autoDML framework of \citet{chernozhukov2022automatic} for regression functionals arises as a special case of our approach, including estimands such as counterfactual means and the average treatment effect. We consider the data structure \(Z = (X, A, Y)\), where \(X \in \mathbb{R}^d\) denotes covariates, \(A \in \{0,1\}\) is a binary treatment indicator, and \(Y \in \mathbb{R}\) is the outcome. The outcome regression is defined by \(\mu_P(a,x) := E_P[Y \mid A=a, X=x]\). This example serves as a running illustration throughout the paper, and all examples are examined in greater detail in Appendix~\ref{sec:examplesback}.

\renewcommand{\theexample}{1a}

\begin{example}[Functionals of regressions and representers]
Suppose that \(E_0[Y^2] < \infty\). Then \(\mu_0\) is an M-estimand corresponding to the squared-error loss \(\ell_{\eta}(z, \theta) := \frac{1}{2}\{y - \theta(x)\}^2\), with \(\mathcal{H} := L^2(P_{0,X,A})\), where \(P_{0,X,A}\) denotes the marginal distribution of \((X, A)\) under \(P_0\). The loss derivative \(\dot{\ell}_{\eta}(\theta)(\alpha)\) equals \(z \mapsto \alpha(x)\{\theta(x) - y\}\), and the Hessian inner product \(\partial_{\theta}^2 L_0(\theta_0, \eta_0)\) equals the \(L^2(P_{0,X,A})\) inner product. The Hessian Riesz representer satisfies \(\alpha_0 := \arg\min_{\alpha \in \mathcal{H}} E_0\left[\alpha(X, A)^2 - 2\,\dot{m}_{\theta_0}(Z, \alpha)\right]\). The ATE, defined as \(E_0[\mu_0(1,X) - \mu_0(0,X)]\), corresponds to the functional \(m : (z, \theta) \mapsto \theta(1,x) - \theta(0,x)\). More generally, suppose that \(\theta_0\) is the Riesz representer in \(L^2(P_0)\) of the continuous linear functional \(\theta \mapsto E_0[r(Z, \theta)]\), where \(r : \mathcal{Z} \times \mathcal{H} \to \mathbb{R}\) is linear in its second argument and satisfies \(|E_0[r(Z, \theta)]| \leq M \|\theta\|_{L^2(P_0)}\) for all \(\theta \in \mathcal{H}\), for some constant \(M < \infty\). Then \(\theta_0\) is an M-estimand corresponding to the quadratic loss \(\ell(\theta,z) := \frac{1}{2}\theta(x)^2 - r(z,\theta)\) \citep{chernozhukov2018auto}; more generally, one may use Riesz losses based on Bregman divergences \citep{hines2025learning}. Examples of Riesz representers include inverse probability weights and density ratios \citep{williams2025riesz, hines2025learning}, while the outcome regression is recovered by taking \(r(z,\theta) = y\theta(x)\). \qedsymbol
\end{example}

Our framework extends the autoDML approach of \citet{chernozhukov2021automatic} for generalized linear regression models by allowing loss functions that depend on unknown nuisance parameters. It accommodates functionals of treatment effect modifiers, such as the conditional average treatment effect (CATE), conditional relative risk, and conditional log-odds difference, which arise as M-estimands under suitable orthogonal losses \citep{nekipelov2022regularised, van2024combining}. Although these treatment effect modifiers can be expressed in terms of the outcome regression, a key advantage of targeting them directly through the loss is that doing so can simplify estimation under semiparametric models by allowing the model class \(\mathcal{H}\) to be tailored to the relevant nuisance components. For example, partially linear models for the outcome regression correspond to optimizing over a linear class for the CATE.

\renewcommand{\theexample}{1b}

\begin{example}[Functionals in semiparametric regression models]
\label{example::splogistic}
Let \(g: \mathbb{R} \to \mathbb{R}\) be a known link function, and suppose the outcome regression \(\mu_0\) satisfies the semiparametric model \(g(\mu_0(A, X)) = h_0(X) + A \tau_0(X)\), where \(h_0 := g(\mu_0(0, \cdot))\) is unrestricted and \(\tau_0 \in \mathcal{T} \subset L^2(P_{0,X})\). In the binary case with a logistic link, \cite{chernozhukov2021automatic} use the binomial log-likelihood loss \(\ell_{\text{loglik}}(z, \theta) := \log(1 + \exp(\theta(a, x))) - y \theta(a, x)\) to estimate functionals of \(\mu_0\). However, often we care only about a functional \(\psi_0(\tau_0)\) of the treatment effect modifier \(\tau_0\). Rather than modeling the entire outcome regression, we can directly target \(\tau_0\), which is an M-estimand for the orthogonal loss \citep{nekipelov2022regularised}:
\[
\ell_{\eta_0}(z, \theta) := \frac{1}{\nu_0(a, x)} \left\{ \log \left(1 + \exp(- \{a - \pi_0(x)\} \theta(x) - h_0(x)) \right) - y \{a - \pi_0(x)\} \theta(x) \right\},
\]
where \(\eta_0 = (\pi_0, \mu_0)\), \(\pi_0(x) := P_0(A = 1 \mid X = x)\), \(h_0(x) := \pi_0(x) \logit \mu_0(1, x) + (1 - \pi_0(x)) \logit \mu_0(0, x)\), and \(\nu_0(a, x) := \mu_0(a, x)\{1 - \mu_0(a, x)\}\). The corresponding loss for the Hessian Riesz representer \(\alpha_0\) is given by \((z, \alpha) \mapsto \{a - \pi_0(x)\}^2 \alpha(a, x)^2 - 2\,\dot{m}_{\theta_0}(z, \alpha)\), which can be implemented using Riesz regression with overlap weights \citep{li2019overlapWeights}. \qedsymbol
\end{example}

Our framework also supports smooth functionals of pseudo-outcome regressions, which arise from Neyman-orthogonal quadratic losses. This class of functionals lies outside the scope of prior work such as \cite{chernozhukov2021automatic, chernozhukov2022automatic, chernozhukov2024automaticRiesz}.


\renewcommand{\theexample}{1c}

\begin{example}[Functionals of pseudo-outcome regressions]
\label{example::orthoLS}
Many Neyman-orthogonal loss functions take the form $(z, \theta) \mapsto \ell_{\eta_0}(z, \theta) = \frac{1}{2} w_{\eta_0}(z)\left\{\zeta_{\eta_0}(z) - \theta(x)\right\}^2,$ where $w_{\eta_0}(Z)$ is a pseudo-weight and $\zeta_{\eta_0}(Z)$ is a pseudo-outcome \citep{rubin2007doubly, morzywolek2023general, yang2023forster}. The loss derivative is $\dot{\ell}_{\eta}(\theta)(\alpha) = z \mapsto w_{\eta}(z)\alpha(x)\{\theta(x) - \zeta_{\eta}(z)\}$, and the Hessian Riesz representer $\alpha_0$ solves $\argmin_{\alpha \in \mathcal{H}} E_0[w_{\eta_0}(Z)\alpha(X,A)^2 - 2\,\dot{m}_{\theta_0}(Z, \alpha)]$. This includes the DR-learner \citep{van2015OptRule, kennedy2023towards, van2023causal} and R-learner \citep{nie2021quasi} for the CATE $\theta_0: x \mapsto \mu_0(1, x) - \mu_0(0, x)$. The DR-learner uses $w_{\eta_0}(z) = 1$ and $\zeta_{\eta_0}(z) = \mu_0(1, x) - \mu_0(0, x) + \frac{a - \pi_0(x)}{\pi_0(x)\{1 - \pi_0(x)\}}\{y - \mu_0(a, x)\},$
where $\eta_0 = (\pi_0, \mu_0)$, with $\mu_0: (a, x) \mapsto E_0[Y \mid A = a, X = x]$ and $\pi_0: x \mapsto P_0(A = 1 \mid X = x)$. The R-learner uses $w_{\eta_0}(z) = (a - \pi_0(x))^2$ and $\zeta_{\eta_0}(z) = \frac{y - m_0(x)}{a - \pi_0(x)}$, where $\eta_0 = (\pi_0, m_0)$ and $m_0: x \mapsto E_0[Y \mid X = x]$. Pseudo-outcome regressions also arise in CATE estimation under unmeasured confounding with instrumental variables \citep{syrgkanis2019machine} and in proximal causal inference \citep{yang2023forster, sverdrup2023proximal}. \qedsymbol
\end{example}

Finally, a key generalization of our approach is that it extends autoDML to functionals of multiple M-estimands that are jointly defined through a coupled risk minimization problem. As a special case, the following example shows that our framework also covers functionals of multiple M-estimands that solve distinct risk minimization problems, even when the target parameter depends on all of them jointly, including the setting considered in Section~3 of \citet{chernozhukov2024automaticRiesz}.

\renewcommand{\theexample}{2}
\begin{example}[Separable product spaces and vector-valued M-estimands]
Suppose that \(\mathcal{H} = \mathcal{H}_1 \times \cdots \times \mathcal{H}_K\), where each \(\mathcal{H}_k\) is a Hilbert space, and write \(\theta = (\theta_1,\dots,\theta_K) \in \mathcal{H}\) and \(\theta_0 = (\theta_{0,1},\dots,\theta_{0,K}) \in \mathcal{H}\). Consider a loss of the form
\[
\ell_{\eta_0}(z,\theta) := \sum_{k=1}^K \ell_{\eta_0}^{(k)}(z,\theta_k),
\]
with componentwise directional derivative \(\dot{\ell}_{\eta_0}^{(k)}(\theta_{0,k})(h_k)(z) := \partial_{\theta_k}\ell_{\eta_0}^{(k)}(z,\theta_{0,k})(h_k)\). Then the population risk is additively separable across components, with \(L_0(\theta,\eta_0) = \sum_{k=1}^K L_0^{(k)}(\theta_k,\eta_0)\), where \(L_0^{(k)}(\theta_k,\eta_0) := P_0\ell_{\eta_0}^{(k)}(\cdot,\theta_k)\), so each \(\theta_{0,k}\) minimizes its own component risk. Moreover, for \(h=(h_1,\dots,h_K)\in\mathcal H\),
\[
\dot{\ell}_{\eta_0}(\theta_0)(h)
=
\sum_{k=1}^K \dot{\ell}_{\eta_0}^{(k)}(\theta_{0,k})(h_k),
\qquad
\partial_{\theta}^2 L_0(\theta,\eta_0)(h,h')
=
\sum_{k=1}^K \partial_{\theta_k}^2 L_0^{(k)}(\theta_k,\eta_0)(h_k,h_k').
\]
The Hessian Riesz representer in the product space \(\mathcal{H}\) is therefore given componentwise by
\[
\alpha_{0,\mathcal H}
=
\bigl(\alpha_{0,\mathcal H_1}^{(1)},\dots,\alpha_{0,\mathcal H_K}^{(K)}\bigr),
\]
where \(\alpha_{0,\mathcal H_k}^{(k)}\) is the Riesz representer of the linear functional \(h_k \mapsto \partial_{\theta_k}\psi_0(\theta_0)(h_k) := \partial_\theta \psi_0(\theta_0)(0,\dots,0,h_k,0,\dots,0)\) with respect to the \(k\)-th component Hessian inner product \((h_k,h_k') \mapsto \partial_{\theta_k}^2 L_0^{(k)}(\theta_{0,k},\eta_0)(h_k,h_k')\). Furthermore,
\[
\dot{\ell}_{\eta_0}(\theta_0)(\alpha_0)
=
\sum_{k=1}^K \dot{\ell}_{\eta_0}^{(k)}(\theta_{0,k})(\alpha_{0,\mathcal H_k}^{(k)}).
\]
 \qed
\end{example}

More generally, our framework accommodates functionals of vector-valued M-estimands defined by a single coupled objective function, in which case the separable structure of the previous example need not hold. In Section~\ref{sec::exp::survival}, we show that the long-term mean survival probability under the semiparametric beta-geometric model also falls within this framework \citep{hubbard2021beta}. In that example, the target parameter is a nonlinear functional of two infinite-dimensional shape parameters that are estimated jointly, and thus lies outside the parameter classes considered in prior autoDML work.

\subsection{Related work}

\label{sec:relatedwork}

Our work builds on the growing literature on automatic debiased machine learning. \citet{chernozhukov2022automatic} introduce autoDML estimators for functionals of the outcome regression, and \citet{chernozhukov2021automatic} extend this approach to generalized linear models. The most closely related work is Section~3 of \citet{chernozhukov2024automaticRiesz}, which outlines autoDML estimators for functionals of real-valued M-estimands that separately minimize known loss functions, none of which depend on nuisance parameters. These estimators rely on the influence function representations developed by \citet{ichimura2022influence} for such functionals.

Our framework generalizes these approaches in several key directions. First, by leveraging Neyman orthogonality, we allow for loss functions that may depend on nuisance parameters requiring data-driven estimation. Second, we allow functionals of vector-valued M-estimands that jointly minimize a single objective function. To enable this, we accommodate M-estimands defined over general linear spaces, such as tensor product spaces, rather than restricting to \(L^2\). Third, we impose more refined and transparent conditions than prior work. For instance, \citet{chernozhukov2024automaticRiesz} assume the existence of certain Riesz representers and require a quadratic remainder condition on derivative approximations (see Assumption~8). In contrast, under our conditions, the existence of the representer is guaranteed, and control of the remainder terms follows directly from functional H\"older smoothness. 

As another key contribution, we formalize the class of parameters estimated by autoDML as \emph{nonparametric projection parameters}. For this class, we characterize the EIF and establish that autoDML estimators are regular and efficient. In contrast, neither \citet{chernozhukov2022automatic} nor \citet{chernozhukov2024automaticRiesz} investigate the efficiency properties of their estimators.

Beyond one-step corrections—the primary focus of prior work—we further introduce two new classes of estimators: automatic targeted minimum loss-based (autoTML) and automatic sieve-based plug-in (autoSieve) estimators. While the one-step estimator applies an explicit bias correction to the plug-in estimator, autoTML directly debiases the M-estimator itself, ensuring that the resulting plug-in estimator is efficient. Although \citet{chernozhukov2022automatic} proposed a special case of the autoTML estimator for linear functionals under squared error loss, they did not consider extensions to general loss functions or target functionals. Finally, for sieve estimation, we propose a novel automatic undersmoothing procedure that eliminates the need for manual tuning or rate selection.

Our framework builds on and substantially generalizes the asymptotic theory of sieve-based M-estimation without nuisance components \citep{shen1997methods, spnpsieve}. The Hessian inner product and Riesz representer \(\alpha_0\) arise implicitly in the analysis of the asymptotic normality of sieve-based plug-in estimators (e.g., Section~4.1.1 of \citealt{chen2014sieve}; \citealt{sieveOneStepPlugin}). In these works, the Riesz representer is primarily used as a technical device in asymptotic arguments. In contrast, we find that this representer plays a central role in automatic debiasing for semiparametric statistics—through its appearance in the efficient influence function of \(\Psi\) and its characterization as the minimizer of a risk functional. As a special case, our results show that sieve plug-in estimators can be viewed as one-step estimators with a specific influence function correction. Consequently, using the machinery of our general framework, we provide simple proofs of classical results in sieve theory for known loss functions and extend the theory to orthogonal losses involving data-driven nuisance estimation \citep{shen1997methods, sieveOneStepPlugin, spnpsieve, qiu2021universal}.

\section{Functional bias expansion and statistical efficiency}
 \label{sec::eif}

\subsection{Differentiability conditions for loss functions and target parameters}

In this section, we formalize the smoothness conditions for the risk functional \((\theta, \eta) \mapsto L_0(\theta, \eta)\) and the target functional \(\psi_0\) in terms of Fr\'echet differentiability. In the next subsection, we present a functional von Mises expansion for \(\psi_0\) and derive the efficient influence function of \(\Psi\).

In the following conditions, let \(\rho_{\mathcal{H}}(\cdot)\) be an auxiliary norm on \(\mathcal{H}\), taking values in \([0, \infty]\). For example, \(\rho_{\mathcal{H}}(\cdot)\) may be an essential supremum norm, while \(\|\cdot\|_{\mathcal{H}}\) may be an \(L^2\) norm. All norms may depend on \(P_0\); for instance, \(\|\cdot\|_{\mathcal{H}}\) may be the \(L^2(P_0)\) norm.

\begin{enumerate}[label=\bf{A\arabic*)}, ref={A\arabic*}, series=cond1]
\item \textit{(Existence and uniqueness).}
For all \(P \in \mathcal{P}\), \(\theta_P = \arg\min_{\theta \in \mathcal{H}} L_P(\theta, \eta_P)\) exists and is unique.
\label{cond::uniqueness}
\end{enumerate}
Hereafter, define the convex subsets \(\mathcal{H}_{\mathcal{P}} := \operatorname{conv}(\{\theta_P : P \in \mathcal{P}\})\) and \(\mathcal{N}_{\mathcal{P}} := \operatorname{conv}(\{\eta_P : P \in \mathcal{P}\})\), where \(\operatorname{conv}(\cdot)\) denotes the convex hull. In what follows, some of our conditions will only be required to hold for over \(\theta \in \mathcal{H}_{\mathcal{P}}\) and \(\eta \in \mathcal{N}_{\mathcal{P}}\). 

Our next conditions concern the functional differentiability of the target and risk. A map \(T : \mathcal{F} \to \mathcal{G}\) between normed spaces \((\mathcal{F}, \|\cdot\|_{\mathcal{F}})\) and \((\mathcal{G}, \|\cdot\|_{\mathcal{G}})\) is \emph{Fr\'echet differentiable} at \(f_0 \in \mathcal{F}\) if there exists a continuous linear operator \(\partial_f T(f_0): \mathcal{F} \to \mathcal{G}\) such that $\| T(f_0 + h) - T(f_0) - \partial_f T(f_0)(h) \|_{\mathcal{G}} = o(\|h\|_{\mathcal{F}}) \quad \text{as } \|h\|_{\mathcal{F}} \to 0$ \citep{rudin1991functional, vanderVaartWellner}. When \(\mathcal{G} = \mathbb{R}\), the derivative \(\partial_f T(f_0)\) lies in the \emph{dual space} \(\mathcal{F}^*\), the normed linear space of continuous linear functionals \(B: \mathcal{F} \to \mathbb{R}\), equipped with the operator norm \(\|B\|_{\mathcal{F}} := \sup_{f \in \mathcal{F}} |B(f)|\). Higher-order derivatives are defined recursively: \(T\) is \emph{twice Fr\'echet differentiable} at \(f_0\) if \(f \mapsto \partial_f T(f)\) is Fr\'echet differentiable into \(\mathcal{F}^*\), with \(\partial_f^2 T(f_0): \mathcal{F} \times \mathcal{F} \to \mathbb{R}\) given by this derivative.

\begin{enumerate}[label=\bf{A\arabic*)}, ref={A\arabic*}, resume=cond1]
    \item \textit{(H\"older Smoothness of target functional):} \label{cond::smoothfunctional}  The functional $\psi_0: (\mathcal{H}, \| \cdot\|_{\mathcal{H}}) \rightarrow \mathbb{R}$ is Fr\'echet differentiable over $\mathcal{H}_{\mathcal{P}} \subset \mathcal{H}$ with Lipschitz continuous derivative $\theta\mapsto \dot{\psi}_0(\theta)$. 

\item \textit{(H\"older smoothness of risk functional):}  \label{cond::targetsmoothloss}
\begin{enumerate}[label={\roman*)}, ref={\ref{cond::targetsmoothloss}\roman*}]
\item \label{cond::targetsmoothloss::one} \textit{(First-order Fr\'echet differentiability)} For all $\theta \in \mathcal{H}_{\mathcal{P}}$ and $\eta \in \mathcal{N}_{\mathcal{P}}$, there exists a continuous linear operator $\dot{\ell}_\eta(\theta): (\mathcal{H}, \|\cdot\|_{\mathcal{H}})\to L^2(P_0)$ such that, for all \(P \in \mathcal{P}\), the map \(L_P(\,\cdot\,, \eta) : (\mathcal{H}, \|\cdot\|_{\mathcal{H}})\rightarrow \mathbb{R}\) is Fr\'echet differentiable at $\theta$ with derivative satisfying $\partial_\theta L_P(\theta, \eta)(h) = P\, \dot{\ell}_\eta(\theta)(h)$ for all $h \in \mathcal{H}$. 

\item \label{cond::targetsmoothloss::two}
\textit{(Second-order Fr\'echet differentiability)} For each $\eta \in \mathcal{N}_{\mathcal{P}}$, the map \( \partial_\theta L_0(\cdot, \eta): (\mathcal{H}, \|\cdot\|_{\mathcal{H}}) \rightarrow (\mathcal{H}, \rho_{\mathcal{H}})^*\) is Fr\'echet differentiable over $ \mathcal{H}_{\mathcal{P}}$, where $\partial_\theta^2 L_0(\theta , \eta )(\cdot, \cdot)$ is symmetric.

\item \textit{(Lipschitz continuity of second derivative)} \label{cond::targetsmoothloss::three}  There exists a constant \(C < \infty\) such that, for all \(\theta, \theta' \in \mathcal{H}_{\mathcal{P}}\), \(\eta, \eta' \in \mathcal{N}_{\mathcal{P}}\), and \(h_1, h_2 \in \mathcal{H}\) with \(\|h_1\|_{\mathcal{H}} + \rho_{\mathcal{H}}(h_2) \leq 1\),
\[
\left| \partial^2_\theta L_0(\theta', \eta')(h_1, h_2) - \partial^2_\theta L_0(\theta, \eta)(h_1, h_2) \right| \leq C \left( \|\theta' - \theta\|_{\mathcal{H}} + \|\eta' - \eta\|_{\mathcal{N}} \right).
\]
\end{enumerate}
 
  \item \label{cond::nuisancesmooth} \textit{(H\"older Smoothness in nuisance parameter):} 
    \begin{enumerate}[label={\roman*)}, ref={\ref{cond::nuisancesmooth}\roman*}]
    \item \label{cond::crossderiv}  The map \( \partial_\theta L_0(\theta_0, \cdot): (\mathcal{N}, \|\cdot\|_{\mathcal{N}}) \rightarrow (\mathcal{H}, \rho_{\mathcal{H}})^*\) is Fr\'echet differentiable at each $\eta \in \mathcal{N}_{\mathcal{P}}$ with derivative \(\partial_\eta \partial_\theta L_0(\theta_0, \eta)\).

    \item \label{cond::crossderivlipschitz} \textit{(Lipschitz continuity of cross derivative).}      There exists a constant \(C < \infty\) such that, for all \(\eta \in \mathcal{N}_{\mathcal{P}}\), \(h \in \mathcal{H}\), and $g \in \mathcal{N}$ with \(\|g\|_{\mathcal{N}} + \rho_{\mathcal{H}}(h) \leq 1\),
$$\left| \partial_\eta \partial_\theta L_0(\theta_0, \eta)(g, h) - \partial_\eta \partial_\theta L_0(\theta_0, \eta_0)(g, h) \right| \leq C     \|\eta - \eta_0\|_{\mathcal{N}} .$$

   \end{enumerate} 
 
\end{enumerate}

When \((\mathcal{H}, \|\cdot\|_{\mathcal{H}})\) is a Hilbert space, Condition~\ref{cond::uniqueness} holds if the map \(\theta \mapsto L_P(\theta, \eta_P)\) is strictly convex, coercive, and lower semi-continuous \citep{ekeland1999convex}. A sufficient condition is strong convexity: there exists \(C > 0\) such that \(\partial_\theta^2 L_P(\theta, \eta_P)(h, h) \geq C \|h\|^2\) for all \(\theta, h \in \mathcal{H}\) \citep{alexanderian2019optimization}, with examples given in \cite{foster2023orthogonal, van2024combining}.  Condition~\ref{cond::smoothfunctional} ensures that the target \(\Psi(P_0) = \psi_0(\theta_0)\) is a smooth functional of the M-estimand \(\theta_0\), a key requirement for \(\sqrt{n}\)-consistent, regular estimation \citep{van1991differentiable}. Lipschitz continuity of the derivative guarantees a first-order Taylor expansion with quadratic remainder: $\psi_0(\theta) = \psi_0(\theta_0) + \dot{\psi}_0(\theta_0)(\theta - \theta_0) + O(\|\theta - \theta_0\|_{\mathcal{H}}^2),$ uniformly over \(\theta \in \mathcal{H}\). In many cases, \(\dot{\psi}_0(\theta)(h) = E_0[\dot{m}_{\theta}(Z, h)]\) for some Gâteaux derivative \(\dot{m}_{\theta}: \mathcal{Z} \times \mathcal{H} \to \mathbb{R}\); for a continuous linear functional \(\psi_0\), we have \(\dot{\psi}_0(\theta) = \psi_0\) and \(\dot{m}_{\theta} = m\).  Condition~\ref{cond::targetsmoothloss::one} ensures that \(\theta_0\) satisfies the first-order condition \(\partial_\theta L_0(\theta_0, \eta_0)(h) = 0\) for all \(h \in \mathcal{H}\). Conditions~\ref{cond::targetsmoothloss::two},~\ref{cond::targetsmoothloss::three}, and~\ref{cond::nuisancesmooth} ensure that \(\partial_\theta L_0(\theta, \eta)\) is jointly smooth in \(\theta\) and \(\eta\), and admits a Taylor expansion: $\partial_\theta L_0(\theta, \eta)(h) = \partial_\theta L_0(\theta_0, \eta_0)(h) + \partial^2_{\theta} L_0(\theta_0 , \eta_0)(\theta - \theta_0, h) + \partial_{\eta}\partial_{\theta} L_0(\theta_0 , \eta_0)(\eta - \eta_0, h) + O(\rho_{\mathcal{H}}(h)(\|\theta - \theta_0\|_{\mathcal{H}}^2 + \|\eta - \eta_0\|_{\mathcal{N}}^2))$ uniformly over $h \in \mathcal{H}$. In many cases, \(\partial^2_\theta L_0(\theta, \eta)(h_1, h_2) = P_0\, \ddot{\ell}_\eta(\theta)(h_1, h_2)\), where \(\ddot{\ell}_\eta(\theta): \mathcal{H} \times \mathcal{H} \to L^1(P_0)\) denotes the second Gâteaux derivative of the loss. However, the loss need not be twice differentiable; these conditions can also hold for nonsmooth losses, such as absolute error loss and quantile loss (see Appendix~\ref{sec:examplesback}).


To facilitate data-driven nuisance estimation, we require that the population risk admits a Neyman-orthogonal loss \citep{foster2023orthogonal}. Such losses enable fast rates for empirical risk minimization, even when nuisances are estimated at slower rates. Formally, we require:
\begin{enumerate}[label=\bf{A\arabic*)}, ref={A\arabic*}, resume=cond1]

     \item \textit{(Neyman orthogonality in nuisance parameter):} \label{cond::orthogonal}   $\partial_\eta \partial_\theta L_0(\theta_0, \eta_0)(g, h) = 0$ for each $(g ,h) \in \mathcal{N} \times \mathcal{H} $.  
\end{enumerate}

Condition~\ref{cond::orthogonal} ensures that the estimating equation $(\partial_{\theta}L_0(\theta_0, \eta_0)(h) = 0: h \in \mathcal{H})$ for $\theta_0$ is Neyman-orthogonal with respect to the nuisance function \(\eta_0\) \citep{robins1995analysis, vanderlaanunified, DoubleML}, making the solution \(\theta_0\) locally insensitive to perturbations in \(\eta_0\). While \ref{cond::orthogonal} requires orthogonality at the particular point \(\theta_0\), many loss functions are universally orthogonal, satisfying \(\partial_\eta \partial_\theta L_0(\theta, \eta_0) = 0\) for all \(\theta \in \mathcal{H}\) \citep{whitehouse2024orthogonal}. Broad classes of Neyman-orthogonal losses in causal inference and missing data settings have been developed in \citet{rubin2007doubly, foster2023orthogonal, yang2023forster, van2024combining}. A Neyman-orthogonal loss \(\ell_{\eta_0}\) often exists if the map \(P \mapsto L_P(\theta, \eta_P)\) is pathwise differentiable for each \(\theta \in \mathcal{H}\), in which case the loss can be derived from the efficient influence function (see, e.g., Theorem~1 of \cite{van2024combining}).  Profiled loss functions—formed by substituting a nuisance parameter that depends on the M-estimand—are generally Neyman-orthogonal to the nuisance components involved in profiling \citep{murphy2000profile}.

The following lemma verifies that the smoothness conditions are satisfied for a broad class of loss functions commonly studied in the literature. In this lemma, we define the linear spaces of vector-valued functions $\mathcal{H} = (L^\infty(\lambda))^{d_1}$ and $\mathcal{N} = (L^\infty(\lambda))^{d_2}$, where $\lambda$ is a measure on $\mathcal{Z} \subset \mathbb{R}^d$ that dominates each $P \in \mathcal{P}$. Define the $L^2$-norms $\|\theta\|_{\mathcal{H}} := \left( \int \|\theta(z)\|_{\mathbb{R}^{d_1}}^2 \, dP_0(z) \right)^{1/2}$ and $\|\eta\|_{\mathcal{N}} := \left( \int \|\eta(z)\|_{\mathbb{R}^{d_2}}^2 \, dP_0(z) \right)^{1/2}$, and the essential supremum norm $\rho_{\mathcal{H}}(\theta) := \operatorname*{ess\,sup}_{z \in \mathcal{Z}} \|\theta(z)\|_{\infty}$.

\begin{lemma}
\label{lemma:smoothlosssuff}
Suppose $\ell_{\eta}(\theta,z) = l\bigl(\theta(z), \eta(z), z\bigr)$, where $l : \mathbb{R}^{d_1} \times \mathbb{R}^{d_2} \times \mathbb{R}^d \to \mathbb{R}$, and let $\mathcal{C}$ denote the closure of the set $\{(a, b, z) : z \in \mathcal{Z},\, \theta \in \mathcal{H}_{\mathcal{P}},\, \eta \in \mathcal{N}_{\mathcal{P}},\, a = \theta(z),\, b = \eta(z)\}$. Assume that $\mathcal{C}$ is compact, and that $l(a,b,z)$ belongs to the class $C^{2,1}$ in $(a,b)$ on $\mathcal{C}$; that is,
\begin{itemize}
\item[(i)] $l(a,b,z)$ is twice continuously differentiable in $a$ for fixed $(b,z)$, and the Hessian $\partial_a^2 l(a,b,z) \in \mathbb{R}^{d_1 \times d_1}$ is Lipschitz continuous in $(a,b)$ with a Lipschitz constant that holds uniformly over $z$;
\item[(ii)] $l(a,b,z)$ is once continuously differentiable in $b$ for fixed $(a,z)$, and the mixed derivative $\partial_a\partial_b l(a,b,z) \in \mathbb{R}^{d_1 \times d_2}$ is Lipschitz continuous in $(a,b)$ with a Lipschitz constant that holds uniformly over $z$.
\end{itemize}
Moreover, assume that $l$, $\partial_a l$, $\partial_a^2 l$, and $\partial_a \partial_b l$ are jointly continuous in $(a, b, z)$. Then \ref{cond::targetsmoothloss} holds with $\dot{\ell}_{\eta}(\theta,z)(h) = (\partial_a l(a,\eta(z),z)\big|_{a = \theta(z)})^\top h(z)$ and  $\partial^2_\theta L_0(\theta, \eta)(h_1, h_2) = P_0\, \ddot{\ell}_\eta(\theta)(h_1, h_2),$
where $\ddot{\ell}_{\eta}(\theta,z)(h_1, h_2) = (h_1(z))^\top \partial_a^2 l(a,\eta(z),z)\big|_{a = \theta(z)} h_2(z)$, and \ref{cond::nuisancesmooth} holds with $\partial_{\eta} \partial_{\theta} L_0(\theta,\eta)(g, h) = \int  (h(z))^\top \partial_a \partial_b l(a,b,z)\big|_{a = \theta(z)}\big|_{b = \eta(z)} g(z)\, P_0(dz).$
\end{lemma}

In Appendix~\ref{sec:examplesback}, we use Lemma~\ref{lemma:smoothlosssuff} to verify our conditions for the orthogonal logistic loss in Example~\ref{example::splogistic}. To illustrate these conditions, we return to our running example.


\renewcommand{\theexample}{1a}

\begin{example}[continued]
We consider the Riesz loss \(\ell(\theta, z) := \frac{1}{2}\theta(x)^2 - r(z, \theta)\), where \(r : \mathcal{Z} \times \mathcal{H} \to \mathbb{R}\) is linear in its second argument. Assume each \(P \in \mathcal{P}\) is dominated by \(P_0\) and satisfies \(c \|\cdot\|_{L^2(P_0)} \leq \|\cdot\|_{L^2(P)} \leq C \|\cdot\|_{L^2(P_0)}\) for constants \(0 < c < C < \infty\). Let \(\mathcal{H} = L^2(P_0)\), with \(\|\cdot\|_{\mathcal{H}}\) the \(L^2(P_0)\) norm and \(\rho_{\mathcal{H}}(\cdot)\) the \(P_0\)-essential supremum norm. Condition~\ref{cond::uniqueness} holds since the risk $\theta \mapsto E_P[\frac{1}{2}\theta(X)^2 - r(Z, \theta)]$ is strongly convex and continuous on $\mathcal{H}$ with respect to the $\| \cdot \|_{\mathcal{H}}$ norm. Condition~\ref{cond::smoothfunctional} holds with $\dot{\psi}_0 := \psi_0$ if $\theta \mapsto m(Z, \theta)$ is almost surely linear, and $\sup_{\theta \in \mathcal{H}} \frac{|E_0[m(Z, \theta)]|}{\|\theta\|_{\mathcal{H}}} < \infty.$ Alternatively, if $m(Z, \theta) = g\bigl(\theta(Z)\bigr)$ for a differentiable function $g: \mathbb{R} \rightarrow \mathbb{R}$ with Lipschitz continuous derivative $\dot{g}: \mathbb{R} \rightarrow \mathbb{R}$, then \ref{cond::smoothfunctional} holds with $\dot\psi_0(\theta)(h) = E_0\bigl[\dot m_\theta(Z,h)\bigr]$, where $\dot m_\theta(z,h) = \dot g\bigl(\theta(z)\bigr)\,h(z)$ (Section~4.1 of \cite{qiu2021universal}). The loss function $\ell$ satisfies \ref{cond::targetsmoothloss::one} with $\dot{\ell}(\theta): (h, z) \mapsto h(x)\theta(x) - r(z, h)$, and satisfies \ref{cond::targetsmoothloss::two} with $\partial_{\theta}^2 L_0(\theta, \eta)(h_1, h_2) = \langle h_1, h_2 \rangle_{\mathcal{H}}$. Conditions~\ref{cond::nuisancesmooth} and~\ref{cond::orthogonal} hold since the loss $\ell$ does not depend on any nuisance functions. \qedsymbol
\end{example}

\subsection{Functional von Mises expansion and efficient influence function}

In this section, we present a functional von Mises expansion for \(\psi_0(\theta_0)\), and establish the pathwise differentiability of \(\Psi\) along with its efficient influence function \citep{mises1947asymptotic, bickel1993efficient}. These results are central to establishing the asymptotic properties of the autoDML estimators.

Let \(\overline{\mathcal{H}}\) denote the completion of \(\mathcal{H}\) under the norm \(\|\cdot\|_{\mathcal{H}}\). The von Mises expansion in the following theorem involves the Hessian Riesz representer of the linearization of \(\psi_0\). Combined with Conditions~\ref{cond::smoothfunctional} and \ref{cond::PDHessian}, the following condition ensures that the functional derivative \(\dot{\psi}_0(\theta_0)\) is continuous with respect to the Hessian inner product and therefore admits such a representer.

\begin{enumerate}[label=\bf{A\arabic*)}, ref={A\arabic*}, resume=cond1]
   \item \textit{(Hessian norm equivalence).}
   There exists \(\infty > \kappa_1, \kappa_2 >  0\) such that \(\kappa_2 \|h\|_{\mathcal{H}}^2 \ge \partial_\theta^2 L_0(\theta_0, \eta_0)(h, h) \ge \kappa_1 \|h\|_{\mathcal{H}}^2\) for all \(h \in \mathcal{H}\).
    \label{cond::PDHessian}
\end{enumerate}
Condition~\ref{cond::PDHessian} automatically holds if one takes \(\|\cdot\|_{\mathcal{H}}\) to be the Hessian-induced norm \(\{\partial_\theta^2 L_0(\theta_0,\eta_0)(h,h)\}^{1/2}\), provided the latter is positive definite. By the Riesz representation theorem, there exists a unique \(\alpha_0 \in \overline{\mathcal{H}}\) such that
\begin{align}
    \dot{\psi}_0(\theta_0)(h) = \partial_\theta^2 L_0(\theta_0, \eta_0)(\alpha_0, h), \quad \forall h \in \overline{\mathcal{H}}. \label{eqn::rieszrep}
\end{align}
The Hessian Riesz representer \(\alpha_0\) solves the optimization problem
\begin{equation}
   \alpha_0 = \argmin_{\alpha \in \overline{\mathcal{H}}} \left[\partial_\theta^2 L_0(\theta_0, \eta_0)(\alpha, \alpha) - 2\dot{\psi}_0(\theta_0)(\alpha) \right], \label{eqn::rieszopt}
\end{equation}
since, by the representer property, this objective differs from $\argmin_{\alpha \in \overline{\mathcal{H}}} \partial_\theta^2 L_0(\theta_0, \eta_0)(\alpha - \alpha_0, \alpha - \alpha_0)$
by a constant. In many cases, the objective in~\eqref{eqn::rieszopt} corresponds to the risk induced by the loss $(z,\alpha) \mapsto \ddot{\ell}_{\eta_0}(\theta_0)(\alpha,\alpha)(z) - 2\,\dot{m}_{\theta_0}(z,\alpha), $for suitable G\^ateaux derivatives \(\ddot{\ell}_{\eta_0}(\theta_0)\) and \(\dot{m}_{\theta_0}\).


In the following theorem, let $\delta_{\mathrm{lin}} := \mathbf{1}\{\psi_0 \neq \dot{\psi}_0(\theta_0)\}$ and $\delta_{\mathrm{quad}} := \mathbf{1}\{L_0(\cdot,\eta_0)\text{ is not quadratic at }\theta_0\}$. Equivalently, $\delta_{\mathrm{lin}} = 0$ if $\psi_0$ is linear, so that $\psi_0 = \dot{\psi}_0(\theta_0)$, and $\delta_{\mathrm{quad}} = 0$ if
\[
L_0(\theta, \eta_0) - L_0(\theta_0, \eta_0)
=
\partial_{\theta} L_0(\theta_0, \eta_0)(\theta - \theta_0)
+
\frac{1}{2}\partial_{\theta}^2 L_0(\theta_0, \eta_0)(\theta - \theta_0, \theta - \theta_0).
\]

\begin{theorem}[Functional von Mises expansion]
\label{theorem::vonmises}
Suppose Conditions~\ref{cond::uniqueness}--\ref{cond::PDHessian} hold. Then, for each $\eta \in \mathcal{N}_{\mathcal{P}}$, $\theta \in \mathcal{H}_{\mathcal{P}}$, and $\alpha \in \mathcal{H}$ with $\rho_{\mathcal{H}}(\alpha) < M$, the following expansion holds:
\begin{align*}
\psi_0(\theta) - \psi_0(\theta_0)
&=
\int \dot{\ell}_{\eta}(\theta)(\alpha)(z)\,P_0(dz)
\\
&\quad
+ \partial_{\theta}^2 L_0(\theta_0, \eta_0)(\alpha_0 - \alpha, \theta - \theta_0)
\\
&\quad
+ (\delta_{\mathrm{lin}}+\delta_{\mathrm{quad}})
O\!\left(\|\theta - \theta_0\|_{\mathcal{H}}^2\right)
\\
&\quad
+ O\!\left(\|\eta - \eta_0\|_{\mathcal{N}} \|\theta - \theta_0\|_{\mathcal{H}}\right)
+ O\!\left(\|\eta - \eta_0\|_{\mathcal{N}}^2\right),
\end{align*}
where the big-$O$ notation depends only on $M$ and the universal constants appearing in our conditions.
\end{theorem}

Given an estimator \(\theta_n\) of \(\theta_0\), a consequence of Theorem~\ref{theorem::vonmises} is that the leading term in the plug-in error \(\psi_0(\theta_n) - \psi_0(\theta_0)\) is \(P_0 \dot{\ell}_{\eta_0}(\alpha_0)\). Moreover, given estimators \(\eta_n\) of \(\eta_0\) and \(\alpha_n\) of \(\alpha_0\), this term is well-approximated by \(P_0 \dot{\ell}_{\eta_n}(\alpha_n)\), up to second-order remainder terms. This characterization suggests a natural strategy for debiasing the plug-in estimator \(\psi_n(\theta_n)\): either explicitly, by adding the empirical correction term \(- P_n \dot{\ell}_{\eta_n}(\theta_n)(\alpha_n)\), or implicitly, by constructing the estimator to satisfy \(P_n \dot{\ell}_{\eta_n}(\theta_n)(\alpha_n) = 0\).   In the next section, we leverage these properties to construct autoDML estimators for \(\Psi(P_0)\). 



 If the functional is linear $(\delta_{\mathrm{lin}} = 0)$ and the risk function is quadratic $(\delta_{\mathrm{quad}} = 0)$, Theorem \ref{theorem::vonmises} only involves the doubly robust (or mixed bias) remainder terms: $\partial_{\theta}^2 L_0(\theta_0, \eta_0)(\alpha_0 - \alpha, \theta - \theta_0)$ and $\|\eta - \eta_0\|_{\mathcal{H}} \|\theta - \theta_0\|_{\mathcal{H}}$. The latter remainder term can be dropped entirely if the loss $\ell_{\eta}$ is universally Neyman-orthogonal, meaning that \(\partial_\eta \partial_{\theta}^2 L_0(\theta, \eta_0) = 0\) for all \(\theta \in \mathcal{H}\), and the map \(\eta \mapsto \partial_\eta \partial_{\theta}^2 L_0(\theta_0, \eta)\) is suitably Lipschitz (see Theorem~\ref{theorem::neymanorthogonality_secondorder} in Appendix~\ref{appendix::vonmises}). This finding extends the double robustness of the von Mises expansion for linear functionals of regression functions—corresponding to the squared error loss—to linear functionals of general pseudo-outcome regressions (Example~\ref{example::orthoLS}). For such functionals, we also obtain a dual identification of the estimand in terms of the Riesz representer: $\Psi(P_0) = -\partial_{\theta} L_0(0, \eta_0)(\alpha_0) = -E_0[\dot{\ell}_{\eta_0}(0)(\alpha_0)(Z)].$


The Hessian Riesz representer $\alpha_0$, viewed as the gradient of $\dot{\psi}_0(\theta_0)$ with respect to the Hessian inner product, identifies the local fluctuation direction that maximizes the squared change in the functional $\psi_0$ at $\theta_0$ per unit increase in loss. Specifically, by the Cauchy--Schwarz inequality, $\alpha_0$ maximizes the following ratio over all directions $h \in \overline{\mathcal{H}}$:
$$
\lim_{\varepsilon \to 0}
\frac{|\psi_0(\theta_0 + \varepsilon h) - \psi_0(\theta_0)|^2}
{L_0(\theta_0 + \varepsilon h, \eta_0) - L_0(\theta_0, \eta_0)}
=
\frac{|\dot{\psi}_0(\theta_0)(h)|^2}
{\frac{1}{2}\partial_\theta^2 L_0(\theta_0, \eta_0)(h, h)}.
$$
For the log-likelihood loss, the right-hand side corresponds to the Cram\'er--Rao bound for the one-dimensional submodel $\varepsilon \mapsto \theta_0 + \varepsilon h$. The submodel $\varepsilon \mapsto \theta_0 + \varepsilon \alpha_0$ may therefore be viewed as a least favorable submodel through $\theta_0$ for the functional $\psi_0$ relative to the loss function $\ell_{\eta_0}$. In particular, the plug-in M-estimator of $\Psi(P_0)$ based on this submodel attains the largest asymptotic variance among all one-dimensional parametric submodels through $\theta_0$.


This least favorable construction based on the Hessian Riesz representer is analogous to the least favorable density submodel through $P_0$ for a pathwise differentiable parameter, given by $\varepsilon \mapsto (1 + \varepsilon \chi_0)\,dP_0$, where $\chi_0$ is the efficient influence function \citep{stein1956efficient,bickel1993efficient}. That submodel plays a central role in efficiency theory and underlies targeted maximum likelihood estimation \citep{laan_rubin_2006,van2016one}. It is therefore natural that the Hessian Riesz representer $\alpha_0$ appears in the efficient influence function of the projection parameter $\Psi$, as the following theorem shows.

We make use of the following conditions. For each \(P \in \mathcal{P}\), let \(\alpha_P \in \overline{\mathcal{H}}\) denote the Hessian Riesz representer under $P$, satisfying $\dot{\psi}_P(\theta_P)(h) = \partial_{\theta}^2 L_P(\theta_P, \eta_P)(\alpha_P, h), \quad \forall h \in \overline{\mathcal{H}}.$

\begin{enumerate}[label=\bf{A\arabic*)}, ref={A\arabic*}, resume=cond1]
  
  \item \textit{(Functional continuity):} $P \mapsto \Psi(P)$, $P \mapsto \theta_P$ and $P \mapsto \eta_P$ are Lipschitz continuous maps at $P_0$ with domain $\mathcal{P}$ equipped with the Hellinger distance and the range equipped with $| \cdot |$, $\| \cdot\|_{\mathcal{H}}$ and $\| \cdot\|_{\mathcal{N}}$, respectively. In addition, $(\eta, \theta) \mapsto \dot{\ell}_{\eta}(\cdot, \theta)$ and $(\eta, \theta) \mapsto m(\cdot, \theta)$ are continuous maps with domain $\mathcal{N} \times \mathcal{H}$ equipped with the total norm $(\eta, \theta) \mapsto  \| \eta\|_{\mathcal{N}} +  \|\theta\|_{\mathcal{H}}$ and range $L^2(P_0)$. \label{cond::smoothparam}
    \item   \textit{(Boundedness):}   \label{cond::boundedtrue}  $\sup_{P \in \mathcal{P}}\rho_{\mathcal{H}}(\alpha_P) < \infty$.    
\end{enumerate}
Condition \ref{cond::smoothparam} ensures that the target and nuisance parameterizations $\{\theta_P: P \in \mathcal{P}\}$ and $\{\eta_P: P \in \mathcal{P}\}$ are Lipschitz Hellinger-continuous. \cite{luedtke2023one} and \cite{luedtke2024simplifying} verify Lipschitz Hellinger continuity for several functions of interest, including regression functions and causal contrasts.  

\begin{theorem}[Pathwise differentiability]
    \label{theorem::EIF}
    If \ref{cond::uniqueness}-\ref{cond::boundedtrue}, then $\Psi:\mathcal{P} \rightarrow \mathbb{R}$ is pathwise differentiable at $P_0$ with efficient influence function $\chi_0(z)= -\dot{\ell}_{\eta_0}(\theta_0)(\alpha_0)(z) + m(z, \theta_0) - \Psi(P_0).$  
\end{theorem}

Theorem~\ref{theorem::EIF} shows that the EIF is characterized by the score (i.e., the derivative) of the Neyman-orthogonal loss \(\ell_{\eta_0}\) in the direction of the Hessian Riesz representer \(\alpha_0\), given by \(\dot{\ell}_{\eta_0}(\theta_0)(\alpha_0)\). As a consequence, it addresses an open question implicit in the targeted minimum loss-based estimation procedure \citep{van2011targeted}: given a loss function for the M-estimand \(\theta_0\) and a functional \(\psi_0\), does there exist a submodel through \(\theta_0\) whose score, with respect to this loss, generates the EIF of \(\psi_0\)? The theorem confirms that such a submodel exists and is given by the least favorable path \(\varepsilon \mapsto \theta_0 + \varepsilon \alpha_0\), provided the loss is Neyman-orthogonal. This result enables the development of TML estimators that directly ``target” the M-estimand for any orthogonal loss by performing minimum loss-based estimation along the least favorable submodel. In this framework, the Hessian representer \(\alpha_0\) assumes the role of the efficient influence function \(\chi_0\) in targeted maximum likelihood estimation, determining the least favorable submodel through \(\theta_0\).

When $\mathcal{H}$ is nonparametric, different loss functions may identify the same parameter $\Psi$, and Theorem~\ref{theorem::EIF} then yields the same EIF. For example, with binary outcomes, the outcome regression may be identified using working log-likelihood losses such as squared error, logistic cross-entropy, or Poisson loss. Although the EIF is unchanged, the Riesz representer $\alpha_0$ and the functional derivative $-\dot{\ell}_{\eta_0}(\theta_0)$ will typically differ. In the context of TMLE, each orthogonal loss function that identifies the target parameter gives rise to its own least favorable submodel through $\theta_0$ and its own targeted estimation procedure. When \(\mathcal{H}\) is semiparametric, different loss functions induce different nonparametric extensions of the parameter \(\Psi\) defined on the semiparametric model. These extensions may yield different EIFs and efficiency bounds. When \(\mathcal{H}\) is correctly specified for the unconstrained M-estimand, the most efficient extension typically corresponds to the log-likelihood loss (see Section~3 of \cite{van2023adaptive}). In that case, the Hessian Riesz representer typically coincides with the EIF in the restricted model \(\{P : \theta_{P,\mathrm{np}} \in \mathcal{H}\}\), where \(\theta_{P,\mathrm{np}}\) denotes the unconstrained M-estimand.

\renewcommand{\theexample}{1a}

\begin{example}[continued]
   Condition~\ref{cond::PDHessian} holds since $\partial_{\theta}^2 L_0(\theta, \eta)(h, h) = \|h\|_{\mathcal{H}}^2$.   For a linear functional $\psi_0$ of the Riesz representer $\theta_0$, the EIF term $-\dot{\ell}_{\eta_0}(\theta_0)(\alpha_0)$ equals $z \mapsto r(z, \alpha_0) - \alpha_0(x)\theta_0(x)$, where  $\alpha_0 := \argmin_{\alpha \in \mathcal{H}}  E_0[\{\alpha(X)\}^2 - 2\,m(Z, \alpha)]$. When $\theta_0$ is the outcome regression, we have $-\dot{\ell}_{\eta_0}(\theta_0)(\alpha_0)(z) =  \alpha_0(x)\{y - \theta_0(x)\}$. The von Mises expansion is doubly robust, satisfying $\psi_0(\theta) - \psi_0(\theta_0) - P_0 \dot{\ell}_{\eta}(\theta)(\alpha) =  \langle \theta - \theta_0, \alpha_0 - \alpha\rangle_{L^2(P_0)}$.  \qedsymbol
\end{example}

 \section{Automatic debiased machine learning}
\label{sec::autodml}
\subsection{Proposed estimators}

\label{sec::autodmlest}

In this section, we propose three autoDML estimators based on one-step estimation, targeted minimum loss-based estimation (TMLE), and the method of sieves. All approaches require an estimator \(\eta_n\) of the nuisance parameter \(\eta_0\). The first two methods additionally require initial estimators \(\theta_n, \alpha_n \in \mathcal{H}\) of the M-estimand \(\theta_0\) and the Hessian Riesz representer \(\alpha_0\) defined in \eqref{eqn::rieszopt}. To ensure strong theoretical guarantees under weak conditions, we recommend cross-fitting these nuisance estimators. Cross-fitting is a standard technique for mitigating overfitting and weakening complexity assumptions in nuisance estimation \citep{van2011cross, DoubleML}. Cross-fitted versions of our estimators appear in Section~\ref{sec::crossfit}.


Our first autoDML estimator of \(\Psi(P_0)\) is the one-step debiased estimator \citep{bickel1993efficient} 
\begin{equation}
    \widehat{\psi}_n^{\mathrm{dml}} := \frac{1}{n}\sum_{i=1}^n m(Z_i, \theta_n) - \frac{1}{n}\sum_{i=1}^n \dot{\ell}_{\eta_n}(\theta_n)(\alpha_n)(Z_i). \label{eqn::onestep}
\end{equation}
The second term provides an influence-function-based bias correction that, by Theorem~\ref{theorem::EIF}, removes the first-order bias \(P_0 \dot{\ell}_{\eta_n}(\theta_n)(\alpha_n)\) of the plug-in estimator. This estimator is automatic in the sense that its construction requires only the specification of the target functional \(m\) and the loss \(\ell_{\eta_n}\). The nuisance components \(\theta_0\) and \(\alpha_0\) are then estimated directly from \(\ell_{\eta_n}\) and its derivatives, as discussed in Section~\ref{sec::overview}.

A limitation of the one-step autoDML estimator \(\widehat{\psi}_n^{\mathrm{dml}}\) is that it is not a plug-in estimator; that is, there may not exist \(\theta_n^* \in \mathcal{H}\) such that \(\widehat{\psi}_n^{\mathrm{dml}} = \psi_n(\theta_n^*)\), where $\psi_n(\cdot) = \frac{1}{n}\sum_{i=1}^n m(Z_i, \cdot)$. The plug-in property ensures that estimates respect model constraints, such as probability constraints that require values to lie in \([0,1]\) and bounds on the outcome. TMLE provides a general framework for constructing efficient plug-in estimators by refining nuisance estimates to reduce bias in the target parameter \citep{laan_rubin_2006, vanderLaanRose2011}. Our proposed \emph{autoTML} estimator is given by $\widehat{\psi}_n^{\mathrm{tmle}} := \frac{1}{n} \sum_{i=1}^n m(Z_i, \theta_n^*),$
where \(\theta_n^*\) is a \emph{targeted} estimator of \(\theta_0\) constructed to solve the parameter-specific efficient score equation:
\begin{equation}
\frac{1}{n} \sum_{i=1}^n \dot{\ell}_{\eta_n}(\theta_n^*)(\alpha_n)(Z_i) = 0. \label{eqn::tmlescore}
\end{equation}
This property ensures that the influence function-based bias correction of the one-step estimator vanishes at \((\eta_n, \theta_n^*, \alpha_n)\), yielding a debiased estimator without requiring an explicit adjustment. Given an initial estimate \(\theta_n\), we construct \(\theta_n^*\) by updating along the least favorable submodel:
\[
\theta_n^* := \theta_n + \varepsilon_n \alpha_n,
\quad \text{where} \quad
\varepsilon_n = \argmin_{\varepsilon} \sum_{i=1}^n \ell_{\eta_n}(Z_i, \theta_n + \varepsilon \alpha_n),
\]
where the first-order optimality condition for \(\varepsilon_n\) is precisely \eqref{eqn::tmlescore}. Intuitively, autoTML performs a functional gradient descent step from \(\theta_n\), moving in the direction of the estimated Hessian representer \(\alpha_n\). The direction \(\alpha_n\) is an estimate of the gradient of \(\psi_0\), ensuring that the update proceeds in the most economical direction—maximizing change in \(\psi_0\) per unit decrease in risk.

To illustrate why autoTML can improve upon autoDML, we derive an approximate closed-form expression for \(\widehat{\psi}_n^{\mathrm{tmle}}\). This approximation is obtained by performing a single Newton–Raphson update from \(\theta_n\) and is asymptotically valid under the consistency of \(\theta_n\) for \(\theta_0\). To this end, assume that the loss \(\ell_{\eta}(\theta)\) admits a second derivative \(\ddot{\ell}_{\eta}(\theta)(\cdot, \cdot)\) and that \(m(z, \theta)\) has a derivative \(\dot{m}_{\theta}(z, \cdot)\). For small \(\varepsilon\), the loss difference \(\ell_{\eta_n}(\cdot, \theta_n + \varepsilon \alpha_n) - \ell_{\eta_n}(\cdot, \theta_n)\) admits the quadratic approximation \(\varepsilon \dot{\ell}_{\eta_n}(\theta_n)(\alpha_n) + \frac{1}{2} \varepsilon^2 \ddot{\ell}_{\eta_n}(\theta_n)(\alpha_n, \alpha_n)\), yielding the first-order update
\[
\theta_n^* \approx \theta_n - \frac{\sum_{i=1}^n \dot{\ell}_{\eta_n}(\theta_n)(\alpha_n)(Z_i)}{\sum_{i=1}^n \ddot{\ell}_{\eta_n}(\theta_n)(\alpha_n, \alpha_n)(Z_i)} \alpha_n.
\]
Applying a first-order Taylor expansion of \(m\) gives
\begin{equation}
\widehat{\psi}_n^{\mathrm{tmle}} \approx \frac{1}{n} \sum_{i=1}^n m(Z_i, \theta_n) - \frac{\sum_{i=1}^n \dot{m}_{\theta_n}(Z_i, \alpha_n)}{\sum_{i=1}^n \ddot{\ell}_{\eta_n}(\theta_n)(\alpha_n, \alpha_n)(Z_i)} \cdot \frac{1}{n} \sum_{i=1}^n \dot{\ell}_{\eta_n}(\theta_n)(\alpha_n)(Z_i).
\label{eqn::stableautoDML}
\end{equation}
Under this approximation, autoTML improves upon the one-step estimator by introducing the normalization factor \(\frac{\sum_{i=1}^n \dot{m}_{\theta_n}(Z_i, \alpha_n)}{\sum_{i=1}^n \ddot{\ell}_{\eta_n}(\theta_n)(\alpha_n, \alpha_n)(Z_i)}\). By the representation property, the population counterpart of this factor—obtained by replacing \(P_n\) with \(P_0\) and \(\alpha_n\) with \(\alpha_0\)—equals 1, thereby recovering the usual one-step update. Hence, it effectively stabilizes \(\alpha_n\) and serves as a calibrated step size for the one-step update. This expression is exact when the risk is quadratic and the functional is linear, as illustrated in the following example.

 \renewcommand{\theexample}{1a}

\begin{example}[continued]
For a linear functional of a Riesz representer \( \theta_0 \), an autoDML estimator is given by $\widehat{\psi}_n^{\mathrm{dml}} := \frac{1}{n} \sum_{i=1}^n m(Z_i, \theta_n) + \frac{1}{n} \sum_{i=1}^n \left\{ r(Z_i, \alpha_n(X_i)) - \theta_n(X_i)\alpha_n(X_i) \right\},$ where, in the special case where \( \theta_0 \) is the outcome regression, this reduces to $\widehat{\psi}_n^{\mathrm{dml}} := \frac{1}{n} \sum_{i=1}^n m(Z_i, \theta_n) + \frac{1}{n} \sum_{i=1}^n \alpha_n(X_i) \{ Y_i - \theta_n(X_i) \}.$ The fluctuation parameter for autoTML is given by $\varepsilon_n = \argmin_{\varepsilon} \sum_{i=1}^n \{\theta_n(X_i) + \varepsilon \alpha_n(X_i)\}^2 - 2 \varepsilon r(Z_i, \alpha_n(X_i)),$ and the autoTML estimator $\widehat{\psi}_n^{\mathrm{tmle}}$ has the closed-form expression: $\frac{1}{n} \sum_{i=1}^n m(Z_i, \theta_n) + \delta_n^* \cdot \frac{1}{n} \sum_{i=1}^n \{r(Z_i, \alpha_n(X_i)) - \theta_n(X_i)\alpha_n(X_i)\}$, where $\delta_n^* := \frac{\sum_{i=1}^n m(Z_i, \alpha_n)}{\sum_{i=1}^n (\alpha_n(X_i))^2}$ is a normalization factor. \qed
\end{example}


An alternative approach to constructing debiased plug-in estimators is the method of sieves \citep{shen1997methods, spnpsieve, sieveOneStepPlugin, SieveQiu, discussionMLEMark, undersmoothedHAL}, a classical technique for efficient estimation in semiparametric models. For a sequence \(k(n) \to \infty\), this method approximates \(\mathcal{H}\) by the finite-dimensional subspace \(\mathcal{H}_{k(n)}\), taken from a sieve—a nested sequence \(\mathcal{H}_1 \subset \mathcal{H}_2 \subseteq \dots \subseteq \mathcal{H}\) of finite-dimensional subspaces whose union is dense in \(\mathcal{H}\). Given a sieve, our autoSieve estimator is defined as the plug-in estimator \(\widehat{\psi}_n^{\mathrm{sieve}} := \frac{1}{n} \sum_{i=1}^n m(Z_i, \theta_{n,k(n)})\), where, for each \( k \in \mathbb{N} \),
\[
\theta_{n,k} := \argmin_{\theta \in \mathcal{H}_{k}} \sum_{i=1}^n \ell_{\eta_n}(Z_i, \theta),
\]
and the sieve dimension \( k(n) \) is selected automatically as described below.

By the first-order optimality conditions for empirical risk minimization, \(\theta_{n,k(n)}\) satisfies the score equations
\[
\frac{1}{n} \sum_{i=1}^n \dot{\ell}_{\eta_n}(\theta_{n,k(n)})(\alpha)(Z_i) = 0
\qquad \text{for all } \alpha \in \mathcal{H}_{k(n)}.
\]
In particular, taking \(\alpha\) to be the projection of \(\alpha_0\) onto \(\mathcal{H}_{k(n)}\) with respect to the Hessian inner product yields
\[
\frac{1}{n} \sum_{i=1}^n \dot{\ell}_{\eta_n}(\theta_{n,k(n)})(\alpha_{0,k(n)})(Z_i) = 0,
\qquad
\alpha_{0,k} := \arg\min_{\alpha \in \mathcal{H}_k}
\partial_{\theta}^2 L_0(\theta_0,\eta_0)(\alpha - \alpha_0, \alpha - \alpha_0).
\]
It follows that the autoSieve estimator \(\widehat{\psi}_n^{\mathrm{sieve}}\) coincides with the one-step estimator constructed using the nuisance estimates \((\eta_n, \theta_{n,k(n)}, \alpha_{0,k(n)})\), and is therefore debiased without any additional correction. To ensure that the asymptotic linearity of \(\widehat{\psi}_n^{\mathrm{sieve}}\) is not distorted by plug-in bias, the sieve dimension \(k(n)\) must increase sufficiently quickly so that both \(\theta_{n,k(n)}\) and \(\alpha_{0,k(n)}\) converge to \(\theta_0\) and \(\alpha_0\), respectively, at suitable rates. This typically requires \emph{undersmoothing}, in the sense that \(k(n)\) grows faster than would be optimal for estimating \(\theta_0\) alone. In practice, however, there is little guidance on how much undersmoothing is needed. Our proposed autoSieve procedure is designed to address this issue.

We now describe the automatic undersmoothing procedure used to select the sieve dimension \(k(n)\). Given an independent validation sample \(\{Z_i'\}_{i \in [n']}\), the sieve dimension is defined as \(k(n) := \max\{k_\theta(n), k_{\alpha}(n)\}\), where \(k_\theta(n)\) and \(k_{\alpha}(n)\) are data-adaptive sequences given by
\begin{align*}
    k_\theta(n) &:= \arg\min_{k \in \mathbb{N}} \sum_{i=1}^{n'} \ell_{\eta_n}(Z_i', \theta_{n,k}), \hspace{2em}
    k_{\alpha}(n) := \arg\min_{k \in \mathbb{N}} \sum_{i=1}^{n'} \left\{ \ddot{\ell}_{\eta_n}(\theta_{n,k_\theta(n)})(\alpha_{n,k}, \alpha_{n,k})(Z_i') - 2\,\dot{m}_{\theta_{n,k_\theta(n)}}(Z_i', \alpha_{n,k}) \right\},
\end{align*}
where, for each \(k \in \mathbb{N}\), \(\alpha_{n,k} := \arg\min_{\alpha \in \mathcal{H}_k} \frac{1}{n} \sum_{i=1}^{n} \left\{ \ddot{\ell}_{\eta_n}(\theta_{n,k_\theta(n)})(\alpha, \alpha)(Z_i) - 2\,\dot{m}_{\theta_n}(Z_i, \alpha) \right\}\) is an estimate of the projected representer \(\alpha_{0,k}\). The sieve dimension \(k_\theta(n)\) optimizes estimation of the \emph{M}-estimand \(\theta_0\), while \(k_\alpha(n)\) optimizes estimation of the Hessian representer \(\alpha_0\). By setting \(k(n) = \max\{k_\theta(n), k_\alpha(n)\}\), the sieve estimator is automatically undersmoothed and ensures that \(\mathcal{H}_{k(n)} = \mathcal{H}_{k_\theta(n)} \oplus \mathcal{H}_{k_\alpha(n)}\) approximates both \(\theta_0\) and \(\alpha_0\) well. For example, if \(\alpha_0\) is less smooth than \(\theta_0\), the sieve space \(\mathcal{H}_{k_\theta(n)}\) may poorly approximate \(\alpha_0\), inducing bias. In such cases, we typically have \(k(n) = k_\alpha(n)\) with high probability, ensuring that \(\mathcal{H}_{k(n)}\) captures the complexity of \(\alpha_0\) \citep{qiu2021universal}. Although we use an external sample, the sieve dimensions could also be selected from the same data using sample-splitting or cross-validation.

\subsection{Cross-fitting and algorithms for autoDML}
\label{sec::crossfit}

We now describe how cross-fitting can be incorporated into the autoDML framework to mitigate overfitting and relax constraints on the complexity of the nuisance estimators \citep{van2011cross, DoubleML}. A key challenge is that both \(\theta_0\) and \(\alpha_0\) depend on other nuisance functions, which must also be cross-fitted. In Algorithm~\ref{alg:crossfit}, we propose an efficient procedure that requires each nuisance function to be cross-fitted only once. To simplify estimation of the Riesz representer, we assume \(\partial_\theta^2 L_0(\theta_0, \eta_0)(\alpha, \alpha) = P_0\, \ddot{\ell}_{\eta_0}(\theta_0)(\alpha, \alpha)\) and \(\dot{\psi}_0(\theta_0)(\alpha) = P_0\, \dot{m}_{\theta_0}(\cdot, \alpha)\), for some functions \(\ddot{\ell}_{\eta_0}(\theta_0)\) and \(\dot{m}_{\theta_0}\), so that we can use the loss function \((z, \alpha) \mapsto \ddot{\ell}_{\eta_0}(\theta_0)(\alpha, \alpha)(z) - 2\, \dot{m}_{\theta_0}(z, \alpha)\).

\begin{algorithm}[!htb]
\begin{algorithmic}[1]
{\small
\caption{Cross-fit nuisance estimation for autoDML}  \label{alg:crossfit}
 
 \vspace{.1in}
\INPUT dataset $\mathcal{D}_n = \{Z_i: i=1,\ldots,n\}$, nuisance estimation algorithm $\mathcal{A}_{\eta}$, number $J$ of cross-fitting splits
\newline \hspace{-0.5em} \texttt{\# Compute data splits.}
\STATE partition $\mathcal{D}_n$ into datasets $\mathcal{T}^{(1)},\mathcal{T}^{(2)},\ldots,\mathcal{T}^{(J)}$;
\STATE for {$s = 1,\ldots,J$}, set $j(i):=s$ for each $i\in \mathcal{T}^{(s)}$ and $\mathcal{I}^{(s)} := \{i \in [n]: Z_i \in  \mathcal{T}^{(s)}\}$;
\newline \hspace{-0.5em} \texttt{\# Cross-fit nuisance function}
\FOR {$s = 1,\ldots,J$}
\STATE get estimator $\eta_{n,s} := \mathcal{A}_{\eta}( \mathcal{D}_n \backslash \mathcal{T}^{(s)})$ of $\eta_0$ using $\mathcal{A}_{\eta}$ from $ \mathcal{D}_n \backslash \mathcal{T}^{(s)}$;
\ENDFOR
\newline \hspace{-0.5em}  \texttt{\# Cross-fit M-estimand.}
\FOR {$s = 1,\ldots,J$}
\STATE {get estimator $\theta_{n,s}$ of $\theta_0$ from $\mathcal{D}_n \backslash \mathcal{T}^{(s)}$ using the cross-fitted orthogonal risk:
    \mbox{\footnotesize $\theta \mapsto \sum_{i \in [n] \backslash \mathcal{I}^{(s)}} \ell_{\eta_{n,j(i)}}(Z_i, \theta);$}}
\ENDFOR
\newline \hspace{-0.5em} \texttt{\# Cross-fit Riesz representer}
\FOR {$s = 1,\ldots,J$}
\STATE get estimator $\alpha_{n,s}$ of $\alpha_0$ from $\mathcal{D}_n \backslash \mathcal{T}^{(s)}$ using the cross-fitted Riesz risk:
    $$\alpha \mapsto \sum_{i \in [n] \backslash \mathcal{I}^{(s)}} \left\{ \ddot{\ell}_{\eta_{n,j(i)}}(\theta_{n,j(i)})(Z_i, \alpha) - 2\,\dot{m}_{\theta_{n,j(i)}}(Z_i, \alpha) \right\};$$
        \vspace{-0.5cm}
\ENDFOR
\RETURN Cross-fit estimators $\{\eta_{n,j},\theta_{n,j},  \alpha_{n,j} : 1 \leq j \leq J\}$
}
\end{algorithmic}
\vspace{.05in}
\end{algorithm}

Algorithm~\ref{alg:crossfit} requires constructing cross-fitted estimators for \(\eta_0\), \(\theta_0\), and \(\alpha_0\) only once, with each depending on previously cross-fitted nuisances. However, since the cross-fitted nuisances \(\{\eta_{n,j}: 1 \leq j \leq J\}\) used in the loss functions for \(\theta_0\) and \(\alpha_0\) are constructed from the full dataset, some data leakage occurs across training folds. Additional leakage arises when estimating \(\alpha_0\), as its objective also involves \(\{\theta_{n,j}: 1 \leq j \leq J\}\), likewise based on the full dataset. While this leakage could be avoided by re-cross-fitting nuisances within each fold, it would significantly increase computational cost. Similar leakage appears in Algorithm~3 of \cite{van2024combining} in the context of cross-validation with nuisances, and empirically, it did not affect estimator performance.

With cross-fitted nuisance estimates computed via Algorithm~\ref{alg:crossfit}, we introduce our cross-fitted (stabilized) autoDML and autoTML estimators in Algorithms~\ref{alg:autoDML} and~\ref{alg:autoTML}, respectively. Both follow the debiasing strategy in Section~\ref{sec::autodmlest}, differing only in that pooled out-of-fold estimates are used in the debiasing step.

\begin{algorithm}[!htb]
\begin{algorithmic}[1]
{\small
\caption{autoDML: automatic debiased machine learning with cross-fitting} \label{alg:autoDML}
 \vspace{.1in}
\INPUT dataset $\mathcal{D}_n = \{Z_i: i=1,\ldots,n\}$, nuisance estimation algorithm $\mathcal{A}_{\eta}$, number $J$ of cross-fitting splits
\newline \hspace{-0.5em} \texttt{\# Cross-fit nuisances.}
\STATE obtain cross-fitted estimators $\{\eta_{n,j}, \theta_{n,j},  \alpha_{n,j} : 1 \leq j \leq J\}$ from Alg. \ref{alg:crossfit} with $\mathcal{D}_n$, $\mathcal{A}_{\eta}$, and $J$;
\newline \hspace{-0.5em} \texttt{\# Optional TMLE-inspired stabilization via \eqref{eqn::stableautoDML}}
\STATE for {$s = 1,\ldots,J$}, set $\alpha_{n,s} := \varepsilon_n \alpha_{n,s}$ with $\varepsilon_n := \frac{\sum_{i=1}^n \dot{m}_{\theta_{n,j(i)}}(Z_i, \alpha_{n,j(i)})}{\sum_{i=1}^n \ddot{\ell}_{\eta_{n,j(i)}}(\theta_{n,j(i)})(\alpha_{n,j(i)}, \alpha_{n,j(i)})(Z_i)}$.
\newline \hspace{-0.5em} \texttt{\# Compute one-step estimator}
\STATE set $\widehat{\psi}_n^{\mathrm{dml}} := \frac{1}{n}\sum_{=1}^n m(Z_i, \theta_{n, j(i)}) -    \frac{1}{n}\sum_{=1}^n \dot{\ell}_{n,j(i)}(\theta_{n, j(i)})(\alpha_{n,j(i)})(Z_i)$
\RETURN autoDML estimate $\widehat{\psi}_n^{\mathrm{dml}}$
}
\end{algorithmic}
\vspace{.05in}
\end{algorithm}

 \begin{algorithm}[!htb]
\begin{algorithmic}[1]
{\small
\caption{autoTML: automatic targeted machine learning with cross-fitting} \label{alg:autoTML}
 \vspace{.1in}
\INPUT dataset $\mathcal{D}_n = \{Z_i: i=1,\ldots,n\}$, nuisance estimation algorithm $\mathcal{A}_{\eta}$, number $J$ of cross-fitting splits
\newline \hspace{-0.5em} \texttt{\# Cross-fit nuisances.}
\STATE obtain cross-fitted estimators $\{\eta_{n,j}, \theta_{n,j},  \alpha_{n,j} : 1 \leq j \leq J\}$ from Alg. \ref{alg:crossfit} with $\mathcal{D}_n$, $\mathcal{A}_{\eta}$, and $J$;
\newline \hspace{-0.5em} \texttt{\# Targeting step.}
\STATE compute fluctuation $\varepsilon_n := \argmin_{\varepsilon}   \sum_{i=1}^n \ell_{\eta_{n,j(i)}}(Z_i, \theta_{n,j(i)} + \varepsilon \alpha_{n,j(i)})$
\STATE for {$s = 1,\ldots,J$}, set $\theta_{n,s}^* := \theta_{n,s} +\varepsilon_n \alpha_{n,s} $;
\newline \hspace{-0.5em} \texttt{\# Compute plug-in estimator.}
\STATE set $\widehat{\psi}_n^{\mathrm{tmle}} := \frac{1}{n}\sum_{=1}^n m(Z_i, \theta_{n, j(i)}^*)$
\RETURN autoTML estimate $\widehat{\psi}_n^{\mathrm{tmle}}$
}
\end{algorithmic}
\vspace{.05in}
\end{algorithm}

\section{Theoretical results}
\label{sec::theory}
\subsection{Asymptotic theory for autoDML}

\label{sec::theoryautodml}

In this section, we establish that the proposed autoDML estimators are regular, asymptotically linear, and semiparametrically efficient for \(\Psi(P_0)\). We first present a general result establishing these properties for the autoDML estimator \(\widehat{\psi}_n^{\mathrm{dml}}\). Building on this result, we derive analogous guarantees for the autoTML and autoSieve estimators, \(\widehat{\psi}_n^{\mathrm{tmle}}\) and \(\widehat{\psi}_n^{\mathrm{sieve}}\).

Denote the EIF estimate corresponding to $\theta_n$, $\eta_n$, and $\alpha_n$ by $\chi_n := m(\cdot, \theta_n) - P_n m(\cdot, \theta_n) - \dot{\ell}_{\eta_n}(\theta_n)(\alpha_n)$. Our main theorem, which we now present, will make use of the following conditions.

\begin{enumerate}[label=\bf{B\arabic*)}, ref={B\arabic*}, series=cond1]
\item \textit{Bounded estimators:} \label{cond::boundnuis} With probability tending to one, $\eta_n \in \mathcal{N}_{\mathcal{P}}$, $\theta_n \in \mathcal{H}_{\mathcal{P}}$, and $\rho_{\mathcal{H}}(\alpha_n) = O_p(1)$.
  \item \textit{Nuisance estimation rate:}  \label{cond::nuisancerate}  $\|\eta_n - \eta_0\|_{\mathcal{N}} = o_p(n^{-\frac{1}{4}})$.
\item \textit{M-estimation rate:} \label{cond::targetrate} One of the following holds:
\begin{enumerate}
    \item[(i)] \(\psi_0\) is linear and the risk is quadratic \((\delta_{\mathrm{lin}} = \delta_{\mathrm{quad}} = 0)\);
    \item[(ii)] \(\|\theta_n - \theta_0\|_{\mathcal{H}} = o_p(n^{-1/4})\).
\end{enumerate}
\item \textit{Doubly robust rate:}  \label{cond::DRrate}   $\partial_{\theta}^2 L_0(\theta_0, \eta_0)(\alpha_0 - \alpha_n,\theta_n - \theta_0) + \|\theta_n - \theta_0\|_{\mathcal{H}}\|\eta_n - \eta_0\|_{\mathcal{N}}= o_p(n^{-1/2})$.
  \item \textit{Empirical process condition:} \label{cond::empproc} $ (P_n-P_0)(\chi_n - \chi_0) = o_p(n^{-1/2})$
\end{enumerate}

\begin{theorem}
\label{theorem::limitautoDML}
     Assume Conditions \ref{cond::uniqueness}-\ref{cond::boundedtrue} and \ref{cond::boundnuis}-\ref{cond::empproc}. Then, the autoDML estimator   $\widehat{\psi}_n^{\mathrm{dml}}$ satisfies the asymptotically linear expansion: $ \widehat{\psi}_n^{\mathrm{dml}} - \Psi(P_0) = P_n \chi_0 + o_p(n^{-1/2}).$
    Consequently, $\widehat{\psi}_n^{\mathrm{dml}}$ is a regular and efficient estimator of $\Psi(P_0)$, and $n^{1/2}\{\widehat{\psi}_n^{\mathrm{dml}} - \Psi(P_0)\} \rightarrow_d N(0, \operatorname{var}_0(\chi_0(Z))).$
\end{theorem}

For linear functionals and quadratic risk functions (\( \delta_{\mathrm{lin}} = \delta_{\mathrm{quad}} = 0 \)), autoDML estimators are \textit{doubly robust}: they are consistent if either \( \|\alpha_n - \alpha_0\|_{\mathcal{H}} = o_p(1) \) or \( \|\theta_n - \theta_0\|_{\mathcal{H}} = o_p(1) \), provided that \( \|\eta_n - \eta_0\|_{\mathcal{N}} = o_p(1) \) and \( (P_n - P_0)(\chi_n - \chi_0) = o_p(1) \). This result extends the well-known double robustness of DML estimators for linear functionals of regression functions to a broad class of orthogonal loss functions, as in Example~\ref{example::orthoLS}. For such functionals, setting \( \theta_n = 0 \) in \( \widehat{\psi}_n \), double robustness implies that the expected gradient estimator \( -\frac{1}{n} \sum_{i=1}^n \dot{\ell}_{\eta_n}(0)(\alpha_n)(Z_i) \) is consistent for \( \psi_0(\theta_0) \) when \( \|\alpha_n - \alpha_0\|_{\mathcal{H}} = o_p(1) \). This recovers consistency of inverse probability weighting and balancing estimators as a special case \citep{li2018balancing, ben2021balancing}.

While autoDML estimators are efficient for the projection parameter \(\Psi\) in the nonparametric model \(\mathcal{P}\), they may not be efficient under semiparametric models for \(P_0\) induced by \(\mathcal{H}\). For example, if \(\mathcal{P}_{\mathcal{H}}\) denotes the semiparametric model in which the outcome regression lies in \(\mathcal{H}\), such as a partially linear model, the autoDML estimator may be inefficient relative to this model. Nonetheless, it remains regular and asymptotically linear.  Under such model restrictions, the semiparametric efficiency of autoDML depends on the choice of loss function. The most efficient choice ensures that the nonparametric EIF for \(\Psi\) coincides with that of the restricted model. This alignment occurs when \(\ell_\eta\) is chosen as the negative log-likelihood, since the resulting EIF lies in the tangent space of the restricted model and thus matches the semiparametric EIF  (see Section~3 of \cite{van2023adaptive}).

We briefly discuss the conditions here and provide a more detailed treatment in Appendix~\ref{appendix::conditions}. Conditions \ref{cond::nuisancerate} and \ref{cond::targetrate} impose that the estimators $\eta_n$ and $\theta_n$ are consistent at a rate faster than $n^{-\frac{1}{4}}$. Such rate conditions are standard in semiparametric statistics and can be achieved under appropriate conditions by several algorithms, including generalized additive models \citep{GAMhastie1987generalized}, the highly adaptive lasso \citep{vanderlaanGenerlaTMLE}, gradient-boosted trees \citep{schuler2023lassoedtreeboosting}, and neural networks \citep{Farrell2018DeepNN}. Condition~\ref{cond::DRrate} is a doubly robust rate condition, where the first rate requirement is satisfied whenever \( \|\alpha_n - \alpha_0\|_{\mathcal{H}} \, \|\theta_n - \theta_0\|_{\mathcal{H}} = o_p(n^{-1/2}) \) by \ref{cond::targetsmoothloss::two}. Condition~\ref{cond::empproc} is a standard empirical process condition that holds when \( \|\chi_n - \chi_0\|_{L^2(P_0)} = o_p(1) \) and the nuisance estimators are constructed using appropriate cross-fitting techniques \citep{van2011cross, DoubleML}.

Noting that the autoTML and autoSieve estimators correspond to one-step DML estimators for appropriate choices of nuisance estimators, Theorem \ref{theorem::limitautoDML} gives the following corollary.
\begin{corollary}
\label{corollary::autoplug}
  Assume Conditions \ref{cond::uniqueness}-\ref{cond::boundedtrue}.  \begin{enumerate}
      \item[(i)] Suppose that Conditions \ref{cond::boundnuis}-\ref{cond::empproc} hold with $(\eta_n, \theta_n^*, \alpha_n)$ as estimators of $(\eta_0, \theta_0, \alpha_0)$. Then, the autoTML estimator $\widehat{\psi}_n^{\mathrm{tmle}}$ is a regular, asymptotically linear, and efficient estimator of $\Psi(P_0)$ with influence function $\chi_0$. 
      \item[(ii)]  Suppose that Conditions \ref{cond::boundnuis}-\ref{cond::empproc} hold with $(\eta_n, \theta_{n,k(n)}, \alpha_{0,k(n)})$ as estimators of $(\eta_0, \theta_0, \alpha_0)$. Then, the autoSieve estimator $\widehat{\psi}_n^{\mathrm{sieve}}$ is a regular, asymptotically linear, and efficient estimator of $\Psi(P_0)$ with influence function $\chi_0$. 
  \end{enumerate}

\end{corollary}

Cross-validation oracle inequalities developed in \cite{van2003unified} and \cite{laan2006cross} establish the validity of the automatic undersmoothing approach used in autoSieve (see also Theorem~4 of \cite{qiu2021universal}). In particular, under mild assumptions, these results show that the selected model \(\mathcal{H}_{k(n)}\) provides sufficiently accurate approximations of both \(\theta_0\) and \(\alpha_0\) through the projections \(\theta_{0,k(n)}\) and \(\alpha_{0,k(n)}\), ensuring that both Conditions~\ref{cond::targetrate} and~\ref{cond::DRrate} hold. Corollary \ref{corollary::autoplug} generalizes classical results on debiased plug-in estimation using the method of sieves with known loss functions to orthogonal loss functions that may depend on estimated nuisance functions \citep{shen1997methods, sieveOneStepPlugin, spnpsieve, qiu2021universal}. Two-stage sieve estimation methods with loss functions that depend on nuisances were considered in \cite{sieveTwoStepPlugin}; however, the authors propose using non-orthogonal loss functions and estimating the nuisance parameters themselves with sieve estimators. In contrast, autoSieve uses Neyman-orthogonal loss functions to leverage generic machine learning techniques for the estimation of the nuisance function $\eta_0$. This flexibility is important because $\eta_0$ may be more complex --- e.g., higher dimensional and less smooth --- than the M-estimand $\theta_0$, and therefore less amenable to estimation at sufficiently fast rates to satisfy \ref{cond::nuisancerate} using the method of sieves.

Given a finite-dimensional model $\mathcal{H}$, the plug-in M-estimator $\frac{1}{n}\sum_{i=1}^n m(Z_i, \theta_n)$ with $\theta_n := \argmin_{\theta \in \mathcal{H}} \sum_{i=1}^n \ell_{\eta_n}(Z_i, \theta)$ is a special case of an autoSieve estimator, with the trivial sieve $\mathcal{H}_k := \mathcal{H}$ for each $k \in \mathbb{N}$. As a consequence, Corollary \ref{corollary::autoplug} implies that plug-in M-estimators based on Neyman-orthogonal loss functions with estimated nuisances are, under mild conditions, asymptotically linear and nonparametrically efficient for the associated projection parameter $\Psi$. In this case, the limiting variance of the plug-in M-estimator can be consistently estimated using either the sandwich variance estimator \citep{white1982MLErobust, kauermann2000sandwich} or the bootstrap \citep{tang2024consistency}, based on the loss $\ell_{\eta_n}$ and treating the nuisance estimators $\eta_n$ as fixed.

\subsection{Model misspecification error and data-driven model selection}
\label{sec::adml}

In practice, the working model used by autoDML may be data-adaptive and misspecified. For example, rather than fixing a single model \(\mathcal H\), one may select among a nested sequence of spaces
\[
\mathcal H_1 \subset \mathcal H_2 \subset \cdots \subset \mathcal H_k \subset \cdots \subset \mathcal H_\infty := \mathcal H,
\]
or, more generally, use variable selection or learned feature representations to construct a data-dependent working model \(\mathcal H_n \subseteq \mathcal H\). Given such a model, the estimator \(\widehat{\psi}_{n,\mathcal H_n}\) is constructed exactly as in Section~\ref{sec::autodmlest}, treating \(\mathcal H_n\) as fixed. The next theorem shows that the resulting misspecification bias, \(\Psi_{\mathcal H_n}(P_0) - \Psi_{\mathcal H}(P_0)\), is controlled by the approximation error of \(\theta_{0,\mathcal H_n}\) together with the approximation error of the Riesz representer under \(\mathcal H_n\). In Appendix~\ref{appendix::adml}, we use this result to establish valid superefficient inference under data-driven model selection \citep{van2023adaptive}.

Let \(\mathcal H_{n,0} := \{h_n + h_0 : h_n \in \mathcal H_n,\ h_0 \in \mathcal H_0,\ \partial_\theta^2 L_0(\theta_0,\eta_0)(h_n,h_0)=0\}\) denote the orthogonal direct sum of the data-adaptive model \(\mathcal H_n\) and the oracle model \(\mathcal H_0\), and let \(\overline{\mathcal H}_n\) and \(\overline{\mathcal H}_{n,0}\) denote their completions under \(\|\cdot\|_{\mathcal H}\). We write \(\alpha_{0,\mathcal H_{n,0}}\) for the Hessian Riesz representer in the enlarged model \(\mathcal H_{n,0}\), defined as the minimizer of \(\partial_\theta^2 L_0(\theta_0,\eta_0)(\alpha,\alpha) - 2\dot\psi_0(\theta_0)(\alpha)\) over \(\overline{\mathcal H}_{n,0}\), and note that \(\alpha_{0,\mathcal H_n}\) is its orthogonal projection onto \(\overline{\mathcal H}_n\) with respect to the Hessian inner product.

\begin{theorem}[Model approximation error]
\label{theorem::admlbias}
Assume \ref{cond::uniqueness}-\ref{cond::targetsmoothloss}, \ref{cond::PDHessian} hold for $\mathcal{H}_{n,0}$. Suppose that \(\theta_0 \in \mathcal H_0\) for some oracle submodel \(\mathcal H_0 \subseteq \mathcal H\), possibly depending on \(P_0\). Then,
   $$\Psi_{\mathcal{H}_n}(P_0) - \Psi_{\mathcal{H}}(P_0) = \partial_{\theta}^2 L_0(\theta_0, \eta_0)(\alpha_{0, \mathcal{H}_{n,0}} - \alpha_{0, \mathcal{H}_n}, \theta_{0, \mathcal{H}_n} - \theta_0) +  (\delta_{\mathrm{lin}} + \delta_{\mathrm{quad}}) O_p\left(\|\theta_{0, \mathcal{H}_n} - \theta_{0}\|_{\mathcal{H}}^2\right).$$
\end{theorem}

For sieve-based model selection, one typically has \(\mathcal H_n = \mathcal H_{k(n)}\) for some sequence \(k(n)\). If \(k(n)\to\infty\), the theorem may be applied with \(\mathcal H_0 := \mathcal H\), so that \(\Psi_{\mathcal H}(P_0)=\Psi_{\mathcal H_0}(P_0)\). If instead \(k(n)\) converges to a fixed finite value \(k_0\), one may take \(\mathcal H_0 := \mathcal H_{k_0}\). In that case, it typically holds that \(\mathcal H_n \subseteq \mathcal H_0\) with probability tending to one, and hence \(\mathcal H_{n,0} = \mathcal H_0\). More generally, the theorem shows that the approximation error vanishes whenever \(\theta_{0,\mathcal H_n} \to \theta_0\) in \(\|\cdot\|_{\mathcal H}\) and the projection error \(\alpha_{0,\mathcal H_{n,0}}-\alpha_{0,\mathcal H_n}\) is asymptotically negligible.

The result also yields a bound on omitted-variable bias in M-estimation problems where \(\mathcal H_0 := \mathcal H\) denotes the class of functions of a full set of variables, while \(\mathcal H_n\) denotes the class of functions of only a subset. In that case, the approximation error is driven by the interaction between the component of \(\theta_0\) excluded from \(\mathcal H_n\) and the corresponding component of the Riesz representer. For quadratic risks and linear functionals, the theorem holds without the remainder term $O_p\!\left(
\|\theta_{0,\mathcal H_n}-\theta_0\|_{\mathcal H}^2
\right).$ See, for example, \citet{chernozhukov2022long} for the regression case and Section~3.3 of \citet{ichimura2022influence} for a related bound under exogenous orthogonality conditions.

\section{Numerical experiments}
 \label{sec::exp}

 \label{sec::exp::survival}

In this experiment, we apply the autoDML framework to perform inference on the long-term mean survival probability under the beta-geometric survival model. In Appendix~\ref{sec::exp::cate}, we present additional experiments evaluating various autoDML estimators for the ATE derived from the R-learner loss, a specific Neyman-orthogonal loss for the conditional average treatment effect (CATE) introduced by \cite{QuasiOracleWager}.

\subsection{Overview of the beta-geometric model}

Consider latent failure and censoring times \(T, C \in \mathbb{N}\) with observed data \(Z = (X, \widetilde{T}, \Delta) \sim P_0\), where \(\widetilde{T} = T \wedge C\) and \(\Delta = \mathbbm{1}(T \leq C)\). The goal is to estimate survival probabilities beyond the observed time horizon, \(P_0(T > t_0)\), for a large reference time \(t_0\). Long-term survival estimation is critical in forecasting applications such as subscription retention and manufacturing failures. Standard nonparametric methods (e.g., Kaplan–Meier and Cox models) cannot extrapolate beyond the observed horizon due to non-identification. An alternative assumes \(T \mid X \sim \text{Geometric}(\pi(X))\), based on a time-homogeneous Markov assumption: \(P_0(T = t \mid T \geq t, X) = \pi(X)\). This model implies a constant hazard over time, which is often unrealistic because failure risk typically decreases with time survived.

The beta-geometric model \citep{hubbard2021beta} addresses this limitation by modeling \(T\) as a mixture of Markov chains. Specifically, it posits that \(T \mid \Pi, X \sim \mathrm{Geometric}(\Pi)\) for a latent probability \(\Pi \mid X \sim \mathrm{Beta}(\alpha_0(X), \beta_0(X))\), where \(\alpha_0(X) = \exp(a_0(X))\) and \(\beta_0(X) = \exp(b_0(X))\), with \(a_0(X)\) and \(b_0(X)\) unrestricted functions. Marginalizing over $\Pi$, this specification yields the conditional hazard function
\[
\lambda_{a_0, b_0}(t \mid x) = \frac{\exp(a_0(x))}{\exp(a_0(x)) + \exp(b_0(x)) + t - 1},
\]
which decays over time, accommodating higher failure risk early on. The parametric time dependence enables extrapolation beyond the observed horizon, while the covariate-dependent \((a_0, b_0)\) flexibly capture heterogeneity across individuals.


Let $\mathcal{H} := \mathcal{H}_a \times \mathcal{H}_b$, where $\mathcal{H}_a \subseteq L^2(\mu_{X})$ and $\mathcal{H}_b \subseteq L^2(\mu_{X})$ for some measure $\mu_{X}$ on $\text{supp}(X)$. The long-term survival probability \(P_0(T > t_0)\) corresponds to the functional \(E_0[m(Z, \theta_0)]\) of the M-estimand \(\theta_0 = \argmin_{(a, b) \in \mathcal{H}} E_0[\ell(Z, (a, b))]\), where \(\ell(z, a, b) = -\delta \log\big(\lambda_{a,b}(t \mid x)\big) - \sum_{s=1}^{t-1} \log\big(1 - \lambda_{a,b}(s \mid x)\big)\) is the negative log-likelihood, and \(m(z, (a, b)) = \prod_{s=1}^{t_0} \big(1 - \lambda_{a,b}(s \mid x)\big)\). Efficient estimation in the beta-geometric model would typically requires a tedious derivation of the EIF, which may not admit a closed-form expression. In our approach, Theorem~\ref{theorem::EIF} immediately characterizes the EIF, and autoDML constructs debiased estimators without relying on explicit formulas for the Riesz representer.
 The gradient and Hessian of the beta-geometric negative log-likelihood, as well as the derivative of the functional, can be computed using straightforward calculus (see Appendix~\ref{appendix::betageo}). Since we use the negative log-likelihood loss, the autoDML estimator also achieves semiparametric efficiency under the beta-geometric model.

 \label{appendix::betageo}

\subsection{Experimental results}

The data-generating distribution \( P_0 \), detailed in Appendix~\ref{appendix::simdetails}, includes three continuous covariates and one binary covariate. The failure time \( T \) follows a beta-geometric distribution with shape parameters \( a_0 \) and \( b_0 \), modeled as linear functions of the covariates. All observations are censored deterministically at \( t = 6 \). The reference time for the mean survival probability is set to \( t_0 = 12 \), twice the observable time horizon. The M-estimand \(\theta_0 = (a_0, b_0)\) is estimated using gradient-boosted trees with bivariate outputs, optimized via the negative log-likelihood loss \(\ell\). The Riesz representer \(\alpha_0 \in \mathcal{H}\) is similarly estimated using gradient-boosted trees based on the Riesz loss. Both procedures are implemented in \texttt{xgboost} \citep{xgboost} with custom objectives, following \citet{hubbard2021beta} for bivariate outcome tree models. We specify an additive structure for \(\mathcal{H}\), with \(\mathcal{H}_a\) and \(\mathcal{H}_b\) each additive in the covariates, enforced via \texttt{interaction\_constraints} in \texttt{xgboost}. We evaluate three autoDML estimators for \(\psi_0(\theta_0)\) as proposed in Section~\ref{sec::autodmlest}: two variants of the one-step estimator—one using TMLE-inspired stabilization (see \eqref{eqn::stableautoDML}, Algorithm~\ref{alg:autoDML}) and one without stabilization—and the autoTML estimator (Algorithm~\ref{alg:autoTML}). As a baseline, we include the plug-in estimator without debiasing. For simplicity, we use sample-splitting: nuisance functions are estimated on an independent sample, while point estimates and confidence intervals are computed on the original sample. For autoDML and autoTML, we construct 95\% Wald-type confidence intervals based on the estimated influence function.

\begin{figure}[ht!]
    \centering
    \includegraphics[width=0.33\linewidth]{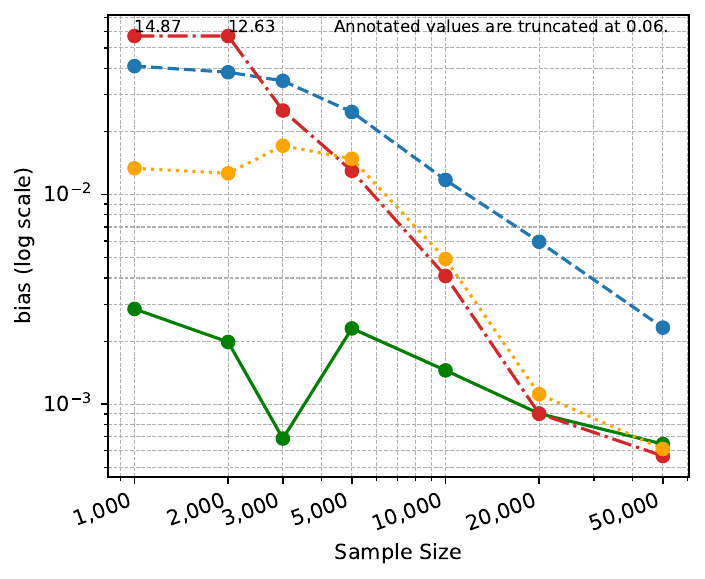}\includegraphics[width=0.33\linewidth]{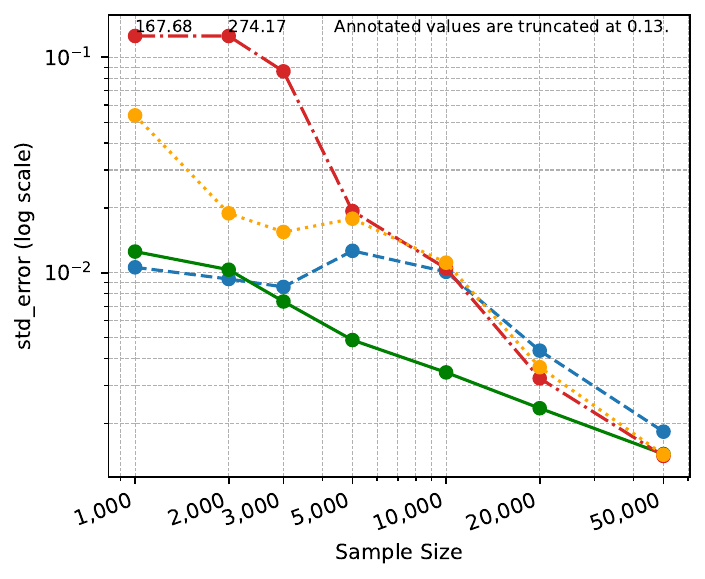} \includegraphics[width=0.33\linewidth]{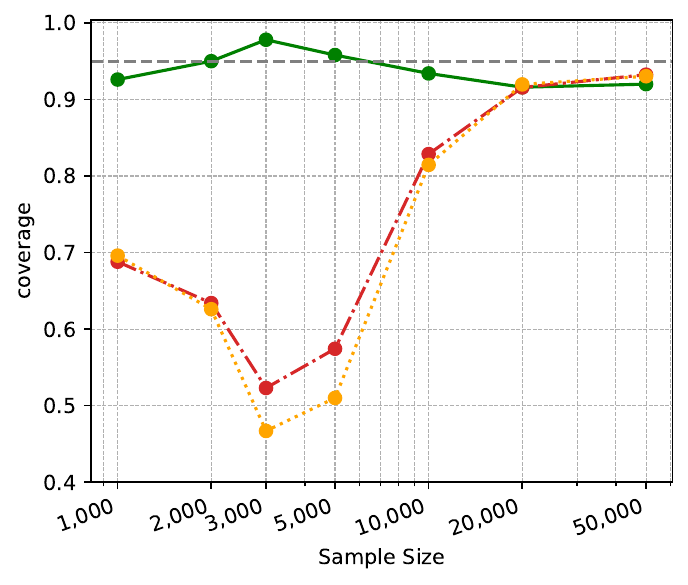}
    
    \includegraphics[width=0.6\linewidth]{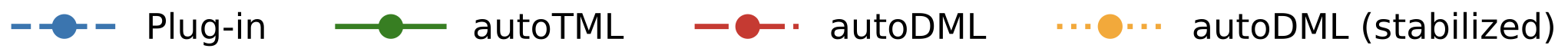}
    \caption{Performance metrics as a function of sample size ($n$) for different estimators: (Left) Absolute Bias, (Middle) Standard Error, and (Right) Coverage (95\%).}
    \label{fig:performance-metrics}
\end{figure}

Figure~\ref{fig:performance-metrics} reports Monte Carlo estimates of absolute bias, standard error, and 95\% confidence interval coverage for the autoDML methods. At large sample sizes (\(n \geq 20{,}000\)), all estimators—autoDML and autoTML—perform similarly across metrics, consistent with their asymptotic equivalence predicted by theory. However, finite-sample performance differs substantially. The autoTML estimator consistently achieves the lowest bias and standard error across all sample sizes. In contrast, the unstabilized autoDML estimator is highly biased, variable, and unstable in smaller samples, performing worse than the naive plug-in estimator. The stabilized autoDML estimator reduces bias relative to the plug-in and is more stable than its unstabilized counterpart, yet remains more biased and variable than autoTML. In terms of coverage, autoTML nearly attains the nominal 95\% level in both small and large samples, while both autoDML variants exhibit poor coverage in smaller samples due to excessive bias.

Overall, these results highlight the robustness and stability of the autoTML estimator across all metrics, particularly in small to moderate sample sizes, relative to autoDML. Although autoTML and autoDML are asymptotically equivalent, their finite-sample performance can differ markedly. Stabilization improves the performance of autoDML but still falls short of matching autoTML. The poor performance of autoDML may stem from the low quality of the initial M-estimator at smaller sample sizes. The validity of the correction term in autoDML relies on quadratic and linear approximations of the loss and functional, respectively. When the M-estimator is inaccurate, these approximations may be poor, and there is no guarantee that the bias correction will effectively remove bias. In contrast, autoTML debiases by improving the fit of the M-estimand and using a plug-in estimator, which heuristically suggests that some bias is removed even when the initial fit is of low quality.

\section{Conclusion}

\label{section::conclusion}


We established several results of independent interest. In Section~\ref{sec::eif}, we extend TMLE to generic loss functions that may depend on nuisances by generalizing the notion of least favorable submodels. In Section~\ref{sec::autodmlest}, we propose a general stabilization method for one-step autoDML estimators via \eqref{eqn::stableautoDML}. We conjecture that stabilization is particularly beneficial for nonlinear functionals and non-quadratic losses, which may explain why autoTML—stabilized by design—consistently outperforms autoDML in our experiments. Finally, in Sections~\ref{sec::theoryautodml} and~\ref{sec::adml}, we present a unified asymptotic analysis of parametric and sieve-based M-estimation under orthogonal losses, extending the results of \cite{shen1997methods}, \cite{sieveOneStepPlugin}, and \cite{sieveTwoStepPlugin}.

There are several promising avenues for future research. First, our approach can be extended to functionals that depend on nuisance parameters, provided an asymptotically linear estimator of the functional is used (e.g., \cite{vanderlaan2024automatic}). More generally, when nuisance parameters are also M-estimands, our approach applies to functionals defined jointly in both the primary and nuisance M-estimands. Second, while we focus on \emph{M-estimands} in this work, we suspect that our theoretical techniques can be extended to \emph{Z-estimands}, defined as smooth functionals of solutions to infinite-dimensional Neyman-orthogonal estimating equations, due to the close connection between Z-estimation and M-estimation \citep{vanderVaartWellner}. Third, for linear functionals and quadratic risk (e.g., linear summaries of pseudo-outcome regressions), we find that autoDML estimators are doubly robust with respect to both the M-estimand and the Hessian representer. In future work, we plan to investigate whether this property can be leveraged to obtain doubly robust inference—such as confidence intervals—for these parameters \citep{benkeser2017doubly, vanderlaan2024automatic}. Finally, Neyman-orthogonal loss functions are often nonconvex, complicating direct estimation of \(\theta_0\). For example, the orthogonal loss for conditional relative risk in \cite{van2024combining} is nonconvex, despite originating from a convex, non-orthogonal loss. In our framework, \(\theta_0\) can still be estimated using convex, non-orthogonal losses, as our method is agnostic to the estimation strategy. Alternatively, the EP-learning framework \citep{van2024combining} constructs orthogonal risks from convex losses via debiased nuisance estimates. Automating EP-learner construction for orthogonal losses is a compelling future direction.



\bibliography{ref}

\newpage

\appendix

\begingroup
\setcounter{tocdepth}{2} 
\renewcommand{\contentsname}{Appendix contents}
\tableofcontents
\endgroup

 \section{Discussion of conditions}
 \label{appendix::conditions}

 Condition \ref{cond::nuisancerate} and Condition \ref{cond::targetrate} impose that the estimators $\eta_n$ and $\theta_n$ are, respectively, consistent for $\eta_0$ and $\theta_0$ at a rate faster than $n^{-\frac{1}{4}}$. Such rate conditions are standard in semiparametric statistics and orthogonal machine learning and can be achieved under appropriate conditions by several algorithms, including generalized additive models \citep{GAMhastie1987generalized}, reproducing kernel Hilbert space estimators \citep{QuasiOracleWager}, the highly adaptive lasso \citep{vanderlaanGenerlaTMLE}, gradient-boosted trees \citep{schuler2023lassoedtreeboosting}, and neural networks \citep{Farrell2018DeepNN}. Fast oracle rates for the estimator $\theta_n$, which would be achieved if $\eta_0$ were known, can be obtained using orthogonal learning \citep{foster2023orthogonal}.

Condition~\ref{cond::empproc} is a standard empirical process requirement that holds when $\|\chi_n - \chi_0\|_{L^2(P_0)} = o_p(1)$ and the realizations of $\chi_n - \chi_0$ belong to a fixed Donsker class. In debiased machine learning, sample-splitting and cross-fitting are commonly used to mitigate overfitting and relax Donsker conditions \citep{van2011cross, DoubleML}. In particular, \ref{cond::empproc} holds for arbitrary machine learning estimators if $\|\chi_n - \chi_0\|_{L^2(P_0)} = o_p(1)$ and the nuisance estimators $\eta_n$, $\theta_n$, and $\alpha_n$ are trained on an external sample independent of $\{Z_i\}_{i=1}^n$. To improve data efficiency, one may instead use cross-fitting, as in Algorithms~\ref{alg:crossfit}–\ref{alg:autoTML}, where nuisance and parameter estimates are computed across multiple data splits. For our modified procedure in Algorithm~\ref{alg:crossfit}, however, \ref{cond::empproc} does not follow from standard cross-fitting arguments due to data leakage across folds from reusing cross-fitted nuisance estimates. Nonetheless, we conjecture that it still holds under reasonable assumptions, since the leakage is limited. For example, the condition is satisfied if, for each \(j \in [J]\), the realizations of the cross-fitted estimators \(\{\theta_{n,j}, \alpha_{n,j}\}\) lie in a (possibly random) function class \(\mathcal{F}_{n,j}\) that is conditionally fixed and uniformly Donsker \citep{sheehy1992uniform}, given the \(j\)-th training set.

\section{Examples}
\label{sec:examplesback}
\subsection{Loss functions satisfying conditions}

Lemma \ref{lemma:smoothlosssuff} verifies that Conditions~\ref{cond::targetsmoothloss} and~\ref{cond::nuisancesmooth} hold for a broad class of twice-differentiable loss functions, but such smoothness is not strictly necessary. In particular, many common losses—such as the quantile loss—do not satisfy classical second-order smoothness, yet still induce a twice-differentiable risk functional.

\setcounter{example}{2}
\renewcommand{\theexample}{\arabic{example}}

\begin{example}[Smoothness of median loss]
Suppose \(Z = (X, Y)\) with \(X \in \mathbb{R}^d\) and \(Y \in \mathbb{R}\), and consider the median loss \(\ell(\theta, z) = |y - \theta(x)|\). This loss is almost everywhere differentiable, with a generalized second derivative that exists in the distributional sense. Under standard regularity conditions on the conditional density \(p_0(Y \mid X)\) (see Appendix~\ref{appendix::examples}), the population risk admits a bounded and Lipschitz continuous second derivative, $\partial_{\theta}^2 L_0(\theta, \eta)(h_1, h_2) = \int p_0(Y = \theta(x) \mid X = x) \cdot h_1(x) h_2(x) \, P_0(dx).$ Hence, \ref{cond::targetsmoothloss} is satisfied for the median loss, and more generally for orthogonal quantile losses, such as those arising in quantile treatment effect settings \citep{leqi2021median, whitehouse2024orthogonal}. \qed
\end{example}

Continuing with our examples in Section \ref{sec::examples}, we demonstrate how various functionals and loss functions of interest in causal inference satisfy our conditions. In each of the following examples, suppose each \(P \in \mathcal{P}\) is dominated by \(P_0\) and satisfies \(c \|\cdot\|_{L^2(P_0)} \leq \|\cdot\|_{L^2(P)} \leq C \|\cdot\|_{L^2(P_0)}\) for constants \(0 < c < C < \infty\). Let \(\mathcal{H} = L^2(P_0)\), with \(\|\cdot\|_{\mathcal{H}}\) the \(L^2(P_0)\) norm and \(\rho_{\mathcal{H}}(\cdot)\) the \(P_0\)-essential supremum norm. 

\renewcommand{\theexample}{1b}
 
\begin{example}[continued]
Let \(\eta_0 := (\pi_0, \mu_0(1, \cdot), \mu_0(0, \cdot))\) denote the nuisance parameter. Let \(\mathcal{N} = L^\infty(\lambda) \times L^\infty(\lambda) \times L^\infty(\lambda)\). The orthogonal loss for semiparametric logistic regression can be written as $\ell_{\eta_0}(\theta, z) = l(\pi_0(x), \mu_0(0,x), \mu_0(1,x), \theta(x), z)$, where the function \(l : \mathbb{R}^4 \times \mathcal{Z} \to \mathbb{R}\) is defined by
\[
l(\overline{\pi}, \overline{\mu}^0,  \overline{\mu}^1, \overline{\theta}, z) :=
 g_1(a, \overline{\mu}^0,  \overline{\mu}^1) \left\{ \log\left(1 + \exp\left(- (a - \overline{\pi}) \, \overline{\theta} - g_2(\overline{\pi}, \overline{\mu}^0, \overline{\mu}^1 )\right)\right) - y (a - \overline{\pi}) \, \overline{\theta} \right\},
\]
where $g_1(a, \overline{\mu}^0,  \overline{\mu}^1) = \{(1-a) \overline{\mu}^0 + a  \overline{\mu}^1\}^{-1}$ and $g_2(\overline{\pi}, \overline{\mu}^0, \overline{\mu}^1 ) = (1 -\overline{\pi}) \operatorname{logit}(\overline{\mu}^0) + \overline{\pi}\operatorname{logit}(\overline{\mu}^1)$. Under reasonable conditions, the function \(l\) satisfies the assumptions of Lemma~\ref{lemma:smoothlosssuff}, and Conditions~\ref{cond::targetsmoothloss} and~\ref{cond::nuisancesmooth} hold (see Appendix \ref{appendix::examples} for details). In particular, \ref{cond::targetsmoothloss::one} is satisfied at \(P_0\), with the map \(h \mapsto \dot{\ell}_{\eta_0}(\theta)(h)\) given elementwise by \(z \mapsto -h(x) \cdot \frac{a - \pi_0(x)}{\nu_0(a, x)} \left[y - \expit\left\{(a - \pi_0(x)) \theta(x) + h_0(x)\right\}\right]\). Furthermore, \ref{cond::targetsmoothloss::two} is satisfied with \(\partial_{\theta}^2 L_0(\theta, \eta_0)(h_1, h_2) = E_0\left[\{A - \pi_0(X)\}^2 h_1(X) h_2(X)\right]\). \ref{cond::PDHessian} holds if \(1 - \delta > \pi_0(X) > \delta\) \(P_0\)-almost surely for some \(\delta > 0\). Finally, when the semiparametric logistic regression model is correctly specified at \(P_0\), \ref{cond::orthogonal} holds by \citet{nekipelov2022regularised} (see Equation~(2.24) and Remark~2.2). The EIF term $-\dot{\ell}_{\eta_0}(\theta_0)(\alpha_0)$ in Theorem \ref{theorem::EIF} equals  $z \mapsto  \alpha_0(x) \frac{(a - \pi_0(x)) }{\nu_0(a, x)} [y -\expit\{(a - \pi_0(x)) \theta_0(x) + h_0(x)\}]$, where $\alpha_0 = \argmin_{\alpha \in \mathcal{H}} E_0[\pi_0(X)\{1 - \pi_0(X)\}\alpha^2(X) - 2\,\dot{m}_{\theta_0}(Z, \alpha)]$ is a weighted Riesz representer with overlap weights $\pi_0\{1 - \pi_0\}$. 
 \qedsymbol
\end{example}

 \renewcommand{\theexample}{1c}

\begin{example}[continued]
\label{example::orthoLS}
Consider the pseudo squared-error loss $\ell_{\eta}: (z, \theta) \mapsto \frac{1}{2} w_{\eta}(z)\left\{\zeta_{\eta}(z) - \theta(x)\right\}^2$. Suppose $\esssup_{z \in \mathcal{Z}, P \in \mathcal{P}}\max\{w_{\eta_P}(z), \zeta_{\eta_P}(z), \theta_P(x)\} < \infty$ with respect to the dominating measure $\lambda$. This loss satisfies \ref{cond::targetsmoothloss::one} with $\dot{\ell}_{\eta_0}(\theta)(h): z \mapsto - w_{\eta_0}(z)h(x)\{\zeta_{\eta_0}(z) - \theta(x)\}$, and \ref{cond::targetsmoothloss::two} with $\partial_{\theta}^2L_0(\theta, \eta_0)(h_1, h_2) = E_0[w_{\eta_0}(Z)h_1(X)h_2(X)]$. Conditions~\ref{cond::targetsmoothloss::three} and \ref{cond::nuisancesmooth} can be verified for many loss functions by applying Lemma~\ref{lemma:smoothlosssuff}. Condition \ref{cond::PDHessian} holds if $w_{\eta_0}(Z) > \delta$ $P_0$-almost surely for some $\delta > 0$. For the R-learner loss, recall that $w_{\eta_0}: z \mapsto (a - \pi_0(x))^2$ and $\zeta_{\eta_0}: z \mapsto \frac{y - m_0(x)}{a - \pi_0(x)}$, where $\eta_0 = (\pi_0, m_0)$ with $m_0: x \mapsto E_0[Y \mid X = x]$. We can compute the cross derivative to be
\[
\begin{aligned}
\partial_\eta \partial_\theta L_0(\theta,\eta_0)\bigl((\delta\pi,\delta m), h\bigr)
&= E_0\Bigl[\Bigl\{
   -(A-\pi_0(X))\,\delta\pi(X)\,\bigl(\zeta_{\eta_0}(Z)-\theta(X)\bigr)\\
&\qquad\quad -\,(A-\pi_0(X))\,\delta m(X)
   +\,(Y-m_0(X))\,\delta\pi(X)
   \Bigr\}\,h(X)\Bigr].
\end{aligned}
\]
The first term on the right-hand side equals \(E_0[\,\delta\pi(X)\,\bigl(Y - m_0(X) - (A - \pi_0(X))\theta(X)\bigr)]\), which is zero since \(E_0[Y - m_0(X) \mid X] = 0\) and \(E_0[A - \pi_0(X) \mid X] = 0\). The second term also vanishes for the same reason. Hence, \ref{cond::orthogonal} holds, and, in fact, the R-learner loss is universally orthogonal, satisfying \(\partial_\eta \partial_\theta L_0(\theta, \eta_0) = 0\) for all \(\theta \in \mathcal{H}\).  \qedsymbol
\end{example}

 \subsection{Examples of estimators}

\renewcommand{\theexample}{1b}

\begin{example}[continued]
Consider the orthogonal semiparametric logistic regression loss, and let $\theta_n$ and $\alpha_n$ be estimators of $\theta_0$ and $\alpha_0$, respectively, where $\alpha_0 = \argmin_{\alpha \in \overline{\mathcal{H}}} E_0[\pi_0(X)\{1 - \pi_0(X)\}\alpha^2(X) - 2\,\dot{m}_{\theta_0}(Z, \alpha)]$. Let $\pi_n$ be an estimator of the propensity score $\pi_0$, and let $\mu_n$ be an estimator of the outcome regression $\mu_0$. Define $\nu_n = \mu_n (1 - \mu_n)$ and $h_n(x) = \pi_n(x)\logit \mu_n(1,x) + (1 - \pi_n(x))\logit \mu_n(0,x)$. Then, an autoDML estimator of $\psi_0(\theta_0)$ is given by
$$
\frac{1}{n}\sum_{i=1}^n m(Z_i, \theta_n) + \frac{1}{n}\sum_{i=1}^n \alpha_n(X_i) \frac{(A_i - \pi_0(X_i)) }{\nu_n(A_i, X_i)} \left[y_i -\expit\{(A_i - \pi_n(X_i)) \theta_n(X_i) + h_n(X_i)\}\right]. \text{ \qedsymbol}
$$
\end{example}

\renewcommand{\theexample}{1c}
\begin{example}[continued]
For general pseudo-outcome regressions, an autoDML estimator is given by $\widehat{\psi}_n^{\mathrm{dml}} := \frac{1}{n}\sum_{i=1}^n m(Z_i, \theta_n) + \frac{1}{n}\sum_{i=1}^n w_{\eta_n}(Z_i) \alpha_n(X_i) \{\zeta_{\eta_n}(Z_i) - \theta_n(X_i)\}.$ For the R-learner loss of the CATE $\theta_0: x \mapsto \mu_0(1, x) - \mu_0(0, x)$, corresponding to $w_{\eta_0}: z \mapsto (a - \pi_0(x))^2$ and $\zeta_{\eta_0}: z \mapsto \frac{y - m_0(x)}{a - \pi_0(x)}$, we have $-\dot{\ell}_{\eta_0}(\theta_0)(\alpha_0): z \mapsto \alpha_0(x)\{a - \pi_0(x)\} \{y - m_0(x) - (a - \pi_0(x))\theta_0(x)\}$, where $\alpha_0 = \argmin_{\alpha \in \mathcal{H}} E_0[\pi_0(X)\{1 - \pi_0(X)\}\alpha^2(X)] - 2\dot{\psi}_0(\theta_0)(\alpha)$ is an overlap-weighted Riesz representer. Taking $\psi_0$ as the ATE, with $\psi_0(\alpha) = \dot{\psi}_0(\theta_0)(\alpha) = E_0[\alpha(X)]$, we find that $\alpha_0$ is the overlap-weighted projection of $\{\pi_0(x)(1 - \pi_0(x))\}^{-1}$ onto the CATE model $\mathcal{H}$. An autoDML estimator of the ATE is given by
\[
\widehat{\psi}_n^{\mathrm{dml}} := \frac{1}{n}\sum_{i=1}^n \theta_n(Z_i) + \frac{1}{n}\sum_{i=1}^n (A_i - \pi_n(X_i))\alpha_n(X_i)\{Y_i - m_n(X_i) - (A_i - \pi_n(X_i)) \theta_n(X_i)\},
\]
where $\eta_n := (\pi_n, m_n)$ is an estimator of $\eta_0 := (\pi_0, m_0)$, $\theta_n$ is an estimator of the CATE, and $\alpha_n$ is an estimator of $\alpha_0$. For the R-learner loss, the fluctuation parameter for autoTML is $\varepsilon_n = \argmin_{\varepsilon} \{Y_i - m_n(X_i) - \{A_i - \pi_n(X_i)\}\{\theta_n(X_i) + \varepsilon \alpha_n(X_i)\}\}^2$, and the autoTML estimator for the ATE is
\[
\widehat\psi_n^{\mathrm{tmle}} = \frac{1}{n}\sum_{i=1}^n \theta_n(X_i) + \frac{\sum_{i=1}^n \alpha_n(X_i)}{\sum_{i=1}^n (A_i - \pi_n(X_i))^2\,\alpha_n(X_i)^2} \frac{1}{n}\sum_{i=1}^n (A_i - \pi_n(X_i))\alpha_n(X_i)\{Y_i - m_n(X_i) - (A_i - \pi_n(X_i)) \theta_n(X_i)\}. \text{ \qedsymbol}
\]
\end{example}

\subsection{Additional technical details for examples}
\label{appendix::examples}

 In each of the following examples, suppose that each \(P \in \mathcal{P}\) is dominated by a measure \(\lambda\), and that \(c \|\cdot\|_{L^2(\lambda)} \leq \|\cdot\|_{L^2(P)} \leq C \|\cdot\|_{L^2(\lambda)}\) for some \(0 < c < C < \infty\). Let \(\mathcal{H} = L^2(\lambda)\), \(\mathcal{H} = L^\infty(\lambda)\), \(\|\cdot\|_{\mathcal{H}}\) be the \(L^2(P)\) norm, and \(\rho_{\mathcal{H}}(\cdot)\) be the \(\lambda\)-essential supremum norm. By norm equivalence, \(\mathcal{H}\) is closed in \(L^2(P)\) for each \(P \in \mathcal{P}\).

 \setcounter{example}{1}
 
\begin{example}[Smoothness of median loss]
Suppose $Z = (X, Y)$ with $X \in \mathbb{R}^d$ and $Y \in \mathbb{R}$, and consider the median loss $\ell(\theta, z) = |y - \theta(x)|$. Then $\ell$ is almost everywhere differentiable with $\dot{\ell}(\theta)(h)(z) = \text{sign}(y - \theta(x)) \cdot h(x)$. The second derivative exists only in a distributional sense and is given by $\ddot{\ell}(\theta)(h_1, h_2)(z) = \delta(y - \theta(x)) \cdot h_1(x) h_2(x)$, where $\delta$ is the Dirac delta function. Consequently, the second derivative of the population risk satisfies $\partial_{\theta}^2 L_0(\theta, \eta)(h_1, h_2) = \int p_0(Y = \theta(x) \mid X = x) \cdot h_1(x) h_2(x) \, P_0(dx)$. Let $\mathcal{H} := L^\infty(\lambda)$ with $\rho_{\mathcal{H}}(\cdot) = \|\cdot\|_{L^\infty(\lambda)}$ and $\|\cdot\|_{\mathcal{H}} = \|\cdot\|_{L^2(P_0)}$. Suppose $\esssup_{x \in \mathcal{X},\, \theta \in \mathcal{H}_{\mathcal{P}}, P \in \mathcal{P}} p(Y = \theta(x) \mid X = x) < \infty$, and that the map $y \mapsto p(Y = y \mid X = x)$ is Lipschitz continuous for each $P \in \mathcal{P}$ with constant $L_x < \infty$, where $L := \esssup_x L_x < \infty$. Then, uniformly over all $h_1, h_2 \in \mathcal{H}$, the second derivative is bounded, satisfying $\partial_{\theta}^2 L_0(\theta, \eta)(h_1, h_2) \leq C \| h_1\|_{\mathcal{H}} \|h_2 \|_{\mathcal{H}}$, and Lipschitz continuous with
\[
|\partial_{\theta}^2 L_0(\theta + h, \eta)(h_1, h_2) - \partial_{\theta}^2 L_0(\theta, \eta)(h_1, h_2)| \leq \int L_x |h(x)| |h_1(x) h_2(x)| \, P_0(dx) \leq L \rho_{\mathcal{H}}(h_1) \|h_2\|_{\mathcal{H}} \|h\|_{\mathcal{H}}.
\]
Hence, Condition~\ref{cond::targetsmoothloss} is satisfied for the median loss, and more generally for orthogonal quantile losses, such as those arising in quantile treatment effect estimation \citep{leqi2021median}. \qed
\end{example}

\setcounter{example}{2}
\begin{example}[Verifying conditions for semiparametric orthogonal logistic loss]
Let \(\eta_0 := (\pi_0, \mu_0(1, \cdot), \mu_0(0, \cdot))\) denote the nuisance parameter. Let \(\mathcal{N} = L^\infty(\lambda) \times L^\infty(\lambda) \times L^\infty(\lambda)\). We identify functions of subvectors of \(z \in \mathcal{Z}\) with elements of \(L^\infty(\lambda)\) via composition with coordinate projections; for example, \(f(z) := f(x)\) when \(z = (x, y)\). The orthogonal loss for semiparametric logistic regression can be written as $\ell_{\eta_0}(\theta, z) = l(\pi_0(x), \mu_0(0,x), \mu_0(1,x), \theta(x), z)$, where the function \(l : \mathbb{R}^4 \times \mathcal{Z} \to \mathbb{R}\) is defined by
\[
l(\overline{\pi}, \overline{\mu}^0,  \overline{\mu}^1, \overline{\theta}, z) :=
 g_1(a, \overline{\mu}^0,  \overline{\mu}^1) \left\{ \log\left(1 + \exp\left(- (a - \overline{\pi}) \, \overline{\theta} - g_2(\overline{\pi}, \overline{\mu}^0, \overline{\mu}^1 )\right)\right) - y (a - \overline{\pi}) \, \overline{\theta} \right\},
\]
where $g_1(a, \overline{\mu}^0,  \overline{\mu}^1) = \{(1-a) \overline{\mu}^0 + a  \overline{\mu}^1\}^{-1}$ and $g_2(\overline{\pi}, \overline{\mu}^0, \overline{\mu}^1 ) = (1 -\overline{\pi}) \operatorname{logit}(\overline{\mu}^0) + \overline{\pi}\operatorname{logit}(\overline{\mu}^1)$. Suppose that for all \(P \in \mathcal{P}\) and \(z \in \mathcal{Z}\), we have \(\delta < \mu_P(a, x) < 1 - \delta\) and \(\max\{|h_P(x)|, |\theta_P(x)|\} < M'\) almost surely, for some \(\delta \in (0,1)\) and \(M' < \infty\), with \(\mathcal{Z}\) compact. Then \ref{cond::boundedtrue} holds and the closure of the set \(\{(\pi(x), \mu(0, x), \mu(1, x), \theta(x), z) : \theta \in \mathcal{H}_K, (\pi, \mu(0, \cdot), \mu(1, \cdot)) \in \mathcal{N}_K, z \in \mathcal{Z}\}\) is compact. Moreover, the function \(l\) is three times Lipschitz continuously differentiable in all arguments on this compact set, where we use that $t \mapsto t^{-1}$ on a closed subset of \((0,1)\) and \(\operatorname{logit}\) on \([-M', M']\) are infinitely often Lipschitz continuously differentiable. Hence, \(l\) satisfies the conditions of Lemma~\ref{lemma:smoothlosssuff}, and Conditions~\ref{cond::targetsmoothloss} and~\ref{cond::nuisancesmooth} hold. In particular, \ref{cond::targetsmoothloss::one} is satisfied at \(P_0\), with the map \(h \mapsto \dot{\ell}_{\eta_0}(\theta)(h)\) given elementwise by \(z \mapsto -h(x) \cdot \frac{a - \pi_0(x)}{\nu_0(a, x)} \left[y - \expit\left\{(a - \pi_0(x)) \theta(x) + h_0(x)\right\}\right]\), where \(\|\cdot\|_{\mathcal{H}} := \|\cdot\|_{L^2(P_0)}\). Furthermore, \ref{cond::targetsmoothloss::two} is satisfied with \(\partial_{\theta}^2 L_0(\theta, \eta_0)(h_1, h_2) = E_0\left[\{A - \pi_0(X)\}^2 h_1(X) h_2(X)\right]\). Condition~\ref{cond::PDHessian} holds if \(1 - \delta > \pi_0(X) > \delta\) \(P_0\)-almost surely for some \(\delta > 0\). Condition~\ref{cond::uniqueness} holds if \(1 - \delta > \pi_P(X) > \delta\) uniformly over \(P \in \mathcal{P}\), in which case the risk \(\theta \mapsto L_P(\theta, \eta_P)\) is strongly convex and continuous on $\mathcal{H}$. Finally, when the semiparametric logistic regression model is correctly specified at \(P_0\), \ref{cond::orthogonal} holds by \citet{nekipelov2022regularised} (see Equation~(2.24) and Remark~2.2).
 \qedsymbol
\end{example}

\section{Model selection with Adaptive Debiased Machine Learning}
\label{appendix::adml}

Adaptive Debiased Machine Learning (ADML) \citep{van2023adaptive} combines debiased machine learning with data-driven model selection to construct superefficient, adaptive estimators of smooth functionals, allowing model assumptions to be learned from the data. The key idea is to estimate $\Psi_n(P_0)$ for a data-adaptive, projection-based parameter $\Psi_n: \mathcal{P} \rightarrow \mathbb{R}$, defined through a submodel $\mathcal{P}_n \subseteq \mathcal{P}$ learned from the data. If $\mathcal{P}_n$ stabilizes asymptotically to a fixed oracle submodel $\mathcal{P}_0 := \mathcal{P}_{P_0}$ in a suitable sense, the resulting estimator remains efficient for an oracle target $\Psi_0: \mathcal{P} \rightarrow \mathbb{R}$ defined through $\mathcal{P}_0$. This oracle parameter coincides with the true target $\Psi$ on $\mathcal{P}_0$, so that $\Psi_0(P_0) = \Psi(P_0)$ when $P_0 \in \mathcal{P}_0$. Crucially, $\Psi_0$ may have a substantially smaller efficiency bound at $P_0$ than $\Psi$, leading to more efficient estimators and tighter confidence intervals, while preserving unbiasedness. The simplest example is when $\mathcal{P}_n := \mathcal{P}_{k(n)}$ with $k(n) \to \infty$, obtained via model selection over a sieve $\mathcal{P}_1 \subseteq \mathcal{P}_2 \subseteq \dots \subseteq \mathcal{P}$, and the limiting model $\mathcal{P}_0$ is either the full model $\mathcal{P}$ or some element of the sequence containing $P_0$.

We apply the ADML framework to construct adaptive estimators of $\Psi(P_0)$ by selecting a data-driven working model for the M-estimand $\theta_0$. Let $\mathcal{H}_n \subseteq \mathcal{H}$ denote this model, and suppose it approximates an unknown oracle submodel $\mathcal{H}_0 \subseteq \mathcal{H}$. The corresponding statistical models are $\mathcal{P}_n := \{P \in \mathcal{P} : \theta_P \in \mathcal{H}_n\}$ and $\mathcal{P}_0 := \{P \in \mathcal{P} : \theta_P \in \mathcal{H}_0\}$. We define the working and oracle parameters by $\Psi_{\mathcal{H}_n}(P) := \psi_P(\theta_{P,\mathcal{H}_n})$ and $\Psi_{\mathcal{H}_0}(P) := \psi_P(\theta_{P,\mathcal{H}_0})$, where $\theta_{P,\mathcal{H}} := \argmin_{\theta \in \mathcal{H}} L_P(\theta, \eta_P)$. Let $\chi_{P, \mathcal{H}_n}$ and $\chi_{P, \mathcal{H}_0}$ denote the efficient influence functions of $\Psi_{\mathcal{H}_n}$ and $\Psi_{\mathcal{H}_0}$, respectively, as given in Theorem~\ref{theorem::EIF}. For example, $\mathcal{H}_n$ may be selected via cross-validation over a sieve of models $\mathcal{H}_1 \subset \mathcal{H}_2 \subset \cdots \subset \mathcal{H}_\infty := \mathcal{H}$, where $\mathcal{H}$ is a correctly specified model containing $\theta_0$, and $\mathcal{H}_0$ is the smallest correctly specified model in the sieve. Alternatively, $\mathcal{H}_n$ could be obtained via variable selection or a data-driven feature transformation, with $\mathcal{H}_0$ corresponding to the limiting set of selected variables or the limiting feature transformation.

Given the selected model $\mathcal{H}_n$, our proposed ADML estimator of $\Psi(P_0)$ is the debiased estimator $\widehat{\psi}_{n,\mathcal{H}_n}$ of the data-adaptive working parameter $\Psi_{\mathcal{H}_n}: \mathcal{P} \rightarrow \mathbb{R}$, constructed using any method from Section~\ref{sec::autodmlest}. We show that, under suitable conditions, this estimator remains valid even if $\mathcal{H}_n$ is misspecified, provided the model approximation error vanishes asymptotically. A key step is showing that the parameter approximation bias $\Psi_{\mathcal{H}_n}(P_0) - \Psi(P_0)$ is second-order in the model error and thus asymptotically negligible. Theorem \ref{theorem::admlbias} in the main text confirms this second-order behavior, establishing that the bias vanishes as long as the model approximation error decays sufficiently fast.

We now present our main result on the asymptotic linearity and superefficiency of the ADML estimator $\widehat{\psi}_{n,\mathcal{H}_n}$ for $\Psi(P_0)$. To establish this result, we assume the following conditions.

\begin{enumerate}[label=\textbf{(C\arabic*)}, ref=C\arabic*, resume = cond]
    \item \textit{Linear expansion:} $\widehat{\psi}_{n, \mathcal{H}_n} - \Psi_n(P_0) = (P_n - P_0) \chi_{0, \mathcal{H}_n} + o_p(n^{-1/2})$. \label{cond::linearexp}
    \item \textit{Stabilization of selected model:} $n^{1/2}(P_n- P_0)\{\chi_{0, \mathcal{H}_n} - \chi_{0, \mathcal{H}_0}\} = o_p(1)$. \label{cond::stable}
    \item \textit{Target approximate rate:} \label{cond::nuisrate}  If $\delta_{\mathrm{lin}} + \delta_{\mathrm{quad}} > 0$ then $\|\theta_{0, \mathcal{H}_n} - \theta_{0}\|_{\mathcal{H}} = o_p(n^{-\frac{1}{4}})$.
    \item \textit{Doubly robust approximate rate:} \label{cond::modelbias} $\partial_{\theta}^2 L_0(\theta_0, \eta_0)(\alpha_{0, \mathcal{H}_{n,0}} - \alpha_{0, \mathcal{H}_n}, \theta_{0, \mathcal{H}_n} - \theta_0)  = o_p(n^{-1/2})$.  
\end{enumerate}

\begin{theorem}
Assume that Conditions \ref{cond::uniqueness}-\ref{cond::boundedtrue} hold for the model $\mathcal{H}_{n,0}$. Suppose that $\mathcal{H}_n$ converges to an oracle submodel $\mathcal{H}_0$ with $\theta_0 \in \mathcal{H}_0$ in the sense that Conditions \ref{cond::linearexp}-\ref{cond::modelbias} hold. Then, $\widehat{\psi}_{n,\mathcal{H}_n} - \Psi(P_0) = P_n \chi_{0, \mathcal{H}_0} + o_p(n^{-1/2})$, and $\widehat{\psi}_{n,\mathcal{H}_n}$ is a locally regular and efficient estimator for the oracle parameter $\Psi_{0}$ under the nonparametric statistical model.
\label{theorem::ADML}
\end{theorem}

It is interesting to apply Theorem \ref{theorem::ADML} with $\mathcal{H}_n = \mathcal{H}_{k(n)}$ and $\mathcal{H}_0 = \mathcal{H}$, where $\mathcal{H}_{k(n)}$ is an element of the sieve $\mathcal{H}_1 \subset \mathcal{H}_2 \subset \mathcal{H}_3 \subset \dots \subset H_{\infty} := \mathcal{H}$ with $k(n) \rightarrow \infty$. In this case, the theorem establishes the asymptotic linearity and efficiency of sieve-based plug-in estimators based on orthogonal losses, generalizing and extending the results of \cite{shen1997methods, sieveOneStepPlugin} and \cite{sieveTwoStepPlugin}.

Condition~\ref{cond::linearexp} requires that the ADML estimator $\widehat{\psi}_{n, \mathcal{H}_n}$ is debiased for the data-adaptive working parameter $\Psi_n$, which can be established under conditions analogous to Theorem~\ref{theorem::limitautoDML} applied with $\mathcal{H} := \mathcal{H}_n$. Conditions~\ref{cond::stable}--\ref{cond::modelbias}, drawn from prior ADML work \citep{van2023adaptive, van2024adaptive}, ensure that data-driven model selection does not invalidate the debiased estimator. Condition~\ref{cond::stable} is an asymptotic stability condition requiring that the learned model $\mathcal{H}_n$ converges to a fixed oracle submodel $\mathcal{H}_0$, typically formalized as $\|\chi_{0, \mathcal{H}_n} - \chi_{0, \mathcal{H}_0}\|_{L^2(P_0)} = o_p(1)$. This holds when the nuisance functions $\alpha_{0, \mathcal{H}_n}$ and $\theta_{0, \mathcal{H}_n}$ converge to their oracle counterparts $\alpha_{0, \mathcal{H}_0}$ and $\theta_0$ in an appropriate sense. Conditions~\ref{cond::nuisrate} and~\ref{cond::modelbias} further ensure that the approximation bias from using $\mathcal{H}_n$ instead of $\mathcal{H}_0$ is negligible, by requiring sufficiently fast convergence of $\alpha_{0, \mathcal{H}_n}$ and $\theta_{0, \mathcal{H}_n}$. Combined with Theorem~\ref{theorem::admlbias}, these conditions imply that $\Psi_n(P_0) - \Psi(P_0) = o_p(n^{-1/2})$.

\section{Proofs for Section \ref{sec::eif}}

\subsection{Proof of Lemma~\ref{lemma:smoothlosssuff}}

\begin{proof}[Proof of Lemma~\ref{lemma:smoothlosssuff}]
We compute the first derivative of the map \(t \mapsto L_0(\theta + t h, \eta)\) evaluated at \(t = 0\):
\begin{align*}
\partial_\theta L_0(\theta, \eta)(h)
&= \left.\frac{d}{dt} L_0(\theta + t h, \eta)\right|_{t=0} \\
&= \left.\frac{d}{dt} \int l(\theta(z) + t h(z), \eta(z), z) \, P_0(dz) \right|_{t=0} \\
&= \int \left.\frac{d}{dt} l(\theta(z) + t h(z), \eta(z), z) \right|_{t=0} \, P_0(dz) \\
&= \int \partial_a l(\theta(z), \eta(z), z)^\top h(z) \, P_0(dz),
\end{align*}
which shows that \(\dot{\ell}_\eta(\theta, z)(h) = \partial_a l(\theta(z), \eta(z), z)^\top h(z)\). To justify the interchange of differentiation and integration, note that for each fixed \(z\), the map
\[
t \;\mapsto\; \partial_a l\bigl(\theta(z) + t h(z), \eta(z), z\bigr)^\top h(z)
\]
is continuous by the joint continuity of \(\partial_a l\). Moreover, by assumption \(\partial_a^2 l(a,b,z)\) is Lipschitz continuous in \((a,b)\) on the compact set \(\mathcal{C}\), so there exists a constant \(C\) such that
\[
\bigl|\partial_a l(\theta(z) + t h(z), \eta(z), z)^\top h(z)\bigr|
\;\le\;
C\,\|h(z)\|_{\mathbb{R}^{d_1}}^2,
\]
and \(\|h\|_{\mathcal{H}}<\infty\) ensures integrability.  The Dominated Convergence Theorem then allows differentiation under the integral.

Next, define the remainder
\[
R(h) := L_0(\theta + h, \eta) - L_0(\theta, \eta) - \partial_\theta L_0(\theta, \eta)(h).
\]
A first‐order Taylor expansion in \(a\) gives
\[
l(\theta(z) + h(z), \eta(z), z)
= l(\theta(z), \eta(z), z)
+ \partial_a l(\theta(z), \eta(z), z)^\top h(z)
+ r(z),
\]
with \(|r(z)| \le \tfrac12\|\partial_a^2 l\|_\infty\,\|h(z)\|_{\mathbb{R}^{d_1}}^2\) and $\|\partial_a^2 l\|_\infty
:= \sup_{(a,b,z)\in\mathcal C}\|\partial_a^2 l(a,b,z)\| \;<\;\infty.$ .  Integrating,
\[
|R(h)| \le \tfrac12 \|\partial_a^2 l\|_\infty \int \|h(z)\|_{\mathbb{R}^{d_1}}^2\,P_0(dz)
= \tfrac12 \|\partial_a^2 l\|_\infty \|h\|_{\mathcal{H}}^2,
\]
so $\frac{|R(h)|}{\|h\|_{\mathcal{H}}}
\;\le\;
\tfrac12 \|\partial_a^2 l\|_\infty \,\|h\|_{\mathcal{H}}
\;\to\;0
\quad\text{as }\|h\|_{\mathcal{H}}\to0,$ verifying first‐order Fr\'echet differentiability (Condition~\ref{cond::targetsmoothloss::one}).

For the second derivative,
\begin{align*}
\partial_\theta^2 L_0(\theta, \eta)(h_1, h_2)
&= \frac{\partial^2}{\partial t\,\partial s}
    L_0\bigl(\theta + t h_1 + s h_2,\eta\bigr)\big|_{t=s=0} \\
&= \int (h_1(z))^\top\,\partial_a^2 l(\theta(z), \eta(z), z)\,h_2(z)\,P_0(dz).
\end{align*}
Moreover, the bound
\[
  \bigl|(h_1(z))^\top\,\partial_a^2 l(\theta(z),\eta(z),z)\,h_2(z)\bigr|
  \;\le\;
  \|\partial_a^2 l\|_\infty\;\|h_1(z)\|_{\mathbb{R}^{d_1}}\;\|h_2(z)\|_{\mathbb{R}^{d_1}},
\]
together with $h_1,h_2\in\mathcal{H}$, allows an identical Dominated Convergence argument to interchange differentiation and integration when computing $\partial^2_\theta L_0$.

Since, by assumption, $\partial_a^2 l(a,b,z)$ is Lipschitz in \emph{both} $a$ and $b$ on $\mathcal C$, there exists $L>0$ such that for any $\theta,\eta$ and perturbations $h,g$,
\[
\bigl\|\partial_a^2 l(\theta(z)+h(z),\,\eta(z)+g(z),\,z)
-\partial_a^2 l(\theta(z),\,\eta(z),\,z)\bigr\|
\;\le\;
L\bigl(\|h(z)\|_{\mathbb{R}^{d_1}}+\|g(z)\|_{\mathbb{R}^{d_2}}\bigr).
\]
Hence, for any directions \(h_1,h_2\),
\begin{align*}
&\bigl|\partial^2_\theta L_0(\theta + h,\;\eta + g)(h_1,h_2)
-\partial^2_\theta L_0(\theta,\;\eta)(h_1,h_2)\bigr|\\
&\quad=
\Bigl|\int h_1(z)^\top
\bigl[\partial_a^2 l(\theta(z)+h(z),\eta(z)+g(z),z)
-\partial_a^2 l(\theta(z),\eta(z),z)\bigr]
h_2(z)\,P_0(dz)\Bigr|\\
&\quad\le
\int \|h_1(z)\|_{\mathbb{R}^{d_1}}\,
L\bigl(\|h(z)\|_{\mathbb{R}^{d_1}}+\|g(z)\|_{\mathbb{R}^{d_2}}\bigr)\,
\|h_2(z)\|_{\mathbb{R}^{d_1}}\,P_0(dz)\\
&\quad\le 
L\,\bigl(\|h\|_{\mathcal{H}} + \|g\|_{\mathcal{N}}\bigr)\,
\rho_{\mathcal{H}}(h_1)\,
\|h_2\|_{\mathcal{H}},
\end{align*}
where the final inequality follows from the Cauchy–Schwarz inequality, which establishes the uniform Lipschitz bound in both \(\theta\) and \(\eta\) required by Condition~\ref{cond::targetsmoothloss::three}. To establish Fr\'echet differentiability as in \ref{cond::targetsmoothloss::two}, define the remainder
\[
R(h_1,h_2) := \partial_\theta L_0(\theta+h_2, \eta)(h_1) - \partial_\theta L_0(\theta, \eta)(h_1) - \partial_\theta^2 L_0(\theta, \eta)(h_1, h_2).
\]
By Taylor's theorem with integral remainder, we have
\[
\partial_\theta L_0(\theta+h_2, \eta)(h_1) - \partial_\theta L_0(\theta, \eta)(h_1) = \int_0^1 \partial_\theta^2 L_0(\theta + t h_2, \eta)(h_1, h_2) \, dt,
\]
so that
\[
R(h_1,h_2) = \int_0^1 \Bigl[\partial_\theta^2 L_0(\theta + t h_2, \eta)(h_1, h_2)- \partial_\theta^2 L_0(\theta, \eta)(h_1, h_2)\Bigr]dt.
\]
Using the Lipschitz property of \(\partial_a^2 l\) and taking the essential supremum over \(z\) for \(h_1\) yields
\[
|R(h_1,h_2)| \le \int_0^1 L\,t\, \rho_{\mathcal{H}}(h_1)\, \|h_2\|_{\mathcal{H}}^2\, dt = \frac{L}{2}\, \rho_{\mathcal{H}}(h_1)\, \|h_2\|_{\mathcal{H}}^2.
\]
In particular,
\[
\sup_{\substack{h_1 \in \mathcal{H}_{\mathcal{P}} \\ \rho_{\mathcal{H}}(h_1)\leq1}} \frac{|R(h_1,h_2)|}{\|h_2\|_{\mathcal{H}}} \le \frac{L}{2}\, \|h_2\|_{\mathcal{H}} \to 0 \quad \text{as } \|h_2\|_{\mathcal{H}}\to 0,
\]
which shows that the map \(\theta \mapsto \partial_\theta L_0(\theta, \eta)\) is Fr\'echet differentiable with derivative \(\partial_\theta^2 L_0(\theta, \eta)\), thereby verifying Condition~\ref{cond::targetsmoothloss::two}.

For the cross‐derivative, we have
\begin{align*}
\partial_\eta \partial_\theta L_0(\theta, \eta)(g, h)
&= \frac{d^2}{ds\,dt}\, L_0(\theta + t h, \eta + s g)\Big|_{t=s=0} \\
&= \int h(z)^\top \, \partial_a \partial_b l(\theta(z), \eta(z), z)\, g(z)\, P_0(dz).
\end{align*}
By the assumed Lipschitz continuity of \(\partial_a \partial_b l\) in \((a, b)\), there exists a constant \(C > 0\) such that
\begin{align*}
\bigl|h(z)^\top\bigl[\partial_a \partial_b l(\theta(z) + h_1(z), \eta(z)+ g_1(z), z)-\partial_a \partial_b l(\theta(z), \eta(z), z)\bigr]g(z)\bigr| \\
\le C\, \|h(z)\|_{\infty}\, \|g(z)\|_{\mathbb{R}^{d_2}}\, \left\{\|g_1(z)\|_{\mathbb{R}^{d_2}} + \|h_1(z)\|_{\mathbb{R}^{d_1}}\right\}.
\end{align*}
Hence, integrating both sides over \(z\) with respect to \(P_0\) and applying the Cauchy–Schwarz inequality, we obtain
\begin{align*}
 |\partial_\eta \partial_\theta L_0(\theta +  h_1, \eta + s g_1)(g,h) - \partial_\eta \partial_\theta L_0(\theta, \eta)(g,h)|   \leq L \rho_{\mathcal{H}}( h) \|g\|_{\mathcal{N}} \left\{ \|g_1\|_{\mathcal{N}}  + \| h_1\|_{\mathcal{H}}   \right\}.
\end{align*}
Thus, \((\theta, \eta) \mapsto \partial_\eta \partial_\theta L_0(\theta, \eta)\) is Lipschitz continuous and satisfies \ref{cond::crossderivlipschitz}.

To establish Fr\'echet differentiability, define the remainder for the \(\eta\)–derivative by
\[
R^c(g,h) := \partial_\eta \partial_\theta L_0(\theta, \eta+g)(g, h) - \partial_\eta \partial_\theta L_0(\theta, \eta)(g, h).
\]
By the integral remainder form, we write
\[
\partial_\eta \partial_\theta L_0(\theta, \eta+g)(g, h) - \partial_\eta \partial_\theta L_0(\theta, \eta)(g, h)
=\int_0^1\Bigl[\partial_\eta \partial_\theta L_0(\theta, \eta+s\,g)(g,h)
-\partial_\eta \partial_\theta L_0(\theta, \eta)(g,h)\Bigr]ds.
\]
Thus, by the previous Lipschitz bound and integrating over \(s\),
\[
|R^c(g,h)| \le \int_0^1 C\, s\, \rho_{\mathcal{H}}(h)\, \|g\|_{\mathcal{N}}^2 ds
= \frac{C}{2}\,\rho_{\mathcal{H}}(h)\,\|g\|_{\mathcal{N}}^2.
\]
In particular, for all \(h\) with \(\rho_{\mathcal{H}}(h) \le 1\),
\[
\frac{|R^c(g,h)|}{\|g\|_{\mathcal{N}}} \le \frac{C}{2}\,\|g\|_{\mathcal{N}} \to 0 \quad \text{as } \|g\|_{\mathcal{N}}\to 0.
\]
This shows that the map \(\eta \mapsto \partial_\theta L_0(\theta,\eta)(h)\) is Fr\'echet differentiable with derivative \(\partial_\eta \partial_\theta L_0(\theta,\eta)\), thereby verifying Condition~\ref{cond::crossderiv}, which verifies \ref{cond::nuisancesmooth} .
\end{proof}

\subsection{Technical Lemmas for functional Taylor expansions}
 
\begin{lemma}[Quadratic expansion for target functional]
\label{lemma::functionalTaylor} Suppose that Condition~\ref{cond::smoothfunctional} holds.  Then, uniformly over \(\theta \in \mathcal{H}_{\mathcal{P}}\) and \(h \in \mathcal{H}\), we have the expansion
\[
\psi_0(\theta + h) = \psi_0(\theta) + \dot{\psi}_0(\theta)(h) + \text{Rem}_{\theta}(h),
\]
 where the remainder term satisfies $| \text{Rem}_{\theta}(h)| \leq L  \|h_2\|_{\mathcal{H}}^2$ for some $L < \infty$.
\end{lemma}
\begin{proof}
Fix \(\theta \in \mathcal{H}_{\mathcal{P}}\) and \(h \in \mathcal{H}\).  By Condition~\ref{cond::smoothfunctional} and the integral form of Taylor’s theorem applied to the map \(t \mapsto \psi_0(\theta + t\,h)\), we have
\[
\psi_0(\theta + h) - \psi_0(\theta)
= \int_{0}^{1}\frac{d}{dt}\,\psi_0(\theta + t\,h)\,dt
= \int_{0}^{1}\dot{\psi}_0(\theta + t\,h)(h)\,dt.
\]
Adding and subtracting \(\dot{\psi}_0(\theta)(h)\) inside the integral gives
\[
\psi_0(\theta + h)
= \psi_0(\theta) + \dot{\psi}_0(\theta)(h)
  + \int_{0}^{1}\bigl[\dot{\psi}_0(\theta + t\,h)-\dot{\psi}_0(\theta)\bigr](h)\,dt.
\]
Hence we may identify
\[
\mathrm{Rem}_{\theta}(h)
= \int_{0}^{1}\bigl[\dot{\psi}_0(\theta + t\,h)-\dot{\psi}_0(\theta)\bigr](h)\,dt.
\]
By Condition~\ref{cond::smoothfunctional}(ii), the map \(\theta\mapsto\dot{\psi}_0(\theta)\) is \(L\)-Lipschitz from \(\|\cdot\|_{\mathcal{H}}\) to the operator norm \(\|\cdot\|_{\mathcal{H}}^*\).  Therefore, for each \(t\in[0,1]\),
\[
\bigl|\bigl[\dot{\psi}_0(\theta + t\,h)-\dot{\psi}_0(\theta)\bigr](h)\bigr|
\;\le\;
L\,\|\theta + t\,h - \theta\|_{\mathcal{H}}\,\|h\|_{\mathcal{H}}
= L\,t\,\|h\|_{\mathcal{H}}^2.
\]
Integrating over \(t\) yields
\[
\bigl|\mathrm{Rem}_{\theta}(h)\bigr|
\le \int_{0}^{1}L\,t\,\|h\|_{\mathcal{H}}^2\,dt
= \tfrac12\,L\,\|h\|_{\mathcal{H}}^2
\;\le\;
L\,\|h\|_{\mathcal{H}}^2,
\]
which completes the proof.
\end{proof}

\begin{lemma}[Quadratic expansion for derivative of risk]
\label{lemma::riskTaylor} Suppose that Condition~\ref{cond::targetsmoothloss} holds.  Then, uniformly over \(\theta \in \mathcal{H}_{\mathcal{P}}\),  \(\eta \in \mathcal{N}_{\mathcal{P}}\), and \(h_1, h_2 \in \mathcal{H}\), we have the expansion
\[
\partial_\theta L_0(\theta + h_2, \eta)(h_1) = \partial_\theta L_0(\theta, \eta)(h_1) + \partial_\theta^2 L_0(\theta, \eta)(h_1, h_2) + \text{Rem}_{\theta, \eta}(h_1, h_2),
\]
 where the remainder term satisfies $| \text{Rem}_{\theta, \eta}(h_1, h_2)| \leq L  \rho_{\mathcal{H}}(h_1) \|h_2\|_{\mathcal{H}}^2$ for some $L < \infty$.
\end{lemma}

\begin{proof}
Let $\theta, h_1, h_2 \in \mathcal{H}_{\mathcal{P}}$ and $\eta \in \mathcal{N}_{\mathcal{P}}$. By Condition~\ref{cond::targetsmoothloss::two} and the integral form of Taylor's theorem applied to the map $t \mapsto \partial_\theta L_0(\theta + t h_2,\eta)(h_1),$ we have
\[
\partial_\theta L_0(\theta + h_2, \eta)(h_1) - \partial_\theta L_0(\theta, \eta)(h_1) = \int_0^1 \partial_\theta^2 L_0(\theta + t h_2, \eta)(h_1, h_2) \, dt.
\]
Thus, we may write
\[
\partial_\theta L_0(\theta + h_2,\eta)(h_1) = \partial_\theta L_0(\theta,\eta)(h_1) + \partial_\theta^2 L_0(\theta,\eta)(h_1,h_2) + \text{Rem}(h_1,h_2),
\]
where the remainder is given by
\[
\text{Rem}(h_1,h_2) = \int_0^1 \Bigl[ \partial_\theta^2 L_0(\theta+ t h_2, \eta)(h_1,h_2) - \partial_\theta^2 L_0(\theta,\eta)(h_1,h_2)\Bigr]\,dt.
\]
By Conditions \ref{cond::targetsmoothloss::two} and \ref{cond::targetsmoothloss::three}, there exists a constant \(L < \infty\) such that
\begin{align*}
  \left| \partial_\theta^2 L_0(\theta + t h_2, \eta)(h_1, h_2) - \partial_\theta^2 L_0(\theta, \eta)(h_1, h_2) \right| &\leq t L \|h_2\|_{\mathcal{H}}^2 \rho_{\mathcal{H}}(h_1)
\end{align*}
uniformly over \(t \in [0,1]\) and \(h_1, h_2 \in \mathcal{H}_{\mathcal{P}}\). Hence, the remainder term satisfies $|\text{Rem}(h_1, h_2)| \leq L  \rho_{\mathcal{H}}(h_1) \|h_2\|_{\mathcal{H}}^2.$
\end{proof}

\begin{lemma}[Quadratic expansion for the cross derivative of risk]
\label{lemma::riskTaylorcross}  Suppose that Condition~\ref{cond::nuisancesmooth} holds. Then, uniformly over \(\theta \in \mathcal{H}_{\mathcal{P}}\),  \(\eta \in \mathcal{N}_{\mathcal{P}}\), \(h \in \mathcal{H}\), and \(g \in \mathcal{N}\), we have the expansion
\[
\partial_\theta L_0(\theta,\eta+g)(h) = \partial_\theta L_0(\theta,\eta)(h) + \partial_\eta\partial_\theta L_0(\theta,\eta)(g,h) + \text{Rem}_{\theta, \eta}(g, h),
\]
 where  $|\text{Rem}_{\theta, \eta}(g, h)| \leq L'  \rho_{\mathcal{H}}(h)\, \|g\|_{\mathcal{N}}^2$ for some constant \(L' < \infty\).
\end{lemma}

\begin{proof}
Let \(\theta \in \mathcal{H}_{\mathcal{P}}\), \(\eta \in \mathcal{N}_{\mathcal{P}}\), \(g \in \mathcal{N}\), and \(h \in \mathcal{H}\). By Condition~\ref{cond::nuisancesmooth} and the integral form of Taylor's theorem applied to the map $t \mapsto \partial_\theta L_0(\theta,\eta+t\, g)(h),$ we have
\[
\partial_\theta L_0(\theta,\eta+g)(h) - \partial_\theta L_0(\theta,\eta)(h) = \int_0^1 \partial_\eta\partial_\theta L_0(\theta,\eta+t\, g)(g,h) \, dt.
\]
Thus, we may write
\[
\partial_\theta L_0(\theta,\eta+g)(h) = \partial_\theta L_0(\theta,\eta)(h) + \partial_\eta\partial_\theta L_0(\theta,\eta)(g,h) + \text{Rem}(g,h),
\]
where the remainder is given by
\[
\text{Rem}(g,h) = \int_0^1 \Bigl[\partial_\eta\partial_\theta L_0(\theta,\eta+t\, g)(g,h) - \partial_\eta\partial_\theta L_0(\theta,\eta)(g,h)\Bigr]\,dt.
\]
By Condition~\ref{cond::crossderivlipschitz}, there exists a constant \(L' < \infty\) such that, for each \(t \in [0,1]\),
\[
\Bigl|\partial_\eta\partial_\theta L_0(\theta,\eta+t\, g)(g,h) - \partial_\eta\partial_\theta L_0(\theta,\eta)(g,h)\Bigr| \le t\, L' \|g\|_{\mathcal{N}}^2\, \rho_{\mathcal{H}}(h).
\]
Integrating over \(t\) from 0 to 1 yields
\[
|\text{Rem}(g,h)| \le \int_0^1 t\, L' \|g\|_{\mathcal{N}}^2\, \rho_{\mathcal{H}}(h)\, dt = \frac{L'}{2} \rho_{\mathcal{H}}(h)\, \|g\|_{\mathcal{N}}^2.
\]
Hence, $|\text{Rem}(g,h)| \leq L'  \rho_{\mathcal{H}}(h) \|g\|_{\mathcal{N}}^2$. 
\end{proof}

\subsection{Proof of functional von Mises expansion}

\label{appendix::vonmises}

\begin{lemma}
    Suppose Conditions \ref{cond::uniqueness}-\ref{cond::targetsmoothloss} and \ref{cond::PDHessian} hold. Then, there exists a unique Hessian Riesz representer $\alpha_0 \in \overline{\mathcal{H}}$ such that the linear functional $\dot{\psi}_0(\theta_0)$ admits the representation:$$\dot{\psi}_0(\theta_0)(h) = \partial_\theta^2 L_0(\theta_0, \eta_0)(\alpha_0, h), \quad \forall h \in \mathcal{H}. $$
\end{lemma}
\begin{proof}
    By Conditions~\ref{cond::targetsmoothloss::two} and~\ref{cond::PDHessian}, we have the norm equivalence: there exist constants \(0 < c < C < \infty\) such that $c \|\cdot\|_{\mathcal{H}} \leq \sqrt{\partial_\theta^2 L_0(\theta_0, \eta_0)(\cdot, \cdot)} \leq C \|\cdot\|_{\mathcal{H}},$ so that \(\overline{\mathcal{H}}\) is also closed under the norm induced by the Hessian inner product \(\partial_\theta^2 L_0(\theta_0, \eta_0)\).
 Consequently, by \ref{cond::smoothfunctional}, the functional derivative $\dot{\psi}_0(\theta_0)$ defines a bounded linear functional on $\mathcal{H}$ with respect to the inner product \(\partial_\theta^2 L_0(\theta_0, \eta_0)(\cdot, \cdot)\). By the Riesz representation theorem, there exists a unique Hessian Riesz representer $\alpha_0 \in \overline{\mathcal{H}}$ such that the linear functional $\dot{\psi}_0(\theta_0)$ admits the representation in \eqref{eqn::rieszrep}.
\end{proof}

\begin{proof}[Proof of Theorem \ref{theorem::vonmises}]
Let \(\theta, \alpha \in \mathcal{H}\) and \(\eta \in \mathcal{N}\) with \(\max \{\rho_{\mathcal{H}}(\theta), \rho_{\mathcal{H}}(\alpha), \|\eta\|_{\mathcal{G}, \diamond}\} < M'\) for some \(M' < \infty\). Hereafter, we will absorb \(M'\) into the constants appearing in our big-\(O\) notation.

Applying Lemma \ref{lemma::functionalTaylor}  with \ref{cond::boundedtrue} and \ref{cond::smoothfunctional} at $ \theta_0$ in the direction $h := \theta - \theta_0$, we obtain the Taylor expansion 
\begin{equation}
\begin{aligned}
    \psi_0(\theta) - \psi_0(\theta_0) - P_0\dot{\ell}_{\eta}(\theta)(\alpha)  &= P_0\{m(\cdot, \theta)  - m(\cdot, \theta_0)\} -\partial_{\theta} L_0(\theta, \eta)(\alpha)   \\
    &=\dot{\psi}_0(\theta_0)( \theta - \theta_0)  -  \partial_{\theta} L_0(\theta, \eta)(\alpha)  +  \delta_{\mathrm{lin}}O\left(\|\theta - \theta_0\|_{\mathcal{H},{P_0}}^2\right), \label{eqn::decompfirst}
\end{aligned}
\end{equation}where $\delta_{\mathrm{lin}} = 0$ if the expansion holds without remainder.
Note that by \ref{cond::smoothfunctional}, \ref{cond::PDHessian}, and the Riesz representation property of $\alpha_0$ in \eqref{eqn::rieszrep}, we have $h \mapsto \dot{\psi}_0(\theta_0)(h) = \partial_{\theta}^2 L_0(\theta_0, \eta_0)(\alpha_0, h)$. Thus, the first two terms on the right-hand side satisfy
\begin{equation}
\begin{aligned}
  \dot{\psi}_0(\theta_0)( \theta - \theta_0)  - \partial_{\theta} L_0(\theta, \eta)(\alpha)  &=  \partial_{\theta}^2 L_0(\theta_0, \eta_0)(\alpha_0, \theta - \theta_0)  - \partial_{\theta} L_0(\theta, \eta)(\alpha).\label{eqn::decomp1}
\end{aligned}
\end{equation}
By Lemma \ref{lemma::riskTaylor}  with \ref{cond::boundedtrue} and \ref{cond::targetsmoothloss}, we have the second-order Taylor expansion:
\begin{equation}
\begin{aligned}
 \partial_{\theta} L_0(\theta, \eta)(\alpha) &=\partial_\theta L_0(\theta_0, \eta)(\alpha)  +  \partial_\theta^2 L_0(\theta_0, \eta)(\theta - \theta_0, \alpha)  + \delta_{\mathrm{quad}} O(\|\theta - \theta_0\|_{\mathcal{H}}^2),\label{eqn::second}
\end{aligned}
\end{equation}
where $\delta_{\mathrm{quad}} = 0$ if the expansion holds without remainder.
By Lemma \ref{lemma::riskTaylorcross} with \ref{cond::boundedtrue} and \ref{cond::nuisancesmooth}, we can expand the first term on the right-hand side in $\eta$:
\begin{equation}
\begin{aligned}
  \partial_\theta L_0(\theta_0, \eta)(\alpha)  &=  \partial_\theta L_0(\theta_0, \eta_0)(\alpha)  +  \partial_{\eta} \partial_\theta L_0(\theta_0, \eta_0)(\eta - \eta_0, \alpha)  +  O(\|\eta - \eta_0\|_{\mathcal{N}}^2). \label{eqn::third}
\end{aligned}
\end{equation}
Since $\theta_0$ uniquely minimizes the risk by \ref{cond::uniqueness}, and the risk is smooth by \ref{cond::targetsmoothloss::one}, the first-order optimality condition implies that $\partial_\theta L_0(\theta_0, \eta_0)(\alpha) = 0$. Moreover, $\partial_{\eta} \partial_\theta L_0(\theta_0, \eta_0)(\eta - \eta_0, \alpha) = 0$ by Neyman orthogonality at $(\theta_0, \eta_0)$ in \ref{cond::orthogonal}. Hence, \eqref{eqn::third} implies that $\partial_\theta L_0(\theta_0, \eta)(\alpha)  = O(\|\eta - \eta_0\|_{\mathcal{N}}^2)$.
Therefore, substituting this relation into \eqref{eqn::second}, we have
\begin{align*}
 \partial_{\theta} L_0(\theta, \eta)(\alpha) =  \partial_\theta^2 L_0(\theta_0, \eta)(\theta - \theta_0, \alpha)  + \delta_{\mathrm{quad}}O( \|\theta - \theta_0\|_{\mathcal{H}}^2)+ O(\|\eta - \eta_0\|_{\mathcal{N}}^2). 
\end{align*}
Returning to \eqref{eqn::decomp1}, we obtain
\begin{align*}
     \dot{\psi}_0(\theta_0)( \theta - \theta_0)  - \partial_{\theta} L_0(\theta, \eta)(\alpha)  &=  \partial_{\theta}^2 L_0(\theta_0, \eta_0)(\alpha_0, \theta - \theta_0)  -\partial_\theta^2 L_0(\theta_0, \eta)(\theta - \theta_0, \alpha) \\
     & \quad + \delta_{\mathrm{quad}}O( \|\theta - \theta_0\|_{\mathcal{H}}^2)+ O(\|\eta - \eta_0\|_{\mathcal{N}}^2).
\end{align*}
By Condition \ref{cond::targetsmoothloss::three}, we have the uniform Lipschitz continuity property:
\begin{equation}
\begin{aligned}
    \left|\partial_\theta^2 L_0(\theta_0, \eta)(\theta - \theta_0, \alpha) - \partial_\theta^2 L_0(\theta_0, \eta_0)(\theta - \theta_0, \alpha) \right| & \lesssim \|\eta- \eta_0\|_{\mathcal{H}} \|\theta - \theta_0\|_{\mathcal{H}}\rho_{\mathcal{H}}(\alpha)\\
    & =  O(\|\eta- \eta_0\|_{\mathcal{H}} \|\theta - \theta_0\|_{\mathcal{H}})
\end{aligned} \label{proofeqn::lipschitz}
\end{equation}
Hence,
\begin{align*}
     \dot{\psi}_0(\theta_0)( \theta - \theta_0)  - \partial_{\theta} L_0(\theta, \eta)(\alpha)  &=  \partial_{\theta}^2 L_0(\theta_0, \eta_0)(\alpha_0 - \alpha, \theta - \theta_0) \\
     & \quad + \delta_{\mathrm{quad}}O( \|\theta - \theta_0\|_{\mathcal{H}}^2)+ O(\|\eta - \eta_0\|_{\mathcal{N}}^2)+    O(\|\eta- \eta_0\|_{\mathcal{H}} \|\theta - \theta_0\|_{\mathcal{H}}).
\end{align*}
Combining the previous expression with \eqref{eqn::decompfirst} and \eqref{eqn::decomp1}, this implies that
\begin{align*}
       \psi_0(\theta) - \psi_0(\theta_0) - P_0\dot{\ell}_{\eta}(\theta)(\alpha)  &= \partial_{\theta}^2 L_0(\theta_0, \eta_0)(\alpha_0 - \alpha,\theta - \theta_0)  +  (\delta_{\mathrm{lin}} + \delta_{\mathrm{quad}})O\left(\|\theta - \theta_0\|_{\mathcal{H}}^2\right)  \\
     & \quad   +   O(\|\eta- \eta_0\|_{\mathcal{H}} \|\theta - \theta_0\|_{\mathcal{H}}) + O\left(\|\eta - \eta_0\|^2_{\mathcal{N}} \right).
\end{align*}

\end{proof}

The following theorem shows that when the loss is universally Neyman-orthogonal and the improved Lipschitz bound
\[
\left|\partial_\theta^2 L_0(\theta_0, \eta)(\theta - \theta_0, \alpha) - \partial_\theta^2 L_0(\theta_0, \eta_0)(\theta - \theta_0, \alpha) \right| = O\left(\|\eta - \eta_0\|_{\mathcal{N}}^2\right)
\]
can be established, the cross-product remainder term \(O\left(\|\eta - \eta_0\|_{\mathcal{N}} \|\theta - \theta_0\|_{\mathcal{H}}\right)\) may be dropped from the von Mises expansion.

\begin{theorem}[Von Mises expansion under universal Neyman‐orthogonality]
\label{theorem::neymanorthogonality_secondorder}
Assume the conditions of Theorem~\ref{theorem::vonmises}, and in addition:
\begin{enumerate}[label=(\roman*)]
  \item (\emph{Universal Neyman‐orthogonality}) $\partial_\eta \partial_\theta^2 L_0(\theta_0, \eta_0)
      \bigl(g ,\, h_1 ,\, h_2\bigr)
    = 0
    \quad \text{for all } h, h_1 \in \mathcal{H}, g \in \mathcal{N}.$
  
\item \textit{(Fr\'echet differentiability of the Hessian in $\eta$)} For each $\theta \in \mathcal{H}_{\mathcal{P}}$, the map $\eta \mapsto \partial_\theta^2 L_0(\theta,\eta)$ is Fr\'echet differentiable as a map $(\mathcal{N},\|\cdot\|_{\mathcal N}) \to (\mathcal{H}_{\mathcal{P}}\times\mathcal{H}_{\mathcal{P}},\rho_{\mathcal H} +\rho_{\mathcal H})^*$, with derivative $\partial_\eta \partial_\theta^2 L_0(\theta,\eta)\colon \mathcal{N}\times\mathcal{H}_{\mathcal{P}}\times\mathcal{H}_{\mathcal{P}}\to\mathbb{R}$.

\item \textit{(Lipschitz continuity of cross‐derivative)}  
There exists a constant $C < \infty$ such that, for all $\eta,\eta' \in \mathcal{N}_{\mathcal{P}}$, $n \in \mathcal{N}_{\mathcal{P}}$, and $h_1,h_2 \in \mathcal{H}_{\mathcal{P}}$ with $\|n\|_{\mathcal{N}} +  \rho_{\mathcal{H}}(h_1)  + \rho_{\mathcal{H}}(h_2) \le 1$, we have
\[
\big| \partial_\eta \partial_\theta^2 L_0(\theta_0,\eta')(n,h_1,h_2)
      - \partial_\eta \partial_\theta^2 L_0(\theta_0,\eta)(n,h_1,h_2) \big|
\ \le\ C\,\|\eta' - \eta\|_{\mathcal{N}}.
\]

\end{enumerate}
Then, for all \(\eta \in \mathcal{N}_{\mathcal{P}},\, \theta \in \mathcal{H}_{\mathcal{P}},\, \alpha \in \mathcal{H}\) satisfying $\rho_{\mathcal{H}}(\alpha) < M,$
the following expansion holds:
\begin{align*}
    \psi_0(\theta) - \psi_0(\theta_0) 
    &= \int \dot{\ell}_{\eta}(\theta)(\alpha)(z)\, P(dz) \\
    &\quad + \partial_\theta^2 L_0(\theta_0, \eta_0)(\alpha_0 - \alpha,\, \theta - \theta_0) \\
    &\quad + \left(\delta_{\mathrm{lin}} + \delta_{\mathrm{quad}}\right) O\left(\|\theta - \theta_0\|_{\mathcal{H}}^2\right)
    + O\left(\|\eta - \eta_0\|_{\mathcal{N}}^2\right),
\end{align*}
where the constants in the big-\(O\) notation depend only on \(M\) and universal constants appearing in our assumptions.
\end{theorem}
\begin{proof}
Let $\Delta := \eta - \eta_0$ and define $F(\eta) := \partial_\theta^2 L_0(\theta_0,\eta)(\theta - \theta_0, \alpha)$.  
By Taylor's theorem with integral remainder in $\dot{\mathcal N}$,
\[
F(\eta) - F(\eta_0) = \int_0^1 \partial_\eta F(\eta_0 + s\Delta)[\Delta]\,ds,
\]
where $\partial_\eta F(\eta)[\Delta] = \partial_\eta \partial_\theta^2 L_0(\theta_0,\eta)(\Delta,\,\theta - \theta_0,\,\alpha)$.  
Universal Neyman--orthogonality (i) yields $\partial_\eta F(\eta_0) = 0$, so
\[
F(\eta) - F(\eta_0) = \int_0^1 \{\partial_\eta F(\eta_0 + s\Delta) - \partial_\eta F(\eta_0)\}[\Delta]\,ds.
\]
By Lipschitz continuity (iii), $\|\partial_\eta F(\eta_0 + s\Delta) - \partial_\eta F(\eta_0)\| \le L\,s\,\|\Delta\|_{\mathcal N,P_0}$, hence
\[
|\{\partial_\eta F(\eta_0 + s\Delta) - \partial_\eta F(\eta_0)\}[\Delta]| \le L\,s\,\|\Delta\|_{\mathcal N,P_0} \bigl(\|\Delta\|_{\mathcal N,P_0} + \rho_{\mathcal H}(\theta - \theta_0) + \rho_{\mathcal H}(\alpha)\bigr).
\]
Since $\max\{\rho_{\mathcal H}(\theta - \theta_0),\ \rho_{\mathcal H}(\alpha)\} \le M$, we have  
$|F(\eta) - F(\eta_0)| \le L\,M^2 \int_0^1 s\,\|\Delta\|_{\mathcal N,P_0}^2\,ds = (L M^2/2)\,\|\Delta\|_{\mathcal N,P_0}^2$.  
Thus
\[
\big|\partial_\theta^2 L_0(\theta_0, \eta)(\theta - \theta_0, \alpha) - \partial_\theta^2 L_0(\theta_0, \eta_0)(\theta - \theta_0, \alpha)\big| = O\bigl(\|\eta - \eta_0\|_{\mathcal N,P_0}^2\bigr),
\]
and substituting into \eqref{proofeqn::lipschitz} in Theorem~\ref{theorem::vonmises} shows that the cross-product remainder $O(\|\eta - \eta_0\|_{\mathcal N,P_0}\,\|\theta - \theta_0\|_{\mathcal H,P_0})$ can be dropped.
\end{proof}

\subsection{Proofs for pathwise differentiability and EIF}

We begin by providing an informal yet illustrative proof of Theorem \ref{theorem::EIF}, intentionally being nonrigorous in our use of differentiation. Subsequently, we present a rigorous proof that applies differentiation with greater care.

\begin{proof}[Informal derivation of EIF for Theorem \ref{theorem::EIF}]
For any bounded score $S \in L^2_0(P_0)$, define the smooth submodel through $P_0 \in \mathcal{P}$ as $dP_{0,\varepsilon} =  (1+ \varepsilon S(z)) dP_0(z)$. The first-order optimality conditions of the M-estimand $\theta_{P_{0,\varepsilon}}$ imply that, for each direction $h \in \mathcal{H}$: 
\begin{equation}\partial_\theta L_{P_{0,\varepsilon}}(\theta_{P_{0,\varepsilon}}, \eta_{0, P_{0,\varepsilon}})(h) = 0. \label{eqn::zero}
\end{equation}
Taking the derivative of both sides, we find that
 \begin{equation*}
     \frac{d}{d\varepsilon} \partial_\theta L_{P_{0,\varepsilon}}(\theta_0, \eta_0)(h) \big|_{\varepsilon = 0} + \frac{d}{d\varepsilon} \partial_\theta L_0(\theta_{P_{0,\varepsilon}}, \eta_0)(h) \big|_{\varepsilon = 0} + \frac{d}{d\varepsilon} \partial_\theta L_0(\theta_0, \eta_{P_{0,\varepsilon}})(h) \big|_{\varepsilon = 0} = 0. 
 \end{equation*}
 By calculations given in the rigorous version of this proof, the first term equals $\langle  \dot{\ell}_{\eta_0}(\theta_0)(h), S\rangle$ and the second term equals $ \partial_\theta^2 L_0(\theta, \eta_0)\left(\frac{d}{d\varepsilon}\theta_{P_{0,\varepsilon}} \big |_{\varepsilon =0}\right) $; the final term is zero by Neyman-orthogonality of the loss. Rearranging terms, we find that
 \begin{equation}
   \partial_\theta^2 L_0(\theta, \eta_0)\left(\frac{d}{d\varepsilon}\theta_{P_{0,\varepsilon}} \big |_{\varepsilon =0}, h\right) = - \langle  \dot{\ell}_{\eta_0}(\theta_0)(h), S\rangle_{L^2(P_0)} .   \label{eqn::totalderiv2}
 \end{equation}
Differentiating the functional and applying the chain rule, we have
 \begin{align}
     \frac{d}{d\varepsilon} \psi_{P_{0,\varepsilon}}(\theta_{P_{0,\varepsilon}}) \big |_{\varepsilon = 0} &=  \dot{\psi}_0(\theta_0)\left(\frac{d}{d\varepsilon}\theta_{P_{0,\varepsilon}} \big |_{\varepsilon =0}\right) + \langle m(\cdot, \theta_0) - \Psi(P) , S\rangle_{L^2(P_0)}. \label{eqn::expandpsi}
 \end{align} 
By the Riesz representation property in \eqref{eqn::rieszrep}, we have 
 $$\dot{\psi}_0(\theta_0)\left(\frac{d}{d\varepsilon}\theta_{P_{0,\varepsilon}} \big |_{\varepsilon =0}\right)  = \partial_{\theta}^2 L_0(\theta_0, \eta_0)\left(\frac{d}{d\varepsilon}\theta_{P_{0,\varepsilon}} \big |_{\varepsilon =0}, \alpha_0\right).$$
Applying \eqref{eqn::totalderiv2} and plugging the above into \eqref{eqn::expandpsi}, we conclude
 \begin{align*}
     \frac{d}{d\varepsilon} \psi_{P_{0,\varepsilon}}(\theta_{P_{0,\varepsilon}}) \big |_{\varepsilon = 0}  = \langle  - \dot{\ell}_{\eta_0}(\theta_0)(\alpha_0) + m(\cdot, \theta_0) - \Psi(P), S \rangle_{L^2(P_0)}.
 \end{align*} 
Since the above holds for any bounded score $S$, we can show that $\chi_0 := - \dot{\ell}_{\eta_0}(\theta_0)(\alpha_0) + m(\cdot, \theta_0) - \Psi(P_0)$ is the efficient influence function of $P \mapsto \psi_0(\theta_0)$.
\end{proof}


\begin{proof}[Proof of Theorem \ref{theorem::EIF}]
By \ref{cond::uniqueness}, the M-estimand $\theta_P$ uniquely exists for each $P \in \mathcal{P}$. For any bounded score \(S \in L^2_0(P_0)\), define the smooth submodel through \(P_0\) as \(dP_{0,\varepsilon} = (1 + \varepsilon S(z))\, dP_0(z)\) with $\varepsilon \in (-\delta, \delta)$ for $\delta > 0$ sufficiently small. Since \(\mathcal{P}\) is convex by assumption, \(P_{0,\varepsilon} \in \mathcal{P}\). Condition~\ref{cond::targetsmoothloss::one} and \eqref{eqn::popminimizer} imply that, for \(\varepsilon\) in a neighborhood of zero, \(\theta_{P_{0,\varepsilon}}\) satisfies the following first-order optimality condition:
\begin{equation}\partial_\theta L_{P_{0,\varepsilon}}(\theta_{P_{0,\varepsilon}}, \eta_{0, P_{0,\varepsilon}})(\alpha_0) = 0. \label{eqn::zero}
\end{equation}
Applying the above and then adding and subtracting terms,
\begin{equation}
\begin{aligned}
  0 =  \partial_\theta L_{P_{0,\varepsilon}}(\theta_{P_{0,\varepsilon}}, \eta_{P_{0,\varepsilon}})(\alpha_0)  
    - \partial_\theta L_0(\theta_0, \eta_0)(\alpha_0)  
    &= \left\{\partial_\theta L_{0}(\theta_{0}, \eta_{P_{0,\varepsilon}})(\alpha_0)- \partial_\theta L_0(\theta_0, \eta_0)(\alpha_0) \right\}\\
    &\quad + \left\{ \partial_\theta L_{0}(\theta_{P_{0,\varepsilon}}, \eta_{P_{0,\varepsilon}})(\alpha_0) - \partial_\theta L_{0}(\theta_{0}, \eta_{P_{0,\varepsilon}})(\alpha_0)\right\} \\
    &\quad + \left\{ \partial_\theta L_{P_{0,\varepsilon}}(\theta_{P_{0,\varepsilon}}, \eta_{P_{0,\varepsilon}})(\alpha_0) - \partial_\theta L_{0}(\theta_{P_{0,\varepsilon}}, \eta_{P_{0,\varepsilon}})(\alpha_0)\right\}.
\end{aligned}
\label{eqn::addsub}
\end{equation}
In what follows, we will show that the three differences on the right-hand side of the above are, respectively, $o(\varepsilon)$, $\partial^2_{\theta}L(\theta_0, \eta_0)(\theta_{P_{0, \varepsilon}} - \theta_0, \alpha_0) + o(\varepsilon)$, and $\langle \dot{\ell}_{\eta_0}(\theta_0), S \rangle_{L^2(P_0)} + o(\varepsilon)$. Consequently, \eqref{eqn::addsub} implies that $ \langle \dot{\ell}_{\eta_0}(\theta_0), S \rangle_{L^2(P_0)} = -\partial^2_{\theta}L(\theta_0, \eta_0)(\varepsilon^{-1}(\theta_{P_{0, \varepsilon}} - \theta_0), \alpha_0) + o(1)$.

To this end, by \ref{cond::nuisancesmooth} and Lemma~\ref{lemma::riskTaylorcross}, the first difference on the right-hand side of \eqref{eqn::addsub} satisfies
\begin{align*}
    \partial_\theta L_0(\theta_0, \eta_{P_{0,\varepsilon}})(\alpha_0) - \partial_\theta L_0(\theta_0, \eta_0)(\alpha_0)
    = \partial_{\eta} \partial_\theta L_0(\theta_0, \eta_0)(\eta_{P_{0,\varepsilon}} - \eta_0, \alpha_0) + O\left( \|\eta_{P_{0,\varepsilon}} - \eta_0\|_{\mathcal{N}}^2 \right).
\end{align*}
By \ref{cond::orthogonal}, we have  $\partial_{\eta} \partial_\theta L_0(\theta_0, \eta_0)(\eta_{P_{0,\varepsilon}} - \eta_0, \alpha_0) = 0,$
and by the Lipschitz Hellinger continuity of \( P \mapsto \eta_P \) in Condition~\ref{cond::smoothparam}, as $\varepsilon \rightarrow 0$, $\|\eta_{P_{0,\varepsilon}} - \eta_0\|_{\mathcal{N}}^2 = O(\varepsilon^2) = o(\varepsilon).$
Hence,
\begin{equation}
    \partial_\theta L_0(\theta_0, \eta_{P_{0,\varepsilon}})(\alpha_0) - \partial_\theta L_0(\theta_0, \eta_0)(\alpha_0) = o(\varepsilon).
    \label{eqn::firstterm}
\end{equation}

The second difference on the right-hand side of \eqref{eqn::addsub} satisfies, by \ref{cond::targetsmoothloss} and Lemma \ref{lemma::riskTaylor},
\begin{align*}
    \partial_\theta L_0(\theta_{P_{0,\varepsilon}}, \eta_{P_{0,\varepsilon}})(\alpha_0) 
    - \partial_\theta L_0(\theta_0, \eta_{P_{0,\varepsilon}})(\alpha_0) 
    &= \partial_\theta^2 L_0(\theta_{P_{0,\varepsilon}}, \eta_{P_{0,\varepsilon}})(\theta_{P_{0,\varepsilon}} - \theta_0, \alpha_0) 
    + O\left(\|\theta_{P_{0,\varepsilon}} - \theta_0\|_{\mathcal{H}}^2\right)\\
    &= \partial_\theta^2 L_0(\theta_0, \eta_0)(\theta_{P_{0,\varepsilon}} - \theta_0, \alpha_0)  + O\left(\|\theta_{P_{0,\varepsilon}} - \theta_0\|_{\mathcal{H}}^2\right) \\
    & \quad + \left\{ \partial_\theta^2 L_0(\theta_{P_{0,\varepsilon}}, \eta_{P_{0,\varepsilon}})(\theta_{P_{0,\varepsilon}} - \theta_0, \alpha_0) - \partial_\theta^2 L_0(\theta_0, \eta_0)(\theta_{P_{0,\varepsilon}} - \theta_0, \alpha_0) \right\}.
\end{align*}
By the Lipschitz Hellinger continuity of \( P \mapsto \theta_P \) and $P \mapsto \eta_P$ in Condition~\ref{cond::smoothparam}, we have \( \|\theta_{P_{0,\varepsilon}} - \theta_0\|_{\mathcal{H}} = O(\varepsilon) \) and \( \|\eta_{P_{0,\varepsilon}} - \eta_0\|_{\mathcal{N}} = O(\varepsilon) \). Moreover, combining this with the Lipschitz continuity of the second derivative in Condition~\ref{cond::targetsmoothloss::three}, we obtain
\begin{align*}
    \partial_\theta^2 L_0(\theta_{P_{0,\varepsilon}}, \eta_{P_{0,\varepsilon}})(\theta_{P_{0,\varepsilon}} - \theta_0, \alpha_0) 
    - \partial_\theta^2 L_0(\theta_0, \eta_0)(\theta_{P_{0,\varepsilon}} - \theta_0, \alpha_0) 
    &\lesssim \rho_{\mathcal{H}}(\alpha_0) \|\theta_{P_{0,\varepsilon}} - \theta_0\|_{\mathcal{H}} 
    \left\{ \|\theta_{P_{0,\varepsilon}} - \theta_0\|_{\mathcal{H}} + \|\eta_{P_{0,\varepsilon}} - \eta_0\|_{\mathcal{N}} \right\} \\
    &= O(\varepsilon^2) = o(\varepsilon).
\end{align*}
Thus,
\begin{equation}
\begin{aligned}
    \partial_\theta L_0(\theta_{P_{0,\varepsilon}}, \eta_{P_{0,\varepsilon}})(\alpha_0) 
    - \partial_\theta L_0(\theta_0, \eta_{P_{0,\varepsilon}})(\alpha_0) 
     &= \partial_\theta^2 L_0(\theta_0, \eta_0)(\theta_{P_{0,\varepsilon}} - \theta_0, \alpha_0) + o(\varepsilon).
\end{aligned}
\label{eqn::secondterm}
\end{equation}

The third difference on the right-hand side of \eqref{eqn::addsub} satisfies, by Condition~\ref{cond::targetsmoothloss::one},
\begin{align*}
    \partial_\theta L_{P_{0, \varepsilon}}(\theta_{P_{0,\varepsilon}}, \eta_{P_{0,\varepsilon}})(\alpha_0) 
    - \partial_\theta L_0(\theta_{P_{0,\varepsilon}}, \eta_{P_{0,\varepsilon}})(\alpha_0) 
    &= (P_{0, \varepsilon} - P_0)\dot{\ell}_{\eta_{P_{0,\varepsilon}}}(\theta_{P_{0,\varepsilon}}) \\
    &= \varepsilon \int \dot{\ell}_{\eta_{P_{0,\varepsilon}}}(\theta_{P_{0,\varepsilon}})(z) S(z) \, P_0(dz) \\
    &= \varepsilon \int \dot{\ell}_{\eta_0}(\theta_0)(z) S(z) \, P_0(dz) 
    + \varepsilon \left\langle \dot{\ell}_{\eta_{P_{0,\varepsilon}}}(\theta_{P_{0,\varepsilon}}) - \dot{\ell}_{\eta_0}(\theta_0), S \right\rangle_{L^2(P_0)}.
\end{align*}
By the Hellinger continuity of \( P \mapsto \eta_P \) and $P \mapsto \theta_P$, the continuity of the map \( (\eta, \theta) \mapsto \dot{\ell}_\eta(\cdot, \theta) \) in Condition~\ref{cond::smoothparam}, and the continuity of the inner product, it follows that, as $\varepsilon \rightarrow 0$,
\[
\left\langle \dot{\ell}_{\eta_{P_{0,\varepsilon}}}(\theta_{P_{0,\varepsilon}}) - \dot{\ell}_{\eta_0}(\theta_0), S \right\rangle_{L^2(P_0)} \leq \|\dot{\ell}_{\eta_{P_{0,\varepsilon}}}(\theta_{P_{0,\varepsilon}}) - \dot{\ell}_{\eta_0}(\theta_0)\|_{L^2(P_0)}\|S\|_{L^2(P_0)}  = o(1).
\]
Hence,
\begin{equation}
    \partial_\theta L_{P_{0, \varepsilon}}(\theta_{P_{0,\varepsilon}}, \eta_{P_{0,\varepsilon}})(\alpha_0)  
    - \partial_\theta L_0(\theta_{P_{0,\varepsilon}}, \eta_{P_{0,\varepsilon}})(\alpha_0)  
    = \varepsilon \langle \dot{\ell}_{\eta_0}(\theta_0), S \rangle_{L^2(P_0)} + o(\varepsilon).
    \label{eqn::thirdterm}
\end{equation}

Finally, plugging the expressions in \eqref{eqn::firstterm}, \eqref{eqn::secondterm}, and \eqref{eqn::thirdterm} into \eqref{eqn::addsub}, we obtain
\begin{align*}
  0=\partial_\theta L_{P_{0,\varepsilon}}(\theta_{P_{0,\varepsilon}}, \eta_{P_{0,\varepsilon}})(\alpha_0)  
    - \partial_\theta L_0(\theta_0, \eta_0)(\alpha_0)  
    =  \partial_\theta^2 L_0(\theta_0, \eta_0)(\theta_{P_{0,\varepsilon}} - \theta_0, \alpha_0) + \varepsilon \langle \dot{\ell}_{\eta_0}(\theta_0), S \rangle_{L^2(P_0)} + o(\varepsilon).
\end{align*}
Rearranging terms and dividing both sides by $\varepsilon$ yields
\begin{equation}
\label{eqn::pathwisetotal}
 \partial_\theta^2 L_0(\theta_0, \eta_0)\left(\frac{\theta_{P_{0,\varepsilon}} - \theta_0}{\varepsilon}, \alpha_0\right)  =  \langle  -\dot{\ell}_{\eta_0}(\theta_0), S \rangle_{L^2(P_0)}  + o(1).
\end{equation}

We are now in a position to compute the pathwise derivative of $\Psi$. By \ref{cond::smoothfunctional}, we have that
 \begin{align*}
    \psi_{P_{0,\varepsilon}}(\theta_{P_{0,\varepsilon}}) - \psi_0(\theta_0) &=    \psi_P(\theta_{P_{0,\varepsilon}}) - \psi_0(\theta_0)  +   \psi_{P_{0,\varepsilon}}(\theta_{P_{0,\varepsilon}}) - \psi_0(\theta_{P_{0,\varepsilon}})  \\
     &= \varepsilon \dot{\psi}_0(\theta_0)\left(\varepsilon^{-1}\{\theta_{P_{0,\varepsilon}} - \theta_0\} \right)  + o\left( \|\theta_{P_{0,\varepsilon}} - \theta_0\|_{\mathcal{H}}\right) + \varepsilon \langle m(\cdot, \theta_{P_{0,\varepsilon}}) , S\rangle_{L^2(P_0)},
 \end{align*} 
 where we use that $\psi_{P_{0,\varepsilon}}(\theta_{P_{0,\varepsilon}}) - \psi_0(\theta_{P_{0,\varepsilon}}) = \varepsilon \langle m(\cdot, \theta_{P_{0,\varepsilon}}) , S\rangle_{L^2(P_0)}$ by the definition of $\psi_P$ and the choice of submodel. Moreover, by the Cauchy-Schwarz inequality and \ref{cond::smoothparam}, we have $\|\theta_{P_{0,\varepsilon}} - \theta_0\|_{\mathcal{H}} = O(\varepsilon)$ and $\langle m(\cdot, \theta_{P_{0,\varepsilon}}) -  m(\cdot, \theta_0 )  , S\rangle_{L^2(P_0)}\leq \| m(\cdot, \theta_{P_{0,\varepsilon}}) -  m(\cdot, \theta_0) \|_{L^2(P_0)} \|S\|_{L^2(P_0)}  = o(1)$ as $\varepsilon \rightarrow 0$ . Hence,
  \begin{align*}
    \psi_{P_{0,\varepsilon}}(\theta_{P_{0,\varepsilon}}) - \psi_0(\theta_0) &=    \varepsilon \dot{\psi}_0(\theta_0)\left(\varepsilon^{-1}\{\theta_{P_{0,\varepsilon}} - \theta_0\} \right) + \varepsilon \langle m(\cdot, \theta_0 )  , S\rangle_{L^2(P_0)} + o(\varepsilon).
 \end{align*} 
 Conditions \ref{cond::targetsmoothloss::two} and \ref{cond::PDHessian} imply that \(\partial_\theta^2 L_0(\theta_0, \eta_0)(\cdot, \cdot)\) defines an inner product on \(\overline{\mathcal{H}}\), and Condition \ref{cond::smoothfunctional} combined with \ref{cond::PDHessian} ensures that the functional derivative \(\dot{\psi}_0(\theta_0)\) is continuous with respect to this inner product. Hence, the Riesz representer $\alpha_0$ exists, and it satisfies the Riesz representation property:
 $$\dot{\psi}_0(\theta_0)(h')  = \partial_{\theta}^2 L_0(\theta_0, \eta_0)(h', \alpha_0).$$
Taking $h' =\varepsilon^{-1}\{\theta_{P_{0,\varepsilon}} - \theta_0\}$ in the above display, we find that
  \begin{align*}
    \psi_{P_{0,\varepsilon}}(\theta_{P_{0,\varepsilon}}) - \psi_0(\theta_0) &=    \varepsilon \partial_{\theta}^2 L_0(\theta_0, \eta_0)\left(\varepsilon^{-1}\{\theta_{P_{0,\varepsilon}} - \theta_0\}, \alpha_0\right) + \varepsilon \langle m(\cdot, \theta_0 )  , S\rangle_{L^2(P_0)} + o(\varepsilon).
 \end{align*} 
Substituting \eqref{eqn::pathwisetotal} into the above expression, this implies that
  \begin{align*}
    \psi_{P_{0,\varepsilon}}(\theta_{P_{0,\varepsilon}}) - \psi_0(\theta_0) &=    \varepsilon \langle  -  \dot{\ell}_{\eta_P}(\theta_0)(\alpha_0) , S\rangle_{L^2(P_0)}   + \varepsilon \langle m(\cdot, \theta_0 )  , S\rangle_{L^2(P_0)} + o(\varepsilon).
 \end{align*} 
Hence,
  \begin{align*}
   \left.\lim_{\varepsilon \rightarrow 0} \frac{\psi_{P_{0,\varepsilon}}(\theta_{P_{0,\varepsilon}}) - \psi_0(\theta_0)}{\varepsilon} \right |_{\varepsilon = 0} &= \langle  -  \dot{\ell}_{\eta_P}(\theta_0)(\alpha_0) , S\rangle_{L^2(P_0)}  +   \langle m(\cdot, \theta_0 ) - \Psi(P_0)  , S\rangle_{L^2(P_0)}  \\
    &= \langle  - \dot{\ell}_{\eta_P}(\theta_0)(\alpha_0) + m(\cdot, \theta_0) - \Psi(P_0), S \rangle_{L^2(P_0)}.
 \end{align*} 
The above holds for any bounded score $S$. Since the linear space of uniformly bounded scores is dense in $L^2_0(P)$ and $\Psi$ is Hellinger Lipschitz continuous by \ref{cond::smoothparam}, the above also holds, by continuity of the inner product, for any score $S \in L^2_0(P)$ \citep[Lemma 2 of][]{luedtke2023one}. We conclude that $ \chi_0 = - \dot{\ell}_{\eta_0}(\theta_0)(\alpha_0) + m(\cdot, \theta_0) - \Psi(P_0)$ is a gradient at $P_0$ for the pathwise derivative of $\Psi$. Since the model is nonparametric, $\chi_0 $ is the efficient influence function.

\end{proof}

\subsection{Pathwise derivative of M-estimand}

In the following theorem, let $\dot{\ell}_{\eta_0}^*(\theta_0): L^2_0(P_0) \rightarrow \overline{\mathcal{H}}$ be the adjoint of $\dot{\ell}_{\eta_0}(\theta_0): \overline{\mathcal{H}} \rightarrow L^2_0(P_0)$, where  $\overline{\mathcal{H}}$ is the closure of $\mathcal{H}$ equipped with the inner product $\partial_\theta^2 L_0(\theta_0, \eta_0)$.

\begin{theorem}[Pathwise derivative of M-estimand]
    Assume conditions \ref{cond::uniqueness}-\ref{cond::boundedtrue}. Suppose that $P \mapsto \theta_P$ is pathwise differentiable at $P_0$ in a Hilbert-valued sense \citep{luedtke2023one} with derivative $\dot{\theta}_0: L^2_0(P_0)\rightarrow \overline{\mathcal{H}}$. Then, 
    $\dot{\theta}_0  =  -\dot{\ell}_{\eta_0}^*(\theta_0) $  and $\partial_\theta^2 L_0(\theta_0, \eta_0)(\dot{\theta}_0(S), h) = \langle - \dot{\ell}_{\eta_0}(\theta_0)(h), S \rangle_{L^2(P_0)}$.
\end{theorem}
\begin{proof}
From \eqref{eqn::pathwisetotal} in the proof of Theorem \ref{theorem::EIF}, we know, as $\varepsilon \rightarrow 0$ and for all $h \in \overline{\mathcal{H}}$, that
\begin{align*}
 \partial_\theta^2 L_0(\theta_0, \eta_0)\left(\frac{\theta_{P_{0,\varepsilon}} - \theta_0}{\varepsilon}, h\right)  =  \langle  -\dot{\ell}_{\eta_0}(\theta_0), S \rangle_{L^2(P_0)}  + o(1)
\end{align*} 
Since $P \mapsto \theta_P$ is pathwise differentiable, we can take the limit as $\varepsilon \rightarrow 0$ to find
$$\partial_\theta^2 L_0(\theta_0, \eta_0)\left(\dot{\theta}_0(S), h\right) =  \langle -  \dot{\ell}_{\eta_0}(\theta_0)(h) , S\rangle_{L^2(P_0)}. $$
Then, by the definition of the adjoint, 
$$\partial_\theta^2 L_0(\theta_0, \eta_0)\left(\dot{\theta}_0(S), h\right)  =   \partial_\theta^2 L_0(\theta_0, \eta_0)\left( -\dot{\ell}_{\eta_0}^*(\theta_0)(S), h\right). $$
Since the above holds for all $h \in \overline{\mathcal{H}}$, we conclude that
$\dot{\theta}_0(S)   =  -\dot{\ell}_{\eta_0}^*(\theta_0)(S). $
\end{proof}

Assuming that $P \mapsto \theta_P$ is pathwise differentiable at $P_0$, we can derive the EIF of a smooth functional $\psi$ of $\theta_0$ using the chain rule \citep{luedtke2023one}. Specifically, by the chain rule, the pathwise derivative of $P \mapsto \psi(\theta_0)$ is $\dot{\psi}(\theta_0) \circ \dot{\theta}_0 = - \dot{\psi}(\theta_0) \circ \dot{\ell}_{\eta_0}^*(\theta_0)$. Suppose that $\dot{\psi}(\theta_0)(h) = \partial_\theta^2 L_0(\theta_0, \eta_0)(\alpha_0, h)$ for a Riesz representer $\alpha_0$. Then, 
\begin{align*}
    (\dot{\psi}(\theta_0)\circ \dot{\theta}_0)(S) &= \partial_\theta^2 L_0(\theta_0, \eta_0)\left( \alpha_0, -\dot{\ell}_{\eta_0}^*(\theta_0)(S) \right)\\
    &= \langle- \dot{\ell}_{\eta_0}(\theta_0)(\alpha_0), S \rangle_{L^2(P_0)}
\end{align*} 
where $-\dot{\ell}_{\eta_0}(\theta_0)(\alpha_0)$ is the efficient influence function of $P \mapsto \psi(\theta_P)$ at $P_0$.

\section{Proofs for Section \ref{sec::autodml}}

Denote the EIF estimate corresponding to $\theta_n$, $\eta_n$, and $\alpha_n$ by $\chi_n := m(\cdot, \theta_n) - P_n m(\cdot, \theta_n) - \dot{\ell}_{\eta_n}(\theta_n)(\alpha_n)$. For each $P \in \mathcal{P}$, $\eta \in \mathcal{N},$ and $\theta, \alpha \in \mathcal{H}$, define the remainder $\text{Rem}_{P_0}(\eta, \theta, \alpha) :=  \psi_0(\theta) - \psi_0(\theta_0) - P_0\dot{\ell}_{\eta}(\theta)(\alpha)$, which by the von Mises expansion in Theorem \ref{theorem::vonmises} is second-order.

\begin{lemma}[Decomposition of bias of one-step estimator]
\label{lemma::biasexpansion} Assume \ref{cond::uniqueness}-\ref{cond::boundedtrue} hold at $P := P_0$. Then,
  $$   \widehat{\psi}_n^{\mathrm{dml}} - \Psi(P_0) = P_n \chi_0 + (P_n-P_0)(\chi_n - \chi_0) + \text{Rem}_{P_0}(\eta_n, \theta_n, \alpha_n).  $$
\end{lemma}
\begin{proof}[Proof of Lemma \ref{lemma::biasexpansion}]
Under conditions \ref{cond::uniqueness}-\ref{cond::boundedtrue}, $\Psi$ is a pathwise differentiable parameter at $P_0$ with EIF $\chi_0$ provided in Theorem \ref{theorem::EIF}. By definition of $ \text{Rem}_{P_0}(\eta_n, \theta_n, \alpha_n)$, we have the following von Mises expansion of the bias:
\begin{align*}
    P_n m(\cdot, \theta_n)  - \Psi(P_0) &= - P_0 \chi_n +   \text{Rem}_{P_0}(\eta_n, \theta_n, \alpha_n),
\end{align*} 
noting that $ \text{Rem}_{P_0}(\eta_n, \theta_n, \alpha_n) = P_n m(\cdot, \theta_n) - \Psi(P_0) + P_0 \chi_n$.
Adding $P_n \chi_n$ to both sides, we find that
\begin{align*}
    \widehat{\psi}_n^{\mathrm{dml}} &=  P_n m(\cdot, \theta_n) + P_n \chi_n - \Psi(P_0)  \\
    &= (P_n - P_0) \chi_n +  R_{n,2}(P_0) \\
       &= (P_n - P_0) \chi_0 + (P_n - P_0)(\chi_n - \chi_0) +   \text{Rem}_{P_0}(\eta_n, \theta_n, \alpha_n).
\end{align*} 
The result follows noting that $(P_n - P_0) \chi_0 = P_n \chi_0$, since $P_0 \chi_0 = 0$.
\end{proof}

\begin{proof}[Proof of Theorem \ref{theorem::limitautoDML}]
    By Lemma \ref{lemma::biasexpansion}, we have the bias expansion
    \[
        \widehat{\psi}_n^{\mathrm{dml}} - \Psi(P_0) = P_n \chi_0 + (P_n - P_0)(\chi_n - \chi_0) +  \text{Rem}_{P_0}(\eta_n, \theta_n, \alpha_n).
    \]
    By \ref{cond::empproc}, we have $(P_n - P_0)(\chi_n - \chi_0) = o_p(n^{-1/2})$, so that
    \[
        \widehat{\psi}_n^{\mathrm{dml}} - \Psi(P_0) = P_n \chi_0 +  \text{Rem}_{P_0}(\eta_n, \theta_n, \alpha_n) + o_p(n^{-1/2}).
    \]
    Next, by \ref{cond::boundnuis}, we can apply Theorem \ref{theorem::vonmises} to conclude that
    \begin{align*}
         \text{Rem}_{P_0}(\eta_n, \theta_n, \alpha_n)  &= O_p\left(\partial_{\theta}^2 L_0(\theta_0, \eta_0)(\alpha_0 - \alpha_n, \theta_n - \theta_0)\right) \\
        &\quad + (\delta_{\mathrm{lin}} + \delta_{\mathrm{quad}}) O_p\left(\|\theta_n - \theta_0\|_{\mathcal{H},{0}}^2\right) + O_p\left(\|\eta_n - \eta_0\|^2_{\mathcal{N}}\right)\\
        & \quad + + O_p\left(\|\eta_n - \eta_0\|_{\mathcal{N}} \|\theta_n - \theta_0\|_{\mathcal{H}}\right)
    \end{align*}
    By \ref{cond::nuisancerate}-\ref{cond::DRrate}, it holds that $ \text{Rem}_{P_0}(\eta_n, \theta_n, \alpha_n) = o_p(n^{-1/2})$. We conclude that
    \[
        \widehat{\psi}_n^{\mathrm{dml}} - \Psi(P_0) = P_n \chi_0 + o_p(n^{-1/2}).
    \]
    Thus, $\widehat{\psi}_n^{\mathrm{dml}}$ is an asymptotically linear estimator of $\Psi(P_0)$ with its influence function being the EIF of $\Psi$. It follows that $\widehat{\psi}_n^{\mathrm{dml}}$ is a regular and efficient estimator of $\Psi$.
\end{proof}

\subsection{Proof of Theorem \ref{theorem::admlbias}}

\begin{proof}[Proof of Theorem \ref{theorem::admlbias}]
By \ref{cond::smoothfunctional} and \ref{cond::targetsmoothloss}, we have
\begin{equation}
    \begin{aligned}
        \Psi_n(P_0) - \Psi_0(P_0) &= \psi_0(\theta_{0, \mathcal{H}_n}) - \psi_0(\theta_0) \\
        &= \dot{\psi}_0(\theta_0)(\theta_{0, \mathcal{H}_n} - \theta_0) + \delta_{\mathrm{lin}} O_p\left(\|\theta_{0, \mathcal{H}_n} - \theta_{0}\|_{\mathcal{H}}^2\right) \\
        &= \partial_{\theta}^2 L_0(\theta_0, \eta_0)(\alpha_{0, \mathcal{H}_{n,0}}, \theta_{0, \mathcal{H}_n} - \theta_0) + \delta_{\mathrm{lin}} O_p\left(\|\theta_{0, \mathcal{H}_n} - \theta_{0}\|_{\mathcal{H}}^2\right),
    \end{aligned} \label{eqn::adml1}
\end{equation}
where the final equality uses the Riesz representation property of $\alpha_{0, \mathcal{H}_{n,0}}$ and that $\theta_{0, \mathcal{H}_n} - \theta_0 \in \mathcal{H}_{n,0}$.

Next, note that the first-order conditions defining the risk minimizer $\theta_{0, \mathcal{H}_n}$ imply that
\begin{align*}
    \partial_{\theta}L_0(\theta_{0, \mathcal{H}_n}, \eta_0)(h) = 0 \text{ for all } h \in \mathcal{H}_{n}.
\end{align*}
Applying \ref{cond::targetsmoothloss} and taking a derivative, we find, for all $h \in \mathcal{H}_n$, 
\begin{align*}
    0 = \partial_{\theta}L_0(\theta_{0, \mathcal{H}_n}, \eta_0)(h) = \partial_{\theta}^2 L_0(\theta_{0}, \eta_0)(\theta_{0, \mathcal{H}_n} - \theta_0, h) + \delta_{\mathrm{quad}} O_p\left(\|\theta_{0, \mathcal{H}_n} - \theta_{0}\|_{\mathcal{H}}^2\right).
\end{align*}
Taking $h$ equal to $\alpha_{0, \mathcal{H}_n}$, we find that
\begin{align*}
    \partial_{\theta}^2 L_0(\theta_{0}, \eta_0)(\theta_{0, \mathcal{H}_n} - \theta_0, \alpha_{0, \mathcal{H}_n}) = \delta_{\mathrm{quad}} O_p\left(\|\theta_{0, \mathcal{H}_n} - \theta_{0}\|_{\mathcal{H}}^2\right).
\end{align*}
Combining the above expression with the final equality in \eqref{eqn::adml1}, we conclude that 
\begin{align*}
    \Psi_n(P_0) - \Psi_0(P_0) &= \partial_{\theta}^2 L_0(\theta_0, \eta_0)(\alpha_{0, \mathcal{H}_{n,0}} - \alpha_{0, \mathcal{H}_n}, \theta_{0, \mathcal{H}_n} - \theta_0) + (\delta_{\mathrm{lin}} + \delta_{\mathrm{quad}}) O_p\left(\|\theta_{0, \mathcal{H}_n} - \theta_{0}\|_{\mathcal{H}}^2\right),
\end{align*}
as desired.
\end{proof}

\section{autoDML for beta-geometric survival model}

\label{appendix::betageo}

\subsection{Functional and loss}

 \label{appendix::betageooverview}

Following our autoDML framework, for $P \in \mathcal{P}$, we define our M-estimand $\theta_P$ as the tuple $(a_P, b_P)$ and our target space as the product Hilbert space $\mathcal{H} := \mathcal{H}_a \times \mathcal{H}_b$, where $\mathcal{H}_a \subseteq L^2(\mu_{X})$ and $\mathcal{H}_b \subseteq L^2(\mu_{X})$ for some measure $\mu_{X}$ on $\text{supp}(X)$. Under the conditional ignorability of the censoring mechanism, the log-transformed shape parameters $\theta_P = (a_P , b_P)$ are identified via the following optimization problem: 
\[
\theta_P = \argmin_{(a, b) \in \mathcal{H}} E_P[\ell(Z, a, b)],
\]
where $\ell$ is the negative log-likelihood loss given by:
\[
    \ell: (z, a, b) \mapsto 
    - \delta \log\big(\lambda_{a,b}(t \mid x)\big) 
    - \sum_{s=1}^{t-1} \log\big(1 - \lambda_{a,b}(s \mid x)\big),
\]
where we use $z := (x, t, \delta)$ to denote a realization of $Z$. Given a reference time $t_0 > 0$, we define our target parameter as $\psi_P(\theta_P) = E_P[m(Z, \theta_P)]$, where the functional $m: \mathcal{Z} \times \mathcal{H} \rightarrow \mathbb{R}$ is given by:
\[
m(z, a, b) := \prod_{s=1}^{t_0} \big(1 - \lambda_{a,b}(s \mid x)\big),
\]
which represents the mean survival probability $P(T > t_0)$ under the beta-geometric model. In Appendix~\ref{appendix::betageo::derive}, we derive the first and second derivatives \(\dot{\ell}(\theta)\) and \(\ddot{\ell}(\theta)\) of the loss \(\ell\) at \(\theta \in \mathcal{H}\), as well as the first derivative \(\dot{m}_{\theta}\) of the functional \(m\), which together determine the loss function for the Riesz representer \(\alpha_0 \in \mathcal{H}\) in Theorem~\ref{theorem::EIF}. Since we use the negative log-likelihood loss, as discussed below Theorem~\ref{theorem::limitautoDML}, the autoDML estimator is not only nonparametrically efficient for the projection parameter \(\Psi\), but also semiparametrically efficient for \(P_0(T > t_0)\) under the beta-geometric model.

\subsection{Derivation of gradients and Hessians}

\label{appendix::betageo::derive}

It is shown in \cite{hubbard2021beta} that the conditional distribution of $T$ given $X$ under a beta-geometric distribution $P$ is determined by the following recursions:
\begin{align*}
    P(T = 1 \mid X)&= \frac{\alpha_P(X)}{\alpha_P(X) + \beta_P(X)};\\
      P_0(T = t \mid X)&= \frac{\beta_P(X) + t - 2}{\alpha_P(X) + \beta_P(X) + t - 1}  P(T = t - 1 \mid X);\\
     P(T > 1 \mid X)&= \frac{\beta_P(X)}{\alpha_0(X) + \beta_P(X)};\\
    P(T > t \mid X)&= \frac{\beta_P(X) + t - 1}{\alpha_P(X) + \beta_P(X) + t - 1}  P(T > t - 1 \mid X),\\
\end{align*}
where $\alpha_P(X) = \exp a_P(X)$ and $\beta_P(X) = \exp b_P(X)$. Recall that our M-estimand is $\theta_P := (a_P, b_P)$, where $\theta_P$ lies in the product space $\mathcal{H} := \mathcal{H}_a \times \mathcal{H}_b$, where $\mathcal{H}_a \subseteq L^2(\mu_{X})$ and $\mathcal{H}_b \subseteq L^2(\mu_{X})$ for some measure $\mu_{X}$ on $\text{supp}(X)$. Under noninformative censoring, the log-transformed shape parameters $a_P$ and $b_P$ are identified via the following optimization problem: 
\[
(a_P, b_P) = \argmin_{(a, b) \in \mathcal{H}} E_P[\ell(Z, a, b)],
\]
where $\ell$ is the negative log-likelihood loss and can be written as:
\begin{align*}
    \ell: (z, a, b) \mapsto - \delta \log P_{a, b}(T = t \mid X = x) - (1 - \delta) \log P_{a, b}(T > t \mid X = x),
\end{align*}
and $P_{a, b}$ denotes the distribution of a beta-geometric random variable with shape parameters $\alpha(\cdot) = \exp a(\cdot)$ and $\beta(\cdot) = \exp b(\cdot)$.

Recall our parameter of interest $\Psi(P) = \psi_P(\theta_P) $ is the mean survival probability $P(T > t_0)$ at some time $t_0 \in \mathbb{N}$.We can write the functional $\psi_P: \mathcal{H} \rightarrow \mathbb{R}$ as
\[
\theta \mapsto \psi_P(\theta) := P_{\theta}(T > t_0) = E_P[m(Z, \theta)],
\]
where $m: (z, \theta) \mapsto P_{\theta}(T > t_0 \mid X = x)$ is a nonlinear functional. The derivative of this functional at $\theta$ in the direction of $h \in \mathcal{H}$ satisfies the expression:
\[
\dot{m}_{\theta}(h): z \mapsto \partial_{\theta} P_{\theta}(T > t_0 \mid X = x)(h) = P_{\theta}(T > t_0 \mid X = x) \cdot \partial_{\theta} \log P_{\theta}(T > t_0 \mid X = x)(h),
\]
where $\partial_{\theta} \log P_{\theta}(T > t_0 \mid X = x)(h)$ is determined by the recursions given later in this section.

The derivative of the functional $m$, as well as the gradient and Hessian operators of the loss function $\ell$, can be computed using the recursion relations derived in Appendix A.2 of \cite{hubbard2021beta}. For completeness, we restate these recursions here and use them to derive the relevant derivative terms. Let $\theta = (a, b) \in \mathcal{H}$ and set $\alpha := \exp a$ and $\beta := \exp b$. For a direction $(h_a, h_b) \in \mathcal{H}$, the derivative of the loss function $\ell$ at $\theta$ is given by:
\[
\dot{\ell}_{\theta}(h_a, h_b)(z) := - \delta \partial_{\theta} \log P_{\theta}(T = t \mid X = x)(h_a, h_b) 
- (1 - \delta) \partial_{\theta} \log P_{\theta}(T > t \mid X = x)(h_a, h_b).
\]
Specifically, the derivative of the log event probabilities satisfies:
 \begin{align*}
    \partial_\theta \log P_{\theta}(T = 1 \mid X)(h_a, h_b) &= \frac{\beta(X)}{\alpha(X) + \beta(X)} h_a(X) 
    - \frac{\beta(X)}{\alpha(X) + \beta(X)} h_b(X), \\
    \partial_\theta \log P_{\theta}(T = t \mid X)(h_a, h_b) &= \partial_\theta \log P_{\theta}(T = t - 1 \mid X)(h_a, h_b) \\
    &\quad - \frac{\alpha(X)}{\alpha(X) + \beta(X) + t - 1} h_a(X) \\
    &\quad + \frac{(\alpha(X) + 1)\beta(X)}{(\beta(X) + t - 2)(\alpha(X) + \beta(X) + t - 1)} h_b(X).
\end{align*}
Similarly, the derivative of the log survival probabilities satisfies:
\begin{align*}
    \partial_\theta \log P_{\theta}(T > 1 \mid X)(h_a, h_b) &= - \frac{\alpha(X)}{\alpha(X) + \beta(X)} h_a(X) 
    + \frac{\alpha(X)}{\alpha(X) + \beta(X)} h_b(X), \\
    \partial_\theta \log P_{\theta}(T > t \mid X)(h_a, h_b) &= \partial_\theta \log P_{\theta}(T > t - 1 \mid X)(h_a, h_b) \\
    &\quad - \frac{\alpha(X)}{\alpha(X) + \beta(X) + t - 1} h_a(X) \\
    &\quad + \frac{\alpha(X) \beta(X)}{(\beta(X) + t - 1)(\alpha(X) + \beta(X) + t - 1)} h_b(X).
\end{align*}

Furthermore, for two directions $\theta_1: (h_a, h_b) \in \mathcal{H}$ and $\theta_2 := (g_a, g_b) \in \mathcal{H}$, the second derivative of the loss function $\ell$ is given by:
\[
\ddot{\ell}_{\theta}(\theta_1, \theta_2)(z) := - \delta \partial_{\theta}^2 \log P_{\theta}(T = t \mid X=x)(\theta_1, \theta_2) 
- (1 - \delta)  \partial_{\theta}^2 \log P_{\theta}(T > t \mid X=x)(\theta_1, \theta_2).
\]
To determine the second derivatives on the right-hand side, we use the fact that, for any functional $\theta \mapsto f_{\theta}(t, a, x)$, its second derivative satisfies the Jacobian formula:
\begin{align*}
    \partial_{\theta}^2 f_{\theta}(t, a, x)(\theta_1, \theta_2) &= \partial_{a}^2 f_{\theta}(t, a, x)(h_a, g_a)   + \partial_{b}^2 f_{\theta}(t, a, x)(h_b, g_b) \\
    &\quad + \partial_{a} \partial_b f_{\theta}(t, a, x)(h_a, g_b) + \partial_b \partial_{a} f_{\theta}(t, a, x)(h_b, g_a).
\end{align*}
Applying the chain rule and using that  $\partial_a\alpha(X) = \partial_a \exp \{a(X) \} = \alpha(X)$ and $\partial_b\beta(X) = \partial_b \exp \{b(X) \} = \beta(X)$, we can show that
\begin{align*}
    \partial_a^2 \log P_{\theta}(T =1 \mid X=x)(h_a, g_a)  &=   -\frac{\alpha(X)\beta(X)}{\{\alpha(X) + \beta(X)\}^2} h_a(X) g_a(X); \\
    \partial_a \partial_b \log P_{\theta}(T =1 \mid X=x)(h_a, g_b)  &=   \frac{\alpha(X)\beta(X)}{\{\alpha(X) + \beta(X)\}^2} h_a(X) g_a(X);\\
   \partial_b     \partial_a \log P_{\theta}(T =1 \mid X=x)(h_b, g_a)  &=   \frac{\alpha(X)\beta(X)}{\{\alpha(X) + \beta(X)\}^2} h_b(X) g_a(X) ;\\
 \partial_b^2 \log P_{\theta}(T =1 \mid X=x)(h_b, g_b)  &=   -\frac{\alpha(X)\beta(X)}{\{\alpha(X) + \beta(X)\}^2} h_b(X) g_b(X).
\end{align*} 
Furthermore, applying the recursions, we find that
\begin{align*}
\partial_a^2& \log P_{\theta}(T =t  \mid X=x)(h_a, g_a)   =    \partial_a^2 \log P_{\theta}(T =t - 1  \mid X=x)(h_a, g_a)\\
    & \quad \hspace{7cm} + \left[\frac{-\alpha(X)(\beta(X) + t - 1)}{(\alpha(X) + \beta(X) + t - 1)^2} \right] h_b(X) g_b(X); \\
     \partial_b^2& \log P_{\theta}(T =t  \mid X=x)(h_b, g_b)  =    \partial_b^2 \log P_{\theta}(T =t - 1  \mid X=x)(h_b, g_b)\\
     & \quad + \beta(X)(\alpha(X) + 1) \cdot \Bigg[
\frac{1}{(\beta(X) + t - 2)(\alpha(X) + \beta(X) + t - 1)} \\
& \hspace{5cm}   - \frac{\beta(X) (2\beta(X) + \alpha(X) + 2t - 3)}{\left[(\beta(X) + t - 2)(\alpha(X) + \beta(X) + t - 1)\right]^2}
\Bigg] h_b(X)g_b(X)\\
\partial_a&  \partial_b   \log P_{\theta}(T =t  \mid X=x)(h_a, g_b)   =   \partial_a \partial_b  \log P_{\theta}(T =t - 1  \mid X=x)(h_a, g_b)\\
    & \quad \hspace{7cm} +  \frac{\alpha(X)\beta(X)}{\{\alpha(X) + \beta(X) + t - 1\}^2} h_a(X)g_b(X) \\
    \partial_b&  \partial_a   \log P_{\theta}(T =t  \mid X=x)(h_b, g_a)   =   \partial_b \partial_a  \log P_{\theta}(T =t - 1  \mid X=x)(h_b, g_a)\\
    & \quad \hspace{7cm} +  \frac{\alpha(X)\beta(X)}{\{\alpha(X) + \beta(X) + t - 1\}^2} h_b(X)g_a(X)
\end{align*}

Similarly, we can show the second derivatives of the log survival probability are equal to those of the log event probabilities:
\begin{align*}
    \partial_a^2 \log P_{\theta}(T > 1 \mid X=x)(h_a, g_a)  &=   -\frac{\alpha(X)\beta(X)}{\{\alpha(X) + \beta(X)\}^2} h_a(X) g_a(X); \\
    \partial_a \partial_b \log P_{\theta}(T >1 \mid X=x)(h_a, g_b)  &=   \frac{\alpha(X)\beta(X)}{\{\alpha(X) + \beta(X)\}^2} h_a(X) g_a(X);\\
   \partial_b     \partial_a \log P_{\theta}(T > 1 \mid X=x)(h_b, g_a)  &=   \frac{\alpha(X)\beta(X)}{\{\alpha(X) + \beta(X)\}^2} h_b(X) g_a(X) ;\\
 \partial_b^2 \log P_{\theta}(T >1 \mid X=x)(h_b, g_b)  &=   -\frac{\alpha(X)\beta(X)}{\{\alpha(X) + \beta(X)\}^2} h_b(X) g_b(X).
\end{align*} 
Moreover, we can show that
\begin{align*}
    \partial_a^2 \log P_{\theta}(T  > t  \mid X=x)(h_a, g_a) & = \partial_a^2 \log P_{\theta}(T =t  \mid X=x)(h_a, g_a) \\
      \partial_a \partial_b \log P_{\theta}(T  > t  \mid X=x)(h_a, g_b) & = \partial_a \partial_b  \log P_{\theta}(T =t  \mid X=x)(h_a, g_b) \\
      \partial_b \partial_a \log P_{\theta}(T  > t  \mid X=x)(h_b, g_a) & = \partial_a \partial_b  \log P_{\theta}(T =t  \mid X=x)(h_b, g_a).
\end{align*}
The remaining Hessian component $ \partial_b^2  \log P_{\theta}(T  > t  \mid X=x)(h_b, g_b)$ can be found in \cite{hubbard2021beta}.

\subsection{Additional details on synthetic experiment}
\label{appendix::simdetails}

The data structure \( Z = (X, A, \widetilde{T}, \Delta) \) is generated as follows. The covariate vector \( X \) is drawn from a uniform distribution on \([-1, 1]^3\). Given \( X \), the treatment assignment \( A \) is sampled from a Bernoulli distribution with success probability 
\[
\pi_0(X) := \text{expit}(2\sqrt{3}(X_1 + X_2 + X_3)).
\]
The uncensored survival time $T$ is drawn conditional on $(A,X)$ from a Beta-geometric event time distribution with shape parameters 
\begin{align*}
\alpha_0(A,X) &:= \exp(-0.1 + \sqrt{3}(X_1 + X_2 + X_3) + 0.25A)\\
\beta_0(A,X) &:= \exp(\sqrt{3}(X_1 + X_2 + X_3) + 0.1).
\end{align*}
We impose administrative censoring at time \( c = 6 \), with the observed event time $\widetilde{T} = \min(T, c)$ and censoring indicator $\Delta = 1(T \leq c)$.

\section{Additional experiments for R-learner loss}

\subsection{Inference on ATE using R-learner}
\label{sec::exp::cate}

In this experiment, we evaluate various autoDML estimators for the ATE derived from the R-learner loss, a specific Neyman-orthogonal loss for the CATE introduced by \cite{QuasiOracleWager}. One advantage of developing autoDML estimators based on a loss function that directly targets the CATE is that it provides a straightforward framework for constructing semiparametric estimators with model assumptions on the CATE, such as partially linear regression models, by selecting an appropriate Hilbert space $\mathcal{H}$. In this experiment, we impose an additive model assumption on the CATE by defining $\mathcal{H}$ as the space of all functions that are additive in the covariates.

We consider the data structure \( Z = (X, A, Y) \), where \( X \in \mathbb{R}^d \) represents covariates, \( A \in \{0,1\} \) is a binary treatment, and \( Y \in \mathbb{R} \) is the outcome. The R-learner loss function is a weighted pseudo-regression loss, given in Example \ref{example::orthoLS}. The M-estimand \( \theta_P \) is the CATE function, \( x \mapsto E_P[Y \mid A = 1, X = x] - E_P[Y \mid A = 0, X = x] \), with the functional of interest \( \psi_P: \theta \mapsto E_P[\theta(1,X) - \theta(0,X)] \). The covariate vector \( X = (X_1, X_2, X_3) \) is drawn independently as \( X_1, X_2, X_3 \sim \text{Uniform}(-1, 1) \). Given \( (X_1, X_2, X_3)=(x_1,x_2,x_3) \), treatment \( A \) follows a Bernoulli distribution with success probability \( \pi_0(x_1,x_2,x_3) := \text{expit}(\sin(2x_1) + 2x_2 + |x_3|) \). We define the auxiliary probability \( \pi_{0,\text{add}}(x_2) := \text{expit}(x_2) \) to construct a local perturbation, where the weight is \( \omega(x_2) = \pi_{0,\text{add}}(x_2)(1 - \pi_{0,\text{add}}(x_2)) \) and the local fluctuation term is \( \text{local\_fluct} = (A - \pi_{0,\text{add}}(X_2)) / \omega(X_2) \). The outcome \( Y \), given \( A = a \) and \( (X_1, X_2, X_3)=(x_1,x_2,x_3) \), follows a Normal distribution with mean \( \mu_0(a,x_1,x_2,x_3) = \text{local\_fluct} / \sqrt{n} + \sin(2x_1) + |x_3| + a (0.3 + 2x_3^2 + \sin(2x_1)) \) and standard deviation \( \sigma_0(x_1,x_3) = 0.5 + (x_1 + x_3)/8 \). Irregular estimators exhibit nonvanishing asymptotc bias under sampling of local perturbation to $P_0$ of distance $O(n^{-1/2})$.

We evaluate four autoDML estimators based on the R-learner loss, each employing an additive model $\mathcal{H}$ for the CATE. For all estimators, the nuisance functions for the R-learner loss, specifically the propensity score $\pi_0$ and the conditional mean of the outcome given covariates $E_0[Y \mid X = \cdot]$, are estimated using an ensemble of multivariate adaptive regression splines \citep{friedman1991multivariate} and a generalized additive spline model \citep{GAMhastie1987generalized}, respectively implemented in the \texttt{R} packages \texttt{earth} \citep{milborrow2019package} and \texttt{mgcv} \citep{wood2001mgcv}. The first estimator is the one-step autoDML estimator, $\widehat{\psi}_n^{\mathrm{dml}}$, as described in Example \ref{examp3est} of Section \ref{sec::autodmlest} and implemented according to Algorithm \ref{alg:autoDML}. The second is the autoTML estimator, $\widehat{\psi}_n^{\mathrm{tmle}}$, which is also based on the R-learner loss and implemented according to Algorithm \ref{alg:autoTML}. For both the autoDML and autoTML estimators, the CATE, $\theta_0$, is learned using the R-learner loss with the highly adaptive lasso (HAL) spline estimator \citep{vanderlaanGenerlaTMLEFIRST, HAL2016}, while the Riesz representer, $\alpha_0$, is learned using the Riesz loss with the HAL estimator. The HAL estimators are implemented using the \texttt{hal9001} \citep{hal2} package, with the additive model constraint enforced by specifying the argument \texttt{max\_degree = 1}. The estimators of the nuisances in the loss, the CATE, and the Riesz representer are cross-fitted using Algorithm \ref{alg:crossfit} with 10 folds.  The third estimator, $\widehat{\psi}_n^{\mathrm{sieve}}$, is an autoSieve estimator that leverages a data-adaptive sieve, $\{\mathcal{H}_{n, \lambda_k}: k \in \mathbb{N}\}$, learned using the additive HAL estimator, trained without cross-fitting on all available data. Here, $\lambda_k \in \mathbb{R}$ represents the lasso regularization parameter, and $\mathcal{H}_{n, \lambda_k}$ denotes the linear span of basis functions with nonzero coefficients from the additive HAL estimator fit with regularization parameter $\lambda_k$. The regularization parameter $\lambda_{k(n)}$ used to compute $\widehat{\psi}_n^{\mathrm{sieve}}$ is selected data-adaptively using the automatic undersmoothing method detailed in Section \ref{sec::autodmlest}. The final estimator is an adaptive sieve plug-in estimator for the ATE, where the sieve is again learned using HAL but without undersmoothing, resulting in superefficient behavior. This estimator is a specific instance of the adaptive DML estimators introduced in \cite{van2023adaptive} and is also a special case of the estimators proposed in Section \ref{sec::adml}, corresponding to a HAL-based sieve estimator fit with the regularization parameter $\lambda_{k(n)}$ selected through cross-validation using the R-learner loss. As a baseline for comparison, we also evaluate the nonparametric AIPW estimator of the ATE, which is constructed using estimators of the propensity score and the outcome regression.

\begin{figure}
    \centering
    \includegraphics[width=0.5\linewidth]{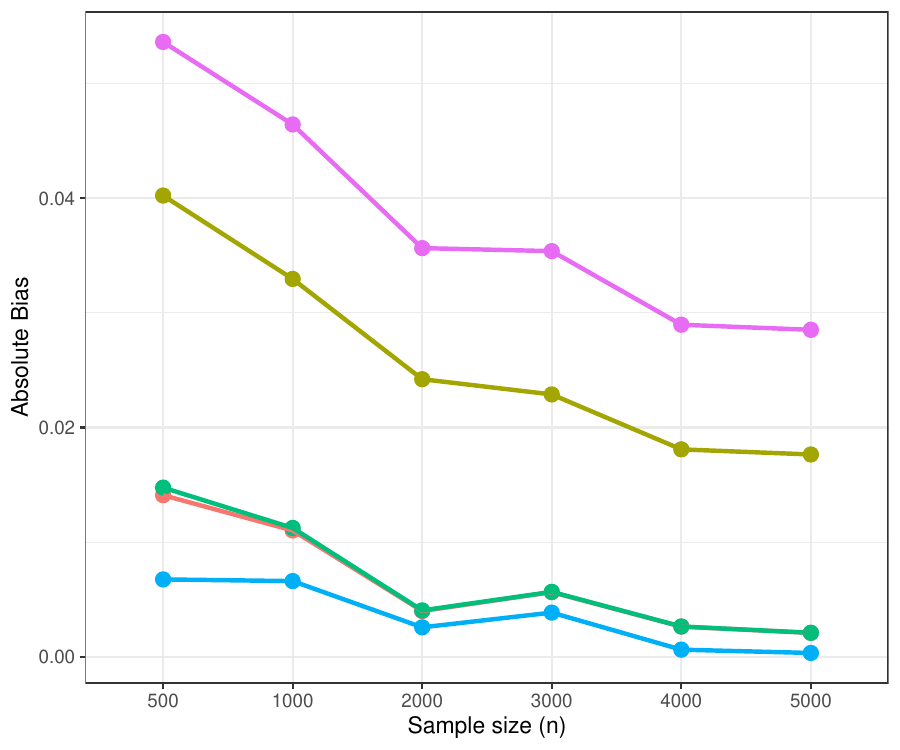}\includegraphics[width=0.5\linewidth]{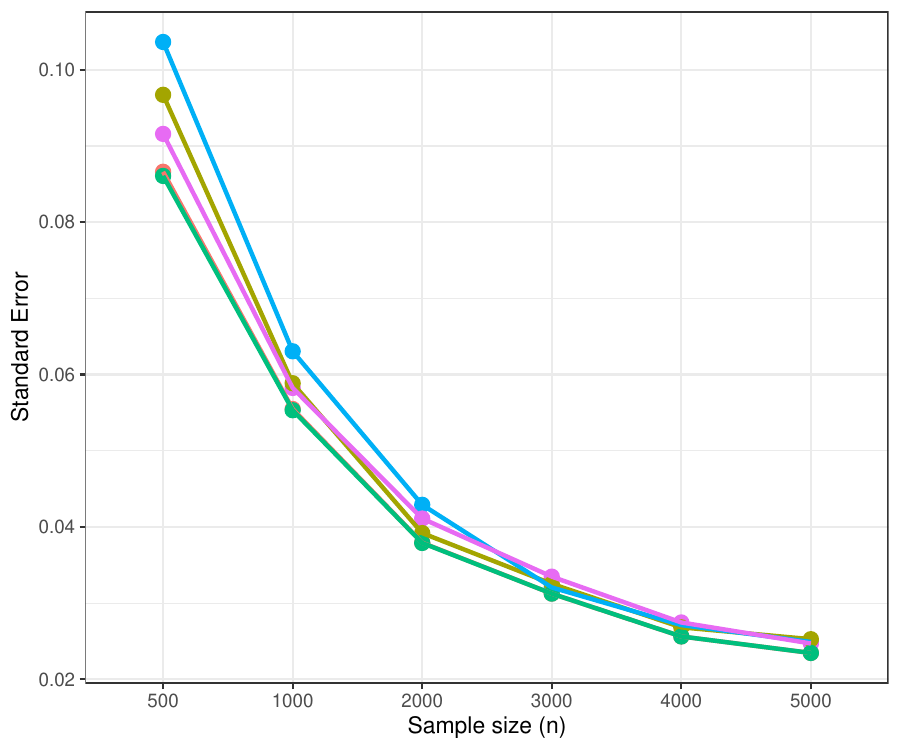}

    \includegraphics[width=0.5\linewidth]{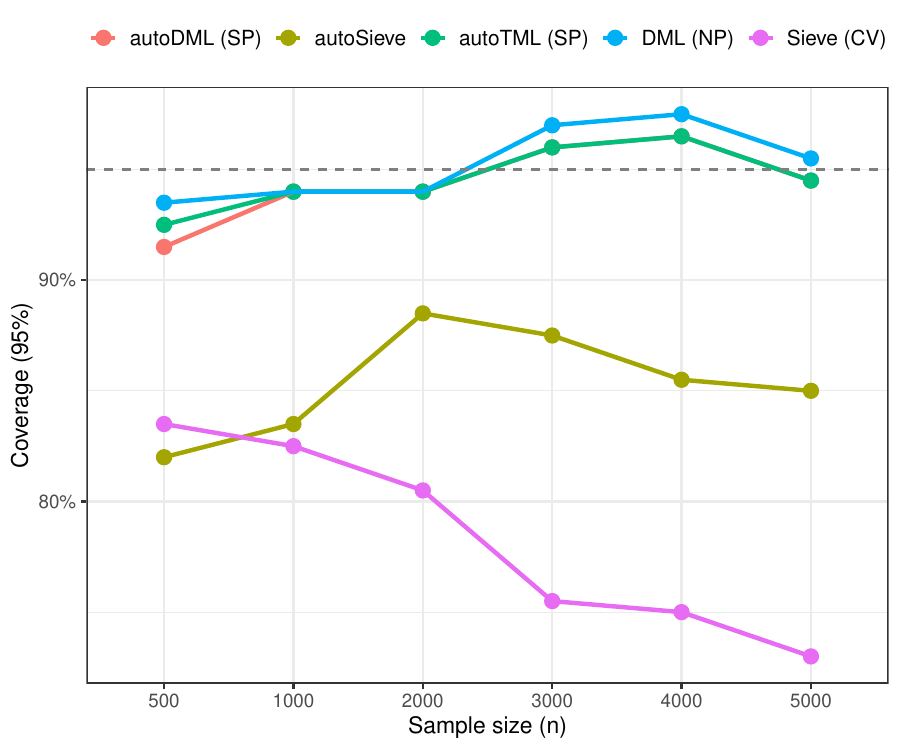}
    \caption{Bias, standard error, and coverage probability for the ATE, comparing various methods across sample sizes.}
    \label{fig:cate-simulation-results}
\end{figure}

The results are displayed in Figure \ref{fig:cate-simulation-results}. In terms of bias and coverage, autoTML, autoDML, and DML perform best among the methods evaluated. autoTML and autoDML exhibit less variability compared to DML because they perform debiasing within a semiparametric model and are thus less sensitive to issues with treatment overlap. Additionally, autoDML and autoTML demonstrate similar performance across all sample sizes, likely due to the quadratic nature of the loss and the linearity of the ATE functional, which aligns the two approaches closely. Conversely, the cross-validated and auto-undersmoothed plug-in HAL estimators exhibit the poorest performance regarding bias and coverage. The cross-validated estimator shows significant bias and inadequate coverage under sampling from local perturbations, primarily due to its irregular nature. Although undersmoothing partially reduces this bias, it remains substantial enough to negatively affect coverage.

\end{document}